%% file: main.tex
\documentclass{siamart171218}
\usepackage[margin=1in]{geometry}
\usepackage{amsmath, amssymb}
\usepackage{bm}
\usepackage{xcolor}
\usepackage{subcaption}
\usepackage[textsize=tiny]{todonotes}
\usepackage[british]{babel}

\usepackage{algorithm}
\usepackage{algorithmic}

\include{macros_no_juq}

\newcommand{\eg}{\textit{e.g.,} }
\newcommand{\ie}{\textit{i.e.,} }

\title{A Generalized Approximate Control Variate Framework for Multifidelity Uncertainty Quantification}
\date{}

\author{Alex A. Gorodetsky\footnote{T\lowercase{he first two authors contributed equally to the concepts in this paper}}~\thanks{Department of Aerospace Engineering, University of Michigan, Ann Arbor, MI (\email{goroda@umich.edu})}
  \and Gianluca Geraci\footnotemark[1]~\thanks{Optimization and Uncertainty Quantification Department, Sandia National Laboratories, Albuquerque, NM, 87185-1318 (\email{{ggeraci,mseldre,jdjakem}@sandia.gov})}
  \and Mike Eldred\footnotemark[3]
  \and John D. Jakeman\footnotemark[3]}

\begin{document}
\maketitle
\begin{abstract}
  We describe and analyze a variance reduction approach for Monte Carlo (MC) sampling that accelerates the estimation of statistics of computationally expensive simulation models using an ensemble of models with lower cost. These lower cost models --- which are typically lower fidelity with unknown statistics --- are used to reduce the variance in statistical estimators relative to a MC estimator with equivalent cost. We derive the conditions under which our proposed approximate control variate framework recovers existing multi-model variance reduction schemes as special cases. We demonstrate that these existing strategies use recursive sampling strategies, and as a result, their maximum possible variance reduction is limited to that of a control variate algorithm that uses only a single low-fidelity model with known mean. This theoretical result holds regardless of the number of low-fidelity models and/or samples used to build the estimator. We then derive new sampling strategies within our framework  that circumvent this limitation to make efficient use of all available information sources. In particular, we demonstrate that a significant gap can exist, of orders of magnitude in some cases, between the variance reduction achievable by using a single low-fidelity model and our non-recursive approach. We also present initial sample allocation approaches for exploiting this gap. They yield the greatest benefit when augmenting the high-fidelity model evaluations is impractical because, for instance, they arise from a legacy database. Several analytic examples and an example with a hyperbolic PDE describing elastic wave propagation in heterogeneous media are used to illustrate the main features of the methodology.  

\end{abstract}

\section{Introduction}

Numerical evaluation of integrals is a foundational aspect of mathematics that has impact on diverse areas such as finance, uncertainty quantification, stochastic programming, and many others. Monte Carlo (MC) sampling is arguably the most robust means of estimating such integrals and can be easily applied to arbitrary integration domains and measures. The MC estimate of an integral is unbiased, and its rate of convergence is independent of the number of variables and the smoothness of the integrand.

Nevertheless, obtaining a moderately accurate estimate of an integral with MC is computationally intractable for integrands that are expensive to evaluate, \eg those arising from a high-fidelity simulation.
This intractability arises because the variance of a MC estimator is proportional to the ratio of the variance of the integrand and inversely proportional to the number of samples used. 
As such, techniques that retain the benefits of MC estimation while reducing its variance are important for extending the applicability of these sampling-based approaches.

Control variates (CV) are a class of such techniques that have a long history of reducing MC variance by introducing additional estimators that are correlated with the MC estimator~\cite{Lavenberg1978,Lavenberg1981,Lavenberg1982,Hesterberg1996}. The use of CV methods has recently seen a resurgence for uncertainty quantification (UQ) problems where the integrands are computationally expensive to evaluate. In these cases, CV approaches can leverage multiple 
simulation models to accelerate the convergence of statistics for both forward~\cite{Ng2013,Goh2013,Peherstorfer2016,Doostan2016,Geraci2017} and inverse~\cite{Baker2018} UQ. 
These additional simulation models arise from either 
different sets of equations (\ie the {\em multifidelity} case of differing model forms)
and/or from varying temporal and spatial discretizations (\ie the {\em multilevel} case of differing numerical resolutions for the same set of equations).  The model ensemble could include reduced-order models~\cite{Roderick2014}, dimension-reduction or surrogate models~\cite{Padron2014} (\eg active subspace approximations), and even data from physical experiments~\cite{Kuya2011}. 
Multiple conceptual dimensions can exist within a modeling hierarchy, leading to multi-index constructions~\cite{Haji2016} in the case of independent resolution controls. Finally, both multi-physics and multi-scale simulations can contribute additional combinatorial richness to the associated modeling ensemble.

Traditional CV methods~\cite{Lavenberg1978} require explicit knowledge of the statistics (for instance the expected value) of their approximate information sources. However, these estimates are frequently unavailable \textit{a priori} in the UQ simulation-based context. Consequently, CV methods must be modified to balance the computational cost of evaluating lower fidelity models and the reduction in error that they each provide.
There exist several strategies that explicitely pursue the goal of estimating the unknown expected values~\cite{Schmeiser2001,Emserman2002,Pasupathy2012,Ng2014} within a control variate framework; however the analysis of these approaches is limited to the case of a single control variate only. As a result, they do not consider how ensembles of low-fidelity information sources could be effectively used to improve variance reduction. Once ensembles are considered two additional questions arise: (1) how are the relationships among the low-fidelity and the high-fidelities simulations represented and (2) how are resources distributed among each model in the ensemble? A large number of algorithms have arisen to address these issues, each with their own assumptions. For instance, when the low-fidelity models arise from a hierarchy of discretization levels, multilevel Monte Carlo (MLMC) and multi-index Monte Carlo (MIMC) approaches have been developed~\cite{Giles2008,Haji2016}. These algorithms use a recursive difference estimator~\cite{Owen2013} within a discretization adaptation scheme to minimize the mean squared error of a targeted statistic. When the variance of the recursive differences decays across levels, significant variance and MSE reduction can be achieved. More general relationships between low-fidelity simulation models have led to the development of multifidelity Monte Carlo~\cite{Peherstorfer2016b} and hybrid MLMC-control variate schemes~\cite{Nobile2015,Fairbanks2017,Geraci2015,Geraci2017}. However, all of these algorithms are limited in how they model the relationships among low-fidelity information sources and are reliant on recursive strategies for sample/resource allocation.

In this paper, we present a framework that employs CV techniques to generate a minimum variance estimator for a given computational budget; this framework is generalized in that no simplified model dependency structure (e.g., hierarchy) is assumed. 

The effectiveness of CV approaches for reducing estimator variance is related to the correlation between the high and low-fidelity models and the relative computational expense of simulating the lower-fidelity models. Two specific cases of interest include: {\em (1)} the number of high fidelity simulations is fixed 
and we can only allocate samples among the low-fidelity models, and {\em (2)} we have the freedom to allocate samples among all models.  Note that the first case can arise in realistic applications 
when the high-fidelity evaluations are obtained from legacy data, correspond to a particular experimental or reference simulation campaign, or for which additional runs are impractical or no longer possible for various reasons.
This work extends theory for multiple control variates without an assumed dependency structure to the case of unknown control-variate means.  The primary contributions include:
\begin{enumerate}
\item New theoretical results demonstrating that recursive sampling allocation strategies limit the maximum attainable variance reduction;
\item A new framework for approximate control variates algorithms --- those that use control variates with unknown statistics --- that guarantees convergence to that of the traditional control variate; and  
\item A sample allocation strategy for optimal variance reduction for a fixed computational budget.
\end{enumerate}

The recursive algorithms alluded to in the first contribution correspond to algorithms such as the multifidelity Monte Carlo (MFMC)~\cite{Peherstorfer2016b} and the recursive difference estimator~\cite{Owen2013}. Both of these recursive algorithms are discussed in detail in the following sections. We will show that our new framework is able to achieve improved variance reduction compared to these existing approaches. We prove this benefit theoretically, and we also demonstrate it numerically using a preliminary optimization-based sample allocation strategy. In other words, the primary focus in this paper is on estimator %
design, and an initial sketch of possible optimization strategies for sample allocation is provided.

The theoretical contributions are justified by three main theorems, whose implications are summarized below
\begin{itemize}
\item Theorem~\ref{th:MFMC_varbound}: the variance reduction of MFMC is limited to that of a CV estimator that only uses a single low-fidelity model.
\item Theorem~\ref{th:MLMC_varbound}: the variance reduction of the recursive difference estimator is limited to that of a CV estimator that only uses a single low-fidelity model.
\item Theorem~\ref{th:ACV}: an approximate control variate scheme can be devised to converge to the optimal linear control variate that uses multiple low-fidelity models.
\end{itemize}

Finally, we note that this work considers \textit{variance reduction} of a Monte Carlo estimator, as distinguished from \textit{mean-squared-error reduction} with respect to some unknown ``truth'' model. That is, we consider the highest-fidelity model to be most accurate and our aim is to construct an unbiased and reduced-variance estimate of the statistics of this model using lower-fidelity models. 
  This differs from %
multilevel~\cite{Giles2008,Giles2015} and multi-index~\cite{Haji2016} Monte Carlo which can target both variance and bias reduction, although %
in practice this requires a case where the bias of the high-fidelity model is both significant and under the user's control (e.g., it is not already at the boundary of what is practical to simulate repeatedly on a high-performance computer).  If both variance and bias control are desired/practical,
  our proposed approaches can be used within the inner loop step of MLMC-type schemes that combine model adaptation for bias reduction (e.g., grid refinement for the high-fidelity) with variance reduction approaches. In the inner loop of such schemes, see e.g.~\cite{Cliffe2011}, a minimum variance multi-model estimator is required, and we envision the incorporation of our proposed approaches within such an adaptive context in the future.

The remainder of the paper is structured as follows: in \S\ref{sec:problem}, we describe a unifying framework for multi-model control variates approaches and then identify existing recursive estimators within this general formalism; in \S\ref{sec:acv}, we develop a set of new approximate control variates estimators; and in \S\ref{sec:numexamples}, we demonstrate our approaches for several numerical test cases.

\section{A unifying approximate control variate framework}\label{sec:problem}
In this section we provide background on CV methods and present a framework for estimating integrals using CV schemes with unknown mean values. We then show how existing recursive methodologies 
are 
special cases of this framework.

Let $(\varspace, \mathcal{F}, P)$ be a probability space and 
Let 
$\qoif:\reals^d \to \qspace \subset \reals$ denote a mapping from a vector of inputs to a scalar-valued output. This mapping is refered to as the ``high-fidelity model'' because it represents the stochastic process whose statistics we desire to estimate. Our goal is estimate \(\mean{}=\mathbb{E}\left[\qoif\right]\) 
using a set of samples $\sset{} = (\sset{}^{(1)},\ldots,\sset{}^{(\nhf)})$ of the input random variables. The MC estimate of the mean
\begin{equation}
\est{}(\sset{})=\nhf^{-1}\sum_{i=1}^\nhf \qoi\left(\sset{}^{(i)}\right)
  \label{eq:mc_mean_estimator}
\end{equation}
is unbiased, that is $\mathbb{E}[\est{}]=\mathbb{E}[Q]$, and converges a.s. For $Q$ with bounded variance, the Central Limit Theorem implies that the error in the estimate becomes normally distributed with variance $\nhf^{-1}\varF{Q}$, as $\nhf\rightarrow\infty$, where $\varF{Q}$ denotes the variance of $Q$.

\subsection{Traditional control variate estimation}
The MC-based\footnote{Control-variate-based variance reduction algorithms can certainly be performed with other sampling-based and non-sampling-based estimators. In this paper we focus on the common case of Monte Carlo estimators.} linear\footnote{In the rest of the paper, we drop the term ``linear'' as all schemes that are studied will be linear. For reference, nonlinear control variates are those that may have polynomial corrections of the form $\cvest = \est{} + \sum_{i=1}^{p} \alpha_i \left(\est{1} - \mean{1}\right)^{i}$.} CV algorithm seeks a new estimator $\cvest$ that has smaller variance than $\est{}$ while still requiring only $\nhf$ evaluations of $\qoif$. The algorithm introduces an additional random variable \(\cv{1}\) with known mean \(\mean{1}\). Then it requires computing an estimator $\est{1}$ of $\mean{1}$ using the same samples that were used for $\est{}$. Finally, these quantities are assembled into the following estimator
\begin{equation}
  \cvest{}(\cvw, \sset{}) = \est{}(\sset{}) + \alpha \left( \est{1}(\sset{}) - \mean{1} \right),
  \label{eq:linear_cv}
\end{equation}
for some scalar \(\alpha\). As we will see shortly, the variance of $\cvest$ is strongly influenced by the correlation between $\qoi$ and $\cv{1}$, and greater variance reduction is achieved for higher correlation. 
To this end we will call $\est{1}$ the correlated mean estimator (CME), and we will refer to $\mean{1}$ as the control variate mean (CVM). This approach can be extended to $\nmodels$ information sources $\left(\qoif_i: \reals^d \to \qspace_i \subset \reals\right)_{i=1}^{\nmodels}$, see e.g., \cite{Lavenberg1978,Lavenberg1981,Lavenberg1982}, using
\begin{equation}\label{eq:lcv}
  \cvest{}(\vec{\cvw}, \sset{}) = \est{}(\sset{}) + \sum_{i=1}^{\nmodels}\cvw_i\left(\est{i}(\sset{}) - \mean{i}\right),
\end{equation}
where the estimator now uses a vector of CV weights $\vec{\cvw} \equiv (\cvw_1,\ldots,\cvw_\nmodels)$.  CV methods can also estimate multiple quantities of interest, e.g. \cite{Rubinstein1985,Venkatraman1986}; for simplicity, here we only consider a scalar $\qoif$.

One of the measures of the effectiveness of a CV scheme is measured 
by the \textit{variance reduction ratio}:
\begin{equation} \label{eq:varrat}
  \gamma^{\cvl}(\vec{\cvw}) \equiv \frac{\varF{\cvest(\vec{\cvw},\sset{})}}{\varF{\est{}(\sset{})}}. 
\end{equation}
The optimal  CV (OCV) uses weights that minimize this variance:
\begin{equation}\label{eq:opt}
  \optcvw = \argmin_{\vec{\cvw}} \gamma^{\cvl}(\vec{\cvw}) = \argmin_{\vec{\cvw}} \varF{\cvest (\vec{\cvw}, \sset{})}.
\end{equation}
Following~\cite{Lavenberg1981}, let $\covm \in \reals^{\nmodels \times \nmodels}$ denote the covariance matrix among $\qoi_i$ and $\covv \in \reals^{\nmodels}$ denote the vector of covariances between $\qoi$ and each $\qoi_i$. The solution to Equation~\eqref{eq:opt} is $\optcvw = -\covm^{-1}\covv$. If we further define \[\covvnorm = \covv/\StDevF{\qoi} = \left[\ccoeff{1} \StDevF{\cv{1}}, \dots,  \ccoeff{\nmodels} \StDevF{\cv{\nmodels}} \right]^{\mathrm{T}},\] where $\ccoeff{i}$ is the Pearson correlation coefficient between $\qoi$ and $\cv{i}$, then the variance reduction becomes
\begin{equation}\label{eq:optcv_cov}
  \gamma^{\cvl}(\optcvw) = 1 - R_{\cvl}^2, \quad 0 \leq R_{\cvl}^2 \leq 1,
\end{equation}
where $R_{\cvl}^2 = \covvnorm^T \covm^{-1}\covvnorm$.  No variance reduction with respect to MC is achieved when $R_{\cvl}^2 = 0$, and maximum reduction occurs when $R_{\cvl}^2 = 1$. For a single low-fidelity information source, this quantity simplifies to $R_{\cvl-1}^2 = \ccoeff{1}^2$.

\subsection{Approximate control variate estimation}\label{sec:approximate_cv_est}
Traditional control variate estimation assumes that the means $\mean{i}$ are known. In our motivating problem of multifidelity uncertainty quantification, these means are not known, but rather must be estimated from lower fidelity simulations. 
As in~\eqref{eq:mc_mean_estimator}, $\sset{}$ denotes a set of samples used to evaluate the high-fidelity function $\qoi$. Now let $\sset{i}$ denote a 
set of $\nhf_i$ samples partitioned into two ordered subsets  $\sset{i}^1 \subset \sset{i}$ and $\sset{i}^2 \subset \sset{i}$ such that $\sset{i}^1 \cup \sset{i}^2 = \sset{i}$. Note that $\sset{i}^1$ and $\sset{i}^2$ are not required to be disjoint, \ie they may overlap such that $\sset{i}^1 \cap \sset{i}^2 \neq \emptyset$. We will construct and analyze approximate control variate (ACV) estimators of the following form
\begin{align}
  \acvest(\vec{\cvw},\vec{\sset{}}) &= \est{}(\sset{}) + \sum_{i=1}^{\nmodels}\cvw_i\left(\est{i}\left(\sset{i}^{1}\right) - \estm{i}\left(\sset{i}^{2}\right)\right)  
  = \est{}(\sset{}) + \sum_{i=1}^{\nmodels} \cvw_i \cvdiff{i}(\sset{i})
  = \est{} + \vec{\cvw}^T \vec{\cvdiff{}}, \label{eq:acv}
\end{align}
where $\vec{\cvdiff{}} = (\cvdiff{1},\ldots,\cvdiff{\nmodels})$, $\cvdiff{i}(\sset{i}) = \est{i}\left(\sset{i}^{1}\right) - \estm{i}\left(\sset{i}^{2}\right)$ and we use $\vec{\sset{}} \equiv (\sset{},\sset{1},\ldots,\sset{\nmodels})$ to denote the input values to the model set: $\sset{}$ for the high-fidelity model and $\sset{i}$ for each of the low-fidelity models, $i=1,\ldots,\nmodels$. The estimate $\estm{i}$ of the CVM $\mu_i$ will be called the estimated control variate mean (ECVM). \textit{Many of the estimators in the literature, and those that we derive in this paper, are differentiated by how the samples $\sset{i}$ are related among the random variables $\cvdiff{i}$ for $i=1,\dots,\nmodels$ and how the subsets $\sset{i}^1$ and $\sset{i}^2$ are defined 
within each random variable $\cvdiff{i}$.} It will be useful to consider ordered sets of samples such that $\sset{i} = [\sset{i}^1, \sset{i}^2] = [\sset{i}^{1(1)},\ldots,\sset{i}^{1(L_1)}, \sset{i}^{2(1)},\ldots, \sset{i}^{2(L_2)}]$ 
for some $L_1,L_2 >0$

Because we now employ extra samples to estimate $\mu_i$, we must account for the additional cost this imposes. Let $\nhf$ denote the number of %
realizations of $Q$ and $$\nhf_{i} = \lceil \rat_{i} \nhf \rceil$$ denote the 
number of %
evaluations of $\cv{i}$ for scaling factor $\rat_{i} \in \reals_{+}$. Let $\cost_i<1$ denote the ratio between the cost of a single realization of $\cv{i}$ and the cost of obtaining a realization of the high-fidelity $Q$. Then the cost of the ACV estimator in equivalent high-fidelity evaluations is
\begin{align}
w_\mathrm{ACV} = \nhf + \sum_{i=1}^\nmodels \nhf_i \cost_i.
\end{align}

For fixed $\vec{\cvw}$, the estimator $\acvest$ is unbiased when each component, $\est{}$, $\est{i}$, and $\estm{i}$ for $i=1,\ldots,\nmodels$ is unbiased. The variance of $\acvest$ is given in the following proposition. 
\begin{proposition}[Variance of the ACV estimator]\label{prop:acv_var}
   The variance of the ACV estimator~\eqref{eq:acv} for fixed $\vec{\cvw}$ is
  \begin{equation} \label{eq:acv_var}
    \varF{\acvest} = \varF{\est{}}\left( 1 + \vec{\cvw}^T \frac{\covdiff}{\varF{\est{}}} \vec{\cvw} + 2 \vec{\cvw}^T \frac{\covhc}{\varF{\est{}}}\right)  \textrm{ where } \vec{\cvdiff{}} = (\cvdiff{1},\ldots,\cvdiff{\nmodels}).
  \end{equation}
\end{proposition}
\begin{proof}
  This result follows from the variance of MC estimators, the variance of scaled random variables, and the variance of sums of random variables. We have
  \begin{align*}
    \varF{\acvest(\vec{\cvw},\vec{\sset{}})} &= \var \left[\est{} + \vec{\cvw}^T \vec{\cvdiff{}} \right]
    = \varF{\est{}} + \vec{\cvw}^T \cov[\vec{\cvdiff{}}, \vec{\cvdiff{}}] \vec{\cvw} + 2 \vec{\cvw}^T \covhc \\
    &= \var[\est{}]\left( 1 + \vec{\cvw}^T \frac{\cov[\vec{\cvdiff{}}, \vec{\cvdiff{}}]}{\var[\est{}]} \vec{\cvw} + 2 \vec{\cvw}^T \frac{\cov[\vec{\cvdiff{}}, \est{}]}{\var[\est{}]}\right),
  \end{align*}
  where the middle term on the first line is the definition, the last term follows from the variance of sums of random variables, and the second line factors out the baseline estimator variance.
\end{proof}

Therefore, the variance ratio of the approximate control variate (ACV) estimator is
\begin{equation}
  \gamma^{\acvl}(\vec{\cvw}) = 1 + \vec{\cvw}^T \frac{\covdiff}{\varF{\est{}}} ~\vec{\cvw} + 2 \vec{\cvw}^T \frac{\covhc}{\varF{\est{}}}.
\end{equation}
The optimal approximate CV estimator consists of the $\vec{\cvw}$ which minimizes this ratio. The properties of this optimal estimator are given in the following proposition.
\begin{proposition}[Optimal ACV]\label{prop:ocv}
  Assume $\covF{\cvdiff{},\cvdiff{}}$ is positive definite and $\varF{\est{}}>0$. The ACV weight that provides the greatest variance reduction for the approximate CV Equation~\eqref{eq:acv} is given by
  \begin{equation}\label{eq:opt_acv_weights}
      \optacvw = \argmin_{\vec{\cvw}} \gamma^{\acvl}(\vec{\cvw}) =  -\covdiff^{-1}\covhc
  \end{equation}
  with corresponding estimator variance
  \begin{equation} \label{eq:ocv_var}
      \varF{\acvest(\optacvw,\vec{\sset{}})} = \varF{\est{}}\left( 1 - R_{\acvl}^2\right) \textrm{ where } R_{\acvl}^2 = \covhc^T \frac{\covdiff^{-1}}{\varF{\est{}}}\covhc.
  \end{equation} 
\end{proposition}
\begin{proof}
  The variance reduction $\gamma^{\acvl}$ is quadratic in $\vec{\cvw}$, is non-negative, 
and has an extremum found by setting the gradient to zero
  \begin{align*}
    0 &=  \frac{\covdiff}{\var[\est{}]} \optacvw + \frac{\covhc}{\varF{\est{}}}  
    \implies \optacvw = -\covdiff^{-1}\covhc.
  \end{align*}
  If we substitute this weight into Equation~\eqref{eq:acv_var}, we obtain the stated result.
\end{proof}

Next, we use these results to derive existing estimators that use different sampling strategies for  computing $\covdiff$ and $\covhc$. These estimators can be broadly grouped as recursive nested  (in the case of MFMC \cite{Peherstorfer2016}) and recursive difference (in the case of MLMC-like \cite{Giles2008,Giles2015} sampling allocations).

\subsection{Recursive nested estimators}\label{sec:mfmc}
In this section, we describe a sampling strategy that has a recursive and nested structure. This sampling structure requires that both the CME and ECVM use the same samples as that of all higher fidelity models, and the ECVM augments this set with additional samples. The most prominant example of such a sampling strategy is that of the multifidelity Monte Carlo (MFMC) estimator~\cite{Peherstorfer2016b,Ng2013,Pasupathy2012,Peherstorfer2016}.
We begin by deriving a new variance reduction expression for MFMC, an alternative to~\cite[Lemma 3.3]{Peherstorfer2016b}, that relates the correlation of the CVs to the variance reduction. Then we use this expression to demonstrate that MFMC is a recursive estimator. Finally, due to this recursive nature, we show that the maximum variance reduction provided by this estimator is limited to that provided by the single optimal CV, regardless of how many samples are used to evaluate $\cv{i}$.

The MFMC estimator can be obtained from the following recursive procedure. Samples $\sset{}$ are used for $\est{}$ and for the CME ($\est{1}$) of $\cv{1}$. Then, an enriched set of samples $\sset{1}$ are used for the ECVM ($\estm{1}$) of $\cv{1}$ to obtain 
\begin{equation}
 \acvest\left(\cvw,\sset{},\sset{1}\right) = \est{}(\sset{}) + \cvw(\est{1}(\sset{}) - \estm{1}(\sset{1})),
\end{equation}
where, recalling our sample partition definitions in~\eqref{eq:acv}, we have chosen $\sset{1}^1 = \sset{}$ and $\sset{1}^2 = \sset{1}$, \ie $\sset{1}$ consists of the original $\sset{}$ samples along with an additional set. 

Now we introduce another CV to reduce the variance of the $\estm{1}$ estimate
\begin{align*}
  \acvest\left(\cvw,\tilde{\cvw},\sset{},\sset{1},\sset{2}\right)   &= \est{}(\sset{}) + \cvw(\est{1}(\sset{1}^1) - \left( \estm{1}(\sset{1}) + \tilde{\cvw}\left(\est{2}(\sset{1}) - \estm{2}(\sset{2})\right) \right) \\
  &= \est{}(\sset{}) + \cvw(\est{1}(\sset{1}^1) - \estm{1}(\sset{1})) - \cvw\tilde{\cvw}\left(\est{2}(\sset{1}) - \estm{2}(\sset{2})\right)
\end{align*}
where we have similarly assigned partitions $\sset{2}^1 = \sset{1}$ and $\sset{2}^2 = \sset{2}.$

After collapsing $\alpha$-products, this pattern produces the MFMC estimator
\begin{equation}\label{eq:MFMC_estimator}
  \begin{split}
  \mfmcest\left(\vec{\cvw},\vec{\sset{}} \right) &= \est{}(\sset{}) + \sum_{i=1}^{\nmodels}\cvw_i \cvdiff{i}(\sset{i}) \ \  \textrm{ where } 
  \cvdiff{i}(\sset{i}) = \est{i}\left(\sset{i}^{1}\right) - \estm{i}\left(\sset{i}\right),
  \end{split}
\end{equation}
and the sampling strategy partitions $\sset{i}$  into $\sset{i}^{1}$ and $\sset{i}^{2}$ according to
\begin{align*}
  \sset{i}^{1}=\sset{i-1} \quad \textrm{ and } \quad \sset{i}^{2}&=\sset{i}.
\end{align*}
for $i = 2, \ldots, \nmodels$ and $\sset{1}^1 = \sset{}$ and $\sset{1}^2 = \sset{1}$. The recursive nested sampling scheme for the MFMC estimator is illustrated in Figure~\ref{fig:mfmc_sampling_scheme}. 

The optimal weights are given in the following Lemma.
\begin{lemma}[Optimal CV weights for MFMC]\label{lem:mfmc_weight}
 The optimal weights for the MFMC estimator are 
 \begin{equation} \label{eq:mfmc_opt_alpha}
  \cvw_i^{\mfmc} = - \dfrac{\covF{\qoi, \cv{i}}}{\varF{\cv{i}}} \  \textrm{ for } \  i = 1,\ldots, \nmodels.
 \end{equation}
\end{lemma}
The proof is provided in~\ref{app:proof:mfmc_weight}, and provides an alternate derivation to the identical result in~\cite[Th. 2.4]{Peherstorfer2016b}. %
In the prior work, these weights were derived using assumptions on model costs and correlation orderings, and our result removes these assumptions. The next Lemma provides a new representation of the variance reduction of this estimator in terms of the correlations between model pairs. 
\begin{lemma}[Variance reduction by optimal MFMC]\label{lem:MFMC_var_reduction}
  Assume $|\ccoeff{1}|>0$ and let $\absr{0} = 1$. The variance reduction of the optimal MFMC estimator is
  \begin{equation}
   \VarF{\mfmccvw} = \frac{\VarF{Q}}{\nhf} \left( 1 - R^2_{\mfmc}\right) \textrm{ where }
      R^2_{\mfmc}
      = \rho_1^2 \left( \frac{\absr{1}-1}{\absr{1}} + \sum_{i=2}^{M} \frac{\absr{i}-\absr{i-1}}{\absr{i}\absr{i-1}} \frac{\rho_i^2}{\rho_1^2} \right).
        \label{eq:mfmc_R}
  \end{equation}
\end{lemma}
The proof is provided in~\ref{app:proof:mfmc_var_reduction}.
Under the assumption made in \cite{Peherstorfer2016b} that $|\ccoeff{1}| \geq |\ccoeff{i}|$ for $i = 2,\ldots,\nmodels$ (can be constructed by reordering without loss of generality), we can use Lemma~\ref{lem:MFMC_var_reduction} to show that the MFMC estimator cannot obtain greater variance reduction than the optimal single CV ($R_{OCV-1}^2=\rho_1^2$). 
\begin{theorem}[Maximum variance reduction of MFMC]\label{th:MFMC_varbound}
  The variance reduction of MFMC is bounded above by that of the optimal single CV, \ie,
  \begin{equation}
    R_{\mfmc}^2 \leq \ccoeff{1}^2.
  \end{equation}
\end{theorem}
\begin{proof}
  The proof follows from algebraic manipulation of \eqref{eq:mfmc_R}.
  \begin{equation*}
   R_{\mfmc}^2 = \ccoeff{1}^2 \left( \frac{\absr{1}-1}{\absr{1}} + \sum_{i=2}^{\nmodels} \frac{\absr{i}-\absr{i-1}}{\absr{i}\absr{i-1}} \frac{\rho_i^2}{\ccoeff{1}^2} \right) 
             \leq \ccoeff{1}^2 \left( 1 - \frac{1}{\absr{1}} + \sum_{i=2}^{\nmodels} \frac{1}{\absr{i-1}} - \frac{1}{\absr{i}} \right)             
   = \ccoeff{1}^2 \left( 1 - \frac{1}{\absr{\nmodels}} \right) < \ccoeff{1}^2,
  \end{equation*}
  where we note that the first inequality derives from the assumption $\lvert \ccoeff{i} \rvert \leq \lvert \ccoeff{1} \rvert$ for $i = 2, \dots, \nmodels$. 
\end{proof}
The variance reduction $R_{\mfmc}^2$ of MFMC is limited by $\ccoeff{1}^2$, irregardless of the amount of low-fidelity data used. Fundamentally, this limitation arises because of the recursive structure of the sampling strategy. In the following section, we show that this limitation also arises in another commonly used sampling approach.

\input{mlmc}

\subsection{Summary and example}\label{sec:motivate_example}

\begin{figure}
  \centering
  \begin{subfigure}[b]{0.23\textwidth}
    \centering
    \includegraphics[width=\textwidth]{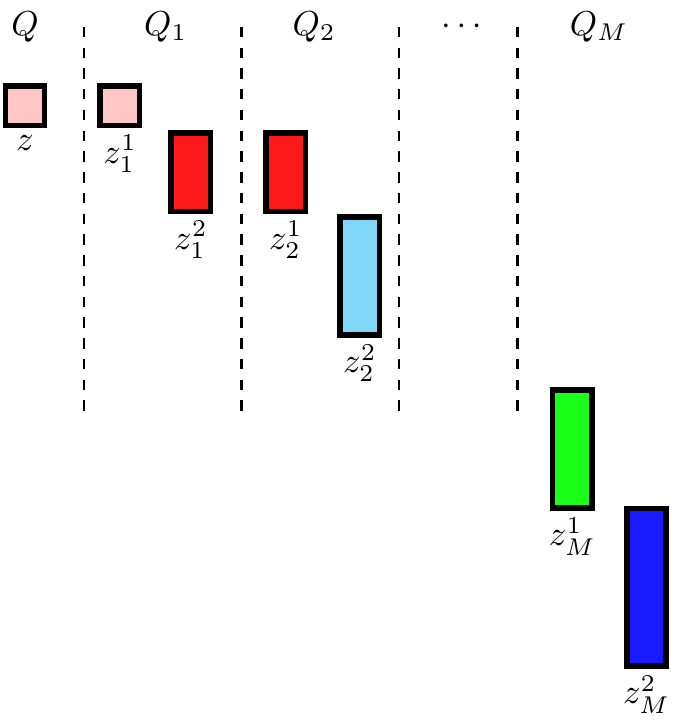}
    \caption{W-RDiff sampling strat.}
    \label{fig:mlmc_sampling_scheme}
  \end{subfigure}
  ~
  \begin{subfigure}[b]{0.23\textwidth}
    \centering
    \includegraphics[width=\textwidth]{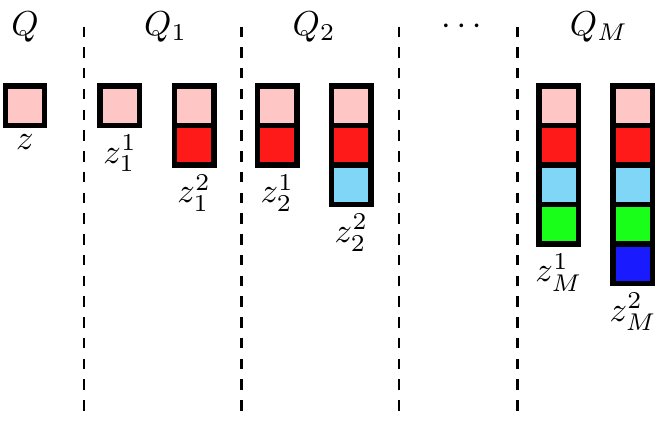}    
    \vspace{47pt}
    \caption{MFMC sampling strat.}
    \label{fig:mfmc_sampling_scheme}
  \end{subfigure}
  ~
  \begin{subfigure}[b]{0.23\textwidth}
    \centering
    \includegraphics[width=\textwidth]{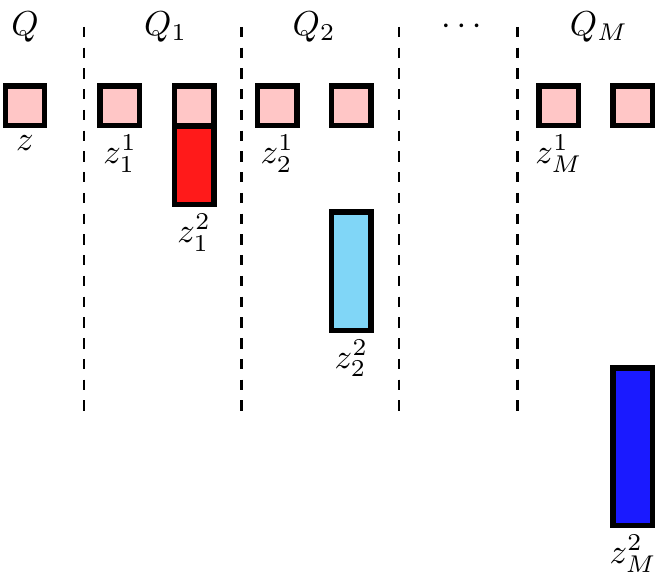}
    \vspace{20pt}
    \caption{ACV-IS sampling strat.}
    \label{fig:acv1_sampling_scheme}
  \end{subfigure}
  ~
  \begin{subfigure}[b]{0.23\textwidth}
    \centering
    \includegraphics[width=\textwidth]{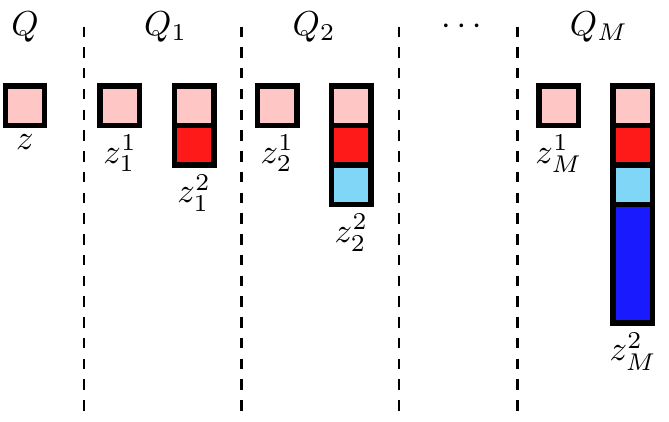}    
    \vspace{35pt}
    \caption{ACV-MF sampling strat.}
    \label{fig:acv2_sampling_scheme}
  \end{subfigure}
  \caption{Visualization of existing (W-RDiff and MFMC) and proposed (ACV-IS and ACV-MF) sampling strategies. Each column represents a distribution of samples, and the colors are used to indicate shared input samples among the levels.}
  \label{fig:existing_sampling_schemes}
\end{figure}

Table~\ref{tab:existingmethods} summarizes the sample distribution and the variance reduction ratio of the methods described to this point. Figure~\ref{fig:existing_sampling_schemes} displays the sampling schemes of the W-RDiff and MFMC algorithms in Figure~\ref{fig:existing_sampling_schemes}(a,b) alongside our proposed algorithms (detailed in \S\ref{sec:acv} to follow) in Figure~\ref{fig:existing_sampling_schemes}(c,d). 

\renewcommand{\arraystretch}{1.1}
\begin{table}
  {\footnotesize
  \begin{tabular}{|c|c|c|c|l|}
    \hline
    Algorithm & Relation between $\sset{}$ and $\sset{i}$ &$\sset{i}^1$ & $\sset{i}^2$ & \multicolumn{1}{c|}{Reduction ratio $\gamma$} \\
    \hline
    \hline
    OCV                           & $\sset{i} = \sset{}$   & $\sset{}$         & $\emptyset$             & $1 - R_{\cvl}^2$     \\
    OCV-1                    & $\sset{i} = \sset{}$   & $\sset{}$         & $\emptyset$             & $1 - \ccoeff{1}^2$ \\
    W-RDiff                  & $\sset{1}^1 = \sset{}$, $\sset{1} \cap \sset{i} = \emptyset$ for $i>1$ & $\sset{i-1}^2$  & $\sset{i} \setminus \sset{i}^1$  &    $1 - R_{\mll}^2$ \ for \ $R^2_{\mll} \color{red}\leq\color{black} \ccoeff{1}^2$ \\ 
    MFMC~\cite{Peherstorfer2016b} & $\sset{i} \supset \sset{}$ for all $i$ & $\sset{i-1}$    & $\sset{i}$  &
    $1 - R_{\mfmc}^2$ \ for \ $R^2_{\mfmc} \color{red}\leq\color{black} \ccoeff{1}^2$
    \\
    \hline
  \end{tabular}
  }
  \centering
  \caption{ Representations of various CV-type variance reduction estimators in the framework of Equation~\eqref{eq:acv}. The estimators OCV and OCV-1 refer to optimal CV estimator with known means, where OCV-1 only uses a single low-fidelity model and OCV uses all $\nmodels$. For these estimators, no samples are required to estimate $\mu$ and therefore  $\sset{i}^2$ is empty. As summarized in the final column, the greatest variance reduction possible with W-RDiff and MFMC is strictly less than or equal to the optimal  CV using a single model.}
  \label{tab:existingmethods}
\end{table}

To demonstrate these results we consider the following simple monomial example: let $\qoi(\omega) = \omega^5$ and $\cv{i}(\omega) = \omega^{5-i}$ for $i = 1,\ldots, 4$, where $\omega \sim \mathcal{U}(0,1)$. The correlation matrix for this problem is given in Table~\ref{tab:corr}. As we have not yet introduced a cost model, we first explore performance through the lens of an assumed sample ratio $r_i = 2^{i+x}$ so that $\vec{r}(x) = [2\times 2^x, 4\times 2^x, 8\times 2^x, 16 \times 2^x]$ for $x = 0,1,\ldots,29$.  As a result, the sample allocations across $i$ are prescribed by $r_i$ without assuming any relationship between these allocations and either the model cost $w_i$ and/or the estimated $\rho_{ij}$; this viewpoint is taken in Figures~\ref{fig:var_reduct}
,~\ref{fig:var_reduct_inter}, and~\ref{fig:var_reduct_inter2}. We explore the introduction of $\cost_i$ and resulting optimal allocation strategies later in \S\ref{sec:sample-allocation} and \S\ref{sec:num_model}. Therefore, in this section, we are concerned only with the problem of demonstrating the existence of a significant gap in variance reduction between OCV-1 and OCV. This gap depends only on the covariance and correlation matrices, and it is not related to the sample allocation per model. However, the possibility to exploit this gap ultimately depends on the actual sample allocation strategy, which is strictly related to model cost, as well as correlation.  This aspect will be discussed later in \S\ref{sec:sample-allocation}.

\begin{table}
  \centering
  {\footnotesize
  \begin{tabular}{|c|ccccc|}
    \hline
             & $\qoi$   & $\cv{1}$ & $\cv{2}$  & $\cv{3}$ & $\cv{4}$  \\
    \hline
    $\qoi$   & 1        & 0.994995 &  0.975042 & 0.927132 & 0.820633 \\
    $\cv{1}$ & 0.994995 & 1        &  0.992172 & 0.958367 & 0.865941 \\
    $\cv{2}$ & 0.975042 & 0.992172 &  1        & 0.986021 & 0.916385 \\
    $\cv{3}$ & 0.927132 & 0.958367 &  0.986021 & 1        & 0.968153 \\
    $\cv{4}$ & 0.820633 & 0.865941 &  0.916385 & 0.968153 & 1        \\
   \hline
  \end{tabular}
  }
  \caption{Correlation matrix for monomial example computed with $10^5$ samples.}
  \label{tab:corr}
\end{table}

Figure~\ref{fig:var_reduct} shows the computed variance reduction ratio $\gamma$. The dotted horizontal lines provide baseline variance reduction ratios corresponding to MC ($\gamma=1$), the single OCV (OCV-1), double OCV (OCV-2), triple OCV (OCV-3), and the optimal control variate that uses all the models (OCV). The variance reduction of the  estimators W-RDiff and MFMC is bounded by that provided by OCV-1, \ie $\ccoeff{1}^2$, as expected from the theory.

To summarize, this example shows that, when the number of high fidelity samples is fixed, the existing recursive estimators do not achieve the maximum-possible variance reduction. Further, we are aware of no existing methods in the form of~\eqref{eq:acv} that can converge to the variance reduction achieved by the OCV estimator (with known means) as the amount of low-fidelity data increases. Our goal in the following section is to overcome this deficiency by developing algorithms that can converge to the OCV estimator (red line in Figure~\ref{fig:var_reduct}).

\begin{figure}
  \centering
  \includegraphics[width=0.4\textwidth]{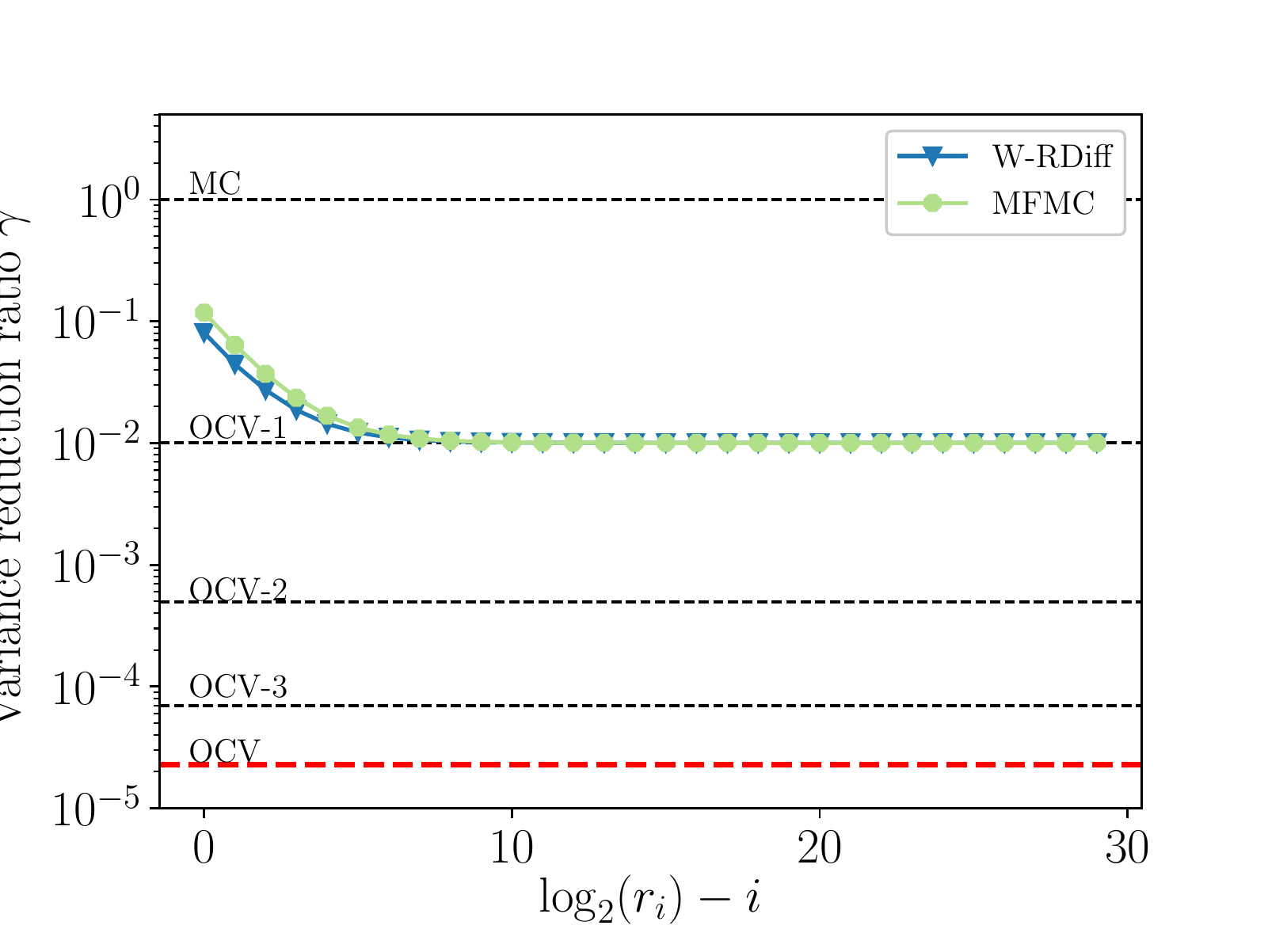}
  \caption{Variance reduction ratios for various estimators as a function of number of sample per level $r_i$ and fixed number of high fidelity samples. Baselines are provided for one (OCV-1), two (OCV-2), three (OCV-3), and four (OCV) CV estimators. Both W-RDiff and MFMC converge to the OCV-1 baseline as expected. None of the estimators converge to the optimal control variate shown in red.}
  \label{fig:var_reduct}
\end{figure}

\section{Approximate Control Variates}\label{sec:acv}
In this section, we use Propositions~\ref{prop:acv_var} and \ref{prop:ocv} to derive approximate control variate (ACV) estimators that provably converge to the optimal control variate (OCV) estimators with increasing low-fidelity data.
To provide motivation for the following discussion, compare the expressions for the OCV estimator variance~\eqref{eq:optcv_cov} and the expression for ACV estimator variance~\eqref{eq:ocv_var}. In the limit of infinite low-fidelity data, we intend for these two expressions to match. The recursive approaches described in \S\ref{sec:mfmc} and \S\ref{sec:mlmc} are limited because their recursive sampling patterns lead to either diagonal (MFMC) or tridiagonal (W-RDiff) covariance structures (see the Appendices), %
and we must recover the full covariance matrix to achieve convergence to OCV. We describe two such approaches next.

\subsection{Two convergent estimators}
The most straightforward way to obtain an ACV estimator with the same covariance structure of an OCV estimator, in terms of $\covdiff$ and $\covhc$, is to set $\sset{i}^1 = \sset{}$ for each CV and then to use all available samples to estimate $\estm{i}$, \ie, $\sset{i}^2 = \sset{i}$. In other words, the estimation of $\est{}$ and each $\est{i}$ employ the exact same set of samples and the estimation of each $\estm{i}$ uses these same samples plus a sample increment.
This simple estimator is termed approximate control variate-independent samples (ACV-IS) and is detailed below.
\begin{definition}[ACV-IS]
  Let $\sset{i}^1 = \sset{}$, $\sset{i}^2 = \sset{i}$ and $(\sset{i} \setminus \sset{i}^1) \cap (\sset{j} \setminus \sset{j}^1) = \emptyset$ for $i\neq j$ and  $i = 1,\ldots, \nmodels$. Then the estimator ACV-IS is defined as 
  \begin{equation}  \label{eq:acv-1}
      \acvoneest(\vec{\cvw},\vec{\sset{}}) = \est{}(\sset{}) + \sum_{i=1}^{\nmodels}\cvw_i\left(\est{i}\left(\sset{} \right) - \estm{i}\left(\sset{i}\right)    \right).
\end{equation}
\end{definition}
The estimator only requires shared evaluations for the $\nhf$ input samples that comprise $\est{}$ and each $\est{i}$. The rest of the samples $\sset{i} \setminus \sset{i}^1$ for each control variate are completely independent. As a result, an attractive feature of the sample distribution strategy for ACV-IS is that each control variate can be evaluated separately and in parallel. The sampling scheme for the ACV-IS estimator is illustrated in Figure~\ref{fig:acv1_sampling_scheme} for reference.
The optimal weights and variance reduction of the ACV-IS estimator are provided below, and the proof is provided in~\ref{app:proof:opt_acv1}.
\begin{theorem}[Optimal CV-weights and variance reduction for ACV-IS] \label{th:opt_acv1}
The optimal CV weights and estimator variance for the ACV-IS estimator are\footnote{In this paper, $\circ$ denotes a Hadamard, or elementwise, product.}
  \begin{equation}
    \optacvone = - \left[ \covm \circ \fmatone\right]^{-1} \left[\Diag{\fmatone} \circ \covv \right]
  \end{equation}
  and
  \begin{equation}
    \VarF{\acvoneest(\optacvone)} = \frac{\varF{\qoi}}{\nhf}\left(1 - R_{\acvone}^2\right),
    \textrm{ where }     R_{\acvone}^2 = \bvec{a}^{\normalfont T} \left[\covm \circ \fmatone\right]^{-1} \bvec{a},
  \end{equation}
  $\bvec{a}= \left[\Diag{\fmatone} \circ \covvnorm \right]$
  and $\fmatone \in \reals^{\nmodels \times \nmodels}$ has elements 
  \begin{equation}
    \fmatone_{ij} = \left\{
      \begin{array}{cl}
        \frac{\absr{i}-1}{\absr{i}}\frac{\absr{j}-1}{\absr{j}}  & \textrm{if } i \neq j \\
        \frac{\absr{i}-1}{\absr{i}} & \textrm{otherwise} 
      \end{array}
      \right. .
  \end{equation}
\end{theorem}
Another convergent ACV estimator can be obtained directly from MFMC. In this case, since a common set of samples is already present among all the models, we only  need to break the recursive pattern by using $\sset{i}^1=\sset{}.$
\begin{definition}[ACV-MF]\label{def:acv2}
  Let $\sset{i}^1 = \sset{}$, $\sset{i}^2 = \sset{i}$ and $\sset{j}^{(k)} = \sset{i}^{(k)}$ for $j > i$ and $k \leq \min(\absr{i},\absr{j})\nhf$. Then ACV-MF estimator is 
  \begin{equation}  \label{eq:acv-2}
      \acvtwoest(\vec{\cvw},\vec{\sset{}}) = \est{}(\sset{}) + \sum_{i=1}^{\nmodels}\cvw_i\left(\est{i}\left(\sset{} \right) - \estm{i}\left(\sset{i}\right)\right).
\end{equation}
\end{definition}
The ACV-MF estimator can use an identical set of samples to the MFMC estimator and can thus be considered a drop-in replacement. \textit{The only difference is that $\est{i}$ is evaluated using only the first $\nhf$ samples instead of $\absr{i-1}\nhf$.} Furthermore, using fewer samples for $\est{i}$ does not cause loss of accuracy because the CV approach does not require an accurate estimate of $\est{i}$ in terms of how close it is to $\mean{i}$; rather, it requires an estimator $\est{i}$ that is \textit{correlated to} $\est{}$ and unbiased. The sampling scheme for the ACV-MF estimator is illustrated in Figure~\ref{fig:acv2_sampling_scheme} for reference.
The optimal weights and variance reduction of the ACV-MF estimator is provided below, and the proof is provided in~\ref{app:proof:opt_acv2}.
\begin{theorem}[Optimal CV-weights and variance reduction for ACV-MF]\label{th:opt_acv2}
  The optimal CV weights and estimator variance for the ACV-MF estimator are
  \begin{equation}
      \optacvtwo = - \left[ \covm \circ \fmattwo\right]^{-1} \left[\Diag{\fmattwo} \circ \covv \right],
  \end{equation}
  and
  \begin{equation}
    \VarF{\acvtwoest(\optacvtwo)} = \frac{\varF{\qoi}}{\nhf}\left(1 - R_{\acvtwo}^2\right),
    \textrm{ where }
    R_{\acvtwo}^2 = \bvec{a}^{\normalfont T} \left[\covm \circ \fmattwo\right]^{-1} \bvec{a},
  \end{equation}
  $\bvec{a} = \left[\Diag{\fmattwo} \circ \covvnorm \right]$
  and $\fmattwo \in \reals^{\nmodels \times \nmodels}$ has elements
  \begin{equation}
    \fmattwo_{ij} = \left\{
      \begin{array}{cl}
        \frac{\min(\absr{i},\absr{j})-1}{\min(\absr{i},\absr{j})}  & \textrm{if } i \neq j \\
        \frac{\absr{i}-1}{\absr{i}} & \textrm{otherwise} 
      \end{array}
      \right. .
  \end{equation}
\end{theorem}
Interestingly, the form of the optimal estimators for ACV-IS and ACV-MF only differ in the terms containing the matrices $\fmatone$ and $\fmattwo$. The way in which these matrices enter the covariance calculation is algebraically identical. 
We analyze the conditions under which this algebraic form converges to the optimal CV in the proposition below.
\begin{proposition}[Convergent estimators]
  If an approximate CV estimator with $\vec{r} = [\absr{1},\ldots,\absr{\nmodels}]$ yields an optimal variance reduction with 
  \begin{equation}
    R^2_{\mat{G}(\vec{r})} = \left[\Diag{\mat{G}} \circ \covv \right]^{\normalfont T} \frac{\left[\covm \circ \mat{G} \right]^{-1}}{\varF{\qoi}} \left[\Diag{\mat{G}} \circ \covv \right],
  \end{equation}
  and $\mat{G}(\vec{\absr{}}) \to \mat{1}_{\nmodels \times \nmodels}$ then the estimator converges to the optimal  CV
  \begin{equation}
    \lim_{\vec{\absr{}} \to \infty} R_{\mat{G}(\vec{\absr{}})}^2 = \frac{\covv^{\normalfont T} \covm^{-1} \covv}{\VarF{Q}} = R_{OCV}^2,
  \end{equation}
  where $\vec{\absr{}} = [\absr{1},\ldots,\absr{\nmodels}]$ and $\vec{\absr{}} \to \infty$ means that $r_i\to \infty$ for $i = 1,\ldots,\nmodels.$. 
\end{proposition}
The proof of this proposition is self evident since $\mat{1}_{\nmodels} \circ \covv = \covv$ and $\mat{1}_{\nmodels \times \nmodels} \circ \covm = \covm$. 

\begin{theorem}[Convergence of ACV-IS and ACV-MF estimators to the optimal  CV.]\label{th:ACV}
  The variance reduction of ACV-IS and ACV-MF converges to that of the optimal  CV with increasing data
  \begin{equation}
    \lim_{\vec{\absr{}} \to \infty}R_{\acvone}^2 = \lim_{\vec{\absr{}} \to \infty} R_{\acvtwo}^2 = R^2.
  \end{equation}
\end{theorem}
The proof is provided in~\ref{app:proof:ACV}.

We now return to the monomial example of \S\ref{sec:motivate_example}. In Figure~\ref{fig:var_reduct_inter}, we add the proposed estimators ACV-IS and ACV-MF. As expected, these estimators converge to the OCV baseline. 
\begin{figure}
  \centering
  \includegraphics[width=0.35\textwidth]{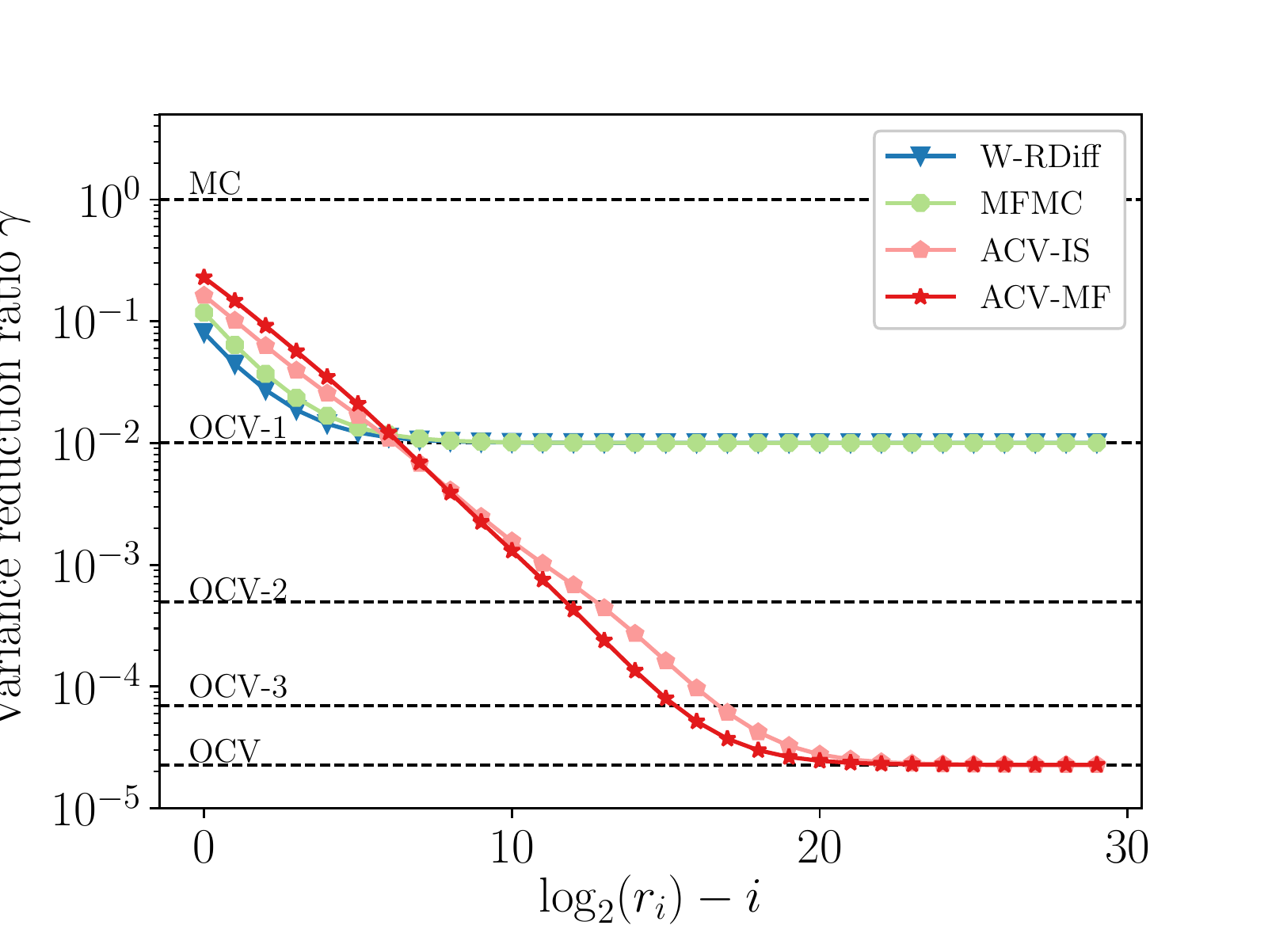}
  \caption{Comparison of variance reduction of ACV-IS and ACV-MF with W-RDiff and MFMC on the example from \S\ref{sec:motivate_example}. ACV-IS and ACV-MF are the only ones to converge to OCV}
  \label{fig:var_reduct_inter}
\end{figure}
When the ratios $\absr{i}$ are small (left side of Figure~\ref{fig:var_reduct_inter}), we are in a regime where the variance reduction is less than OCV-1. 
For large ratios $\absr{i}$ (right side of Figure~\ref{fig:var_reduct_inter}), both ACV estimators converge to OCV. In regimes with smaller ratios, it might be appropriate to target a particular variance reduction level based on the ``maximum'' attainable variance reduction. Once this level is identified, we can use recursion to accelerate the estimator to this optimum. Such a hybrid approach is described in the next section.

To summarize, the key features of ACV-IS and ACV-MF are given in Table~\ref{tab:novelmethods} and Figures~\ref{fig:acv1_sampling_scheme} and~\ref{fig:acv2_sampling_scheme}.

\begin{table}
{\footnotesize  
    \begin{tabular}{|c|c|c|c|c|}
    \hline
    Algorithm & Relation between $\sset{}$ and $\sset{i}$ &$\sset{i}^1$ & $\sset{i}^2$ & Reduction Ratio $\gamma$ \\
    \hline
    \hline
    ACV-IS   & $\sset{} \cap \sset{i} = \sset{i}^1$, $\sset{i}^2 \cap \sset{j}^2 = \emptyset$ for $1\leq i \neq j$& $\sset{}$ & $ \sset{i} \setminus \sset{i}^1 $ & $1 - R_{\acvone}^2$ Th.~\ref{th:opt_acv1} 
    \\
    \hline
    ACV-MF   & $\sset{} \cap (\sset{i}\setminus\sset{i}^1) = \emptyset$ & $\sset{}$ &  $\sset{i}$ & $1 - R_{\acvtwo}^2$ Th.~\ref{th:opt_acv2} 
    \\
    \hline
    \end{tabular}
}
  \centering
  \caption{Summary of the two convergent estimators ACV-IS and ACV-MF. The main difference between the two estimators is that the sets $\sset{i}^2$ are independent between two distinct models for ACV-IS, whereas we have $\sset{i+1}^2 \cap \sset{i}^2 = \sset{i}^2$ for ACV-MF. ACV-MF is closely related to MFMC; the main difference is that only $\nhf$ samples are used for $\sset{i}^1$. The variance reduction of both ACV-IS and ACV-MF \textit{in the limit of infinite samples of $\cv{i}$} is greater than or equal to $\ccoeff{1}^2$ ($R_{\acvone}^2 \geq \ccoeff{1}^2$ and $R_{\acvtwo}^2 \geq \ccoeff{1}^2$). For reference, the reverse is true for W-RDiff and MFMC.}
  \label{tab:novelmethods}
\end{table}

\subsection{Accelerating the approximate CV}
Sections~\ref{sec:mfmc} and \ref{sec:mlmc} show that a recursive CV estimator limits the maximum achievable variance reduction. However these strategies may be useful to accelerate convergence to a target CV level in cases where there is not enough data from higher-fidelity models to achieve the targeted performance directly. This targeted level could be the single or two-level optimal CV (OCV-1 and OCV-2), or any other level up to OCV (all available CVs).
The recursive techniques we have discussed all had sampling strategies that accelerated the convergence of $\estm{i}$ by CV $\cv{i+1}$, and as such were limited to convergence to OCV-1. 
In this section, we develop an algorithm to show how our framework can be used to create new schemes that accelerate convergence to any given target. This proposed algorithm is only one realization of a myriad of possible approaches. Our intention is to provide an example that uses this general framework to develop and explore new algorithms. 

The scheme is conceptually simple. We first partition all of the control variates into two groups; the first $K$ variables form a $K$-level approximate control variate, and the last $\nmodels-K$ variables are used to reduce the variance of estimating $\mu_{L}$ for some $L \leq K$. The resulting estimator accelerates convergence to OCV-$K$, and $L$ provides a degree of freedom for targeting a control variate level that contributes the greatest to the estimator variance.
\begin{definition}[ACV-KL: accelerated ACV]
  Let $K,L \leq \nmodels$ and $K \in \posint$ with $0 \leq L \leq K$. The ACV-KL estimator is 
  \begin{equation}  \label{eq:K-acv2}
    \begin{split}
      &\Kacvest(\vec{\cvw},\vec{\sset{}}) = \est{}(\sset{}) + \sum_{i=1}^{K}\cvw_i\left(\est{i}\left(\sset{}\right) - \estm{i}\left(\sset{i}\right)\right) + 
       \sum_{i=K+1}^{\nmodels} \cvw_i\left(\est{i}\left(\sset{L}\right) - \estm{i}\left(\sset{i}\right)\right),
    \end{split}
  \end{equation}
    where we select $\sset{i}^1 = \sset{}$ for $i \leq K$ and $\sset{i}^1 = \sset{L}$ for $i > K$. Furthermore $\sset{i}^2 = \sset{i}$ for all $i$. The sets $\sset{i} \setminus \sset{i}^1$ can be chosen in several ways. Here we choose the same sampling strategy as ACV-MF\footnote{Note that the ACV-IS sampling strategy for $\sset{i}^2$ could also have been chosen, but we do not analyze this approach here. Our aim is to demonstrate a basic framework for deriving approximate CV estimators, and many combinations are possible. We have chosen representative realizations of the framework to convey the main concepts.}: $\sset{i}^1 = \sset{}$, $\sset{i}^2 = \sset{i}$ and $\sset{j}^{(k)} = \sset{i}^{(k)}$ for $j > i$ and $k \leq \min(\absr{i},\absr{j})\nhf$.
\end{definition}

This estimator differs from the previous recursive estimators because
the first two terms in \eqref{eq:K-acv2} correspond to an ACV-MF estimator with $K$ CVs and the last term adds a CV scheme to the ACV-MF estimator, \ie 
\begin{equation}
    \begin{split}
    \Kacvest(\vec{\cvw},\vec{\sset{}}) &= \acvtwoestK(\cvw_1,\ldots,\cvw_K,\sset{},\sset{1},\ldots,\sset{k}) + 
     \sum_{i=K+1}^{\nmodels} \cvw_i\left(\est{i}\left(\sset{L}\right) - \estm{i}\left(\sset{i}\right)\right).
  \end{split}
\end{equation}
The inclusion of the ACV-MF estimator enables the ACV-KL estimator to converge to the OCV estimator and the last term reduces the variance of $\estm{L}$, thereby accelerating convergence of the scheme. The optimal weights and variance reduction for the ACV-KL estimator are now provided.

\begin{theorem}[Optimal CV-weights and variance reduction for ACV-KL]\label{th:ACV-KL}
  Assume $\absr{i} > \absr{L}$ for $i > L$, then the optimal weights for the ACV-KL control variate are
  \begin{equation}
      \optKacv(K,L) = - \left[ \covm \circ \fmat^{(K,L)}\right]^{-1} \left[\Diag{\fmat^{(K,L)}} \circ \covv \right],
  \end{equation}
  and the estimator variance $V \equiv \VarF{\Kacvest (\optKacv(K,L))}$ is
  \begin{equation}
    V = \frac{\varF{\qoi}}{\nhf}\left(1 - R^2_{\Kacv}(K,L)\right),
    \textrm{ where } R_{\Kacv}^2(K,L) = \bvec{a}^{\normalfont T} \frac{\left[\covm \circ \fmat^{(K,L)}\right]^{-1}}{\varF{\qoi}} \bvec{a},
  \end{equation}
  $\bvec{a} = \left[\Diag{\fmat^{(K,L)}} \circ \covv \right]$
  and $\fmat^{(K,L)} \in \reals^{\nmodels \times \nmodels}$ has elements
  \begin{equation}
    \begin{split}
      & \fmat^{(K,L)}_{ij} = 
      \left\{
    \begin{array}{cl}
      \frac{\min(\absr{i},\absr{j}) - 1}{\min(\absr{i},\absr{j})} & \textrm{ if } i,j \leq K \\
      \frac{(\absr{i}-\absr{L})(\absr{j}-\absr{L}) + \absr{L}(\min(\absr{i},\absr{j}) - \absr{L})}{\absr{i}\absr{j}\absr{L}} & \textrm{ if } i,j > K \\
      \left[\frac{\absr{i} - \absr{L}} {\absr{i} \absr{L}} \right] & \textrm{ if } L < i \leq K, \ j > K \\
      \left[\frac{\absr{j} - \absr{L}} {\absr{j} \absr{L}} \right]  & \textrm{ if }  L < j \leq K, \ i > K \\
      0 & \textrm{ otherwise }
    \end{array}
    \right. , 
    \end{split}
  \end{equation}
  for $i \neq j.$ The diagonal elements are  $\fmat^{(K,L)}_{ii} = \frac{\absr{i}-1}{\absr{i}}$ if $i \leq K$ and $\fmat^{(K,L)}_{ii} = \frac{\absr{i} - \absr{L}}{\absr{i}\absr{L}}$ otherwise.
\end{theorem}

The proof, provided in~\ref{app:proof:ACV-KL}, is constructive in that it provides an explicit expression for the variance reduction and resulting estimator. Algorithm~\ref{alg:acvkl} provides pseudocode that summarizes this procedure. The algorithm requires (an estimate of) the covariances $\covm$ and $\covv$ and the variance of $\qoi$. Using these quantities, the algorithm provides an (approximate) optimal control variate weight and associated variance reduction for any given sample sizes defined by the ratios $r_i$ and the parameters $K$ and $L$. Note that, for this algorithm, $\fmat^{(K,L)}$ does not converge to $\mat{1}_{\nmodels \times \nmodels}$ as the number of samples goes to infinity, unless $K = L = \nmodels$ for which ACV-KL becomes ACV-MF. Furthermore, since this ACV-KL approach generalizes the ACV-MF sampling strategy, it is also a drop-in replacement of MFMC.

\begin{figure}
  \centering
  \begin{subfigure}[b]{0.355\textwidth}
    \centering
    \includegraphics[width=\textwidth,clip,trim=0 0 40 0]{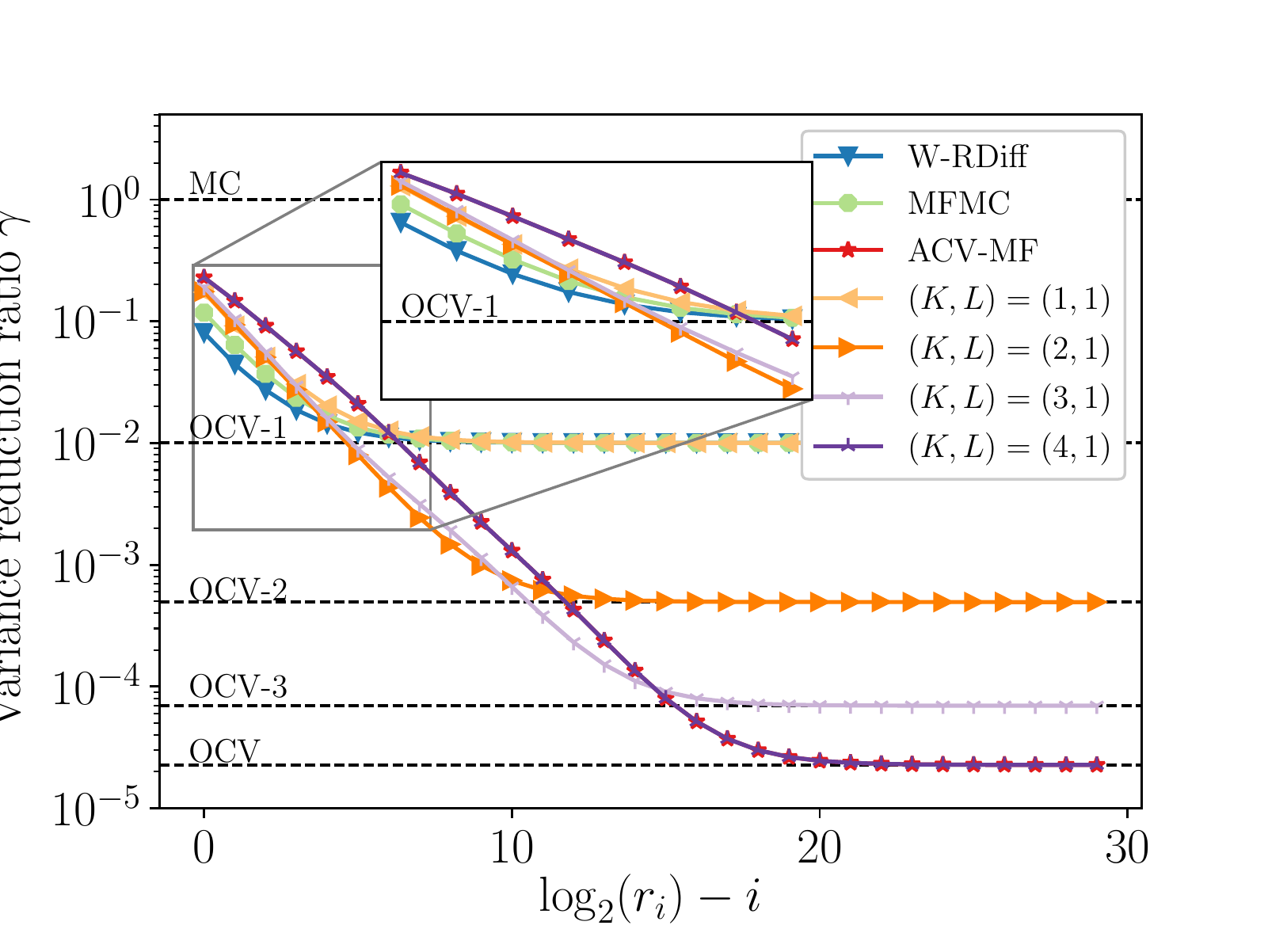}
    \caption{Varying $K$ with $L=1$}
  \end{subfigure}
  \begin{subfigure}[b]{0.31\textwidth}
    \centering
    \includegraphics[width=\textwidth,clip,trim=51 0 40 0]{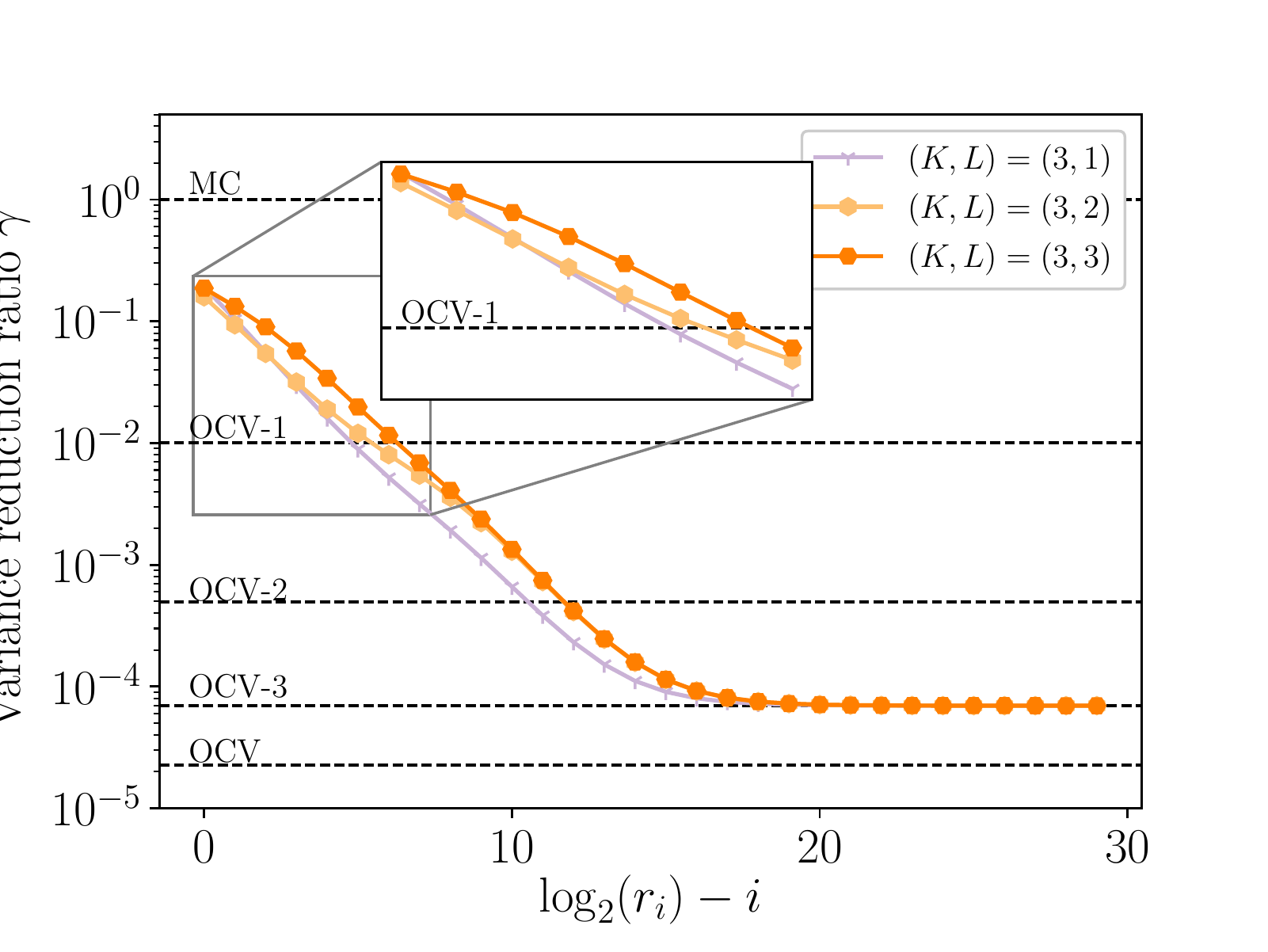}
    \caption{Varying $L$ fixed $K=3$}
  \end{subfigure}
  \begin{subfigure}[b]{0.31\textwidth}
    \centering
    \includegraphics[width=\textwidth,clip,trim=51 0 40 0]{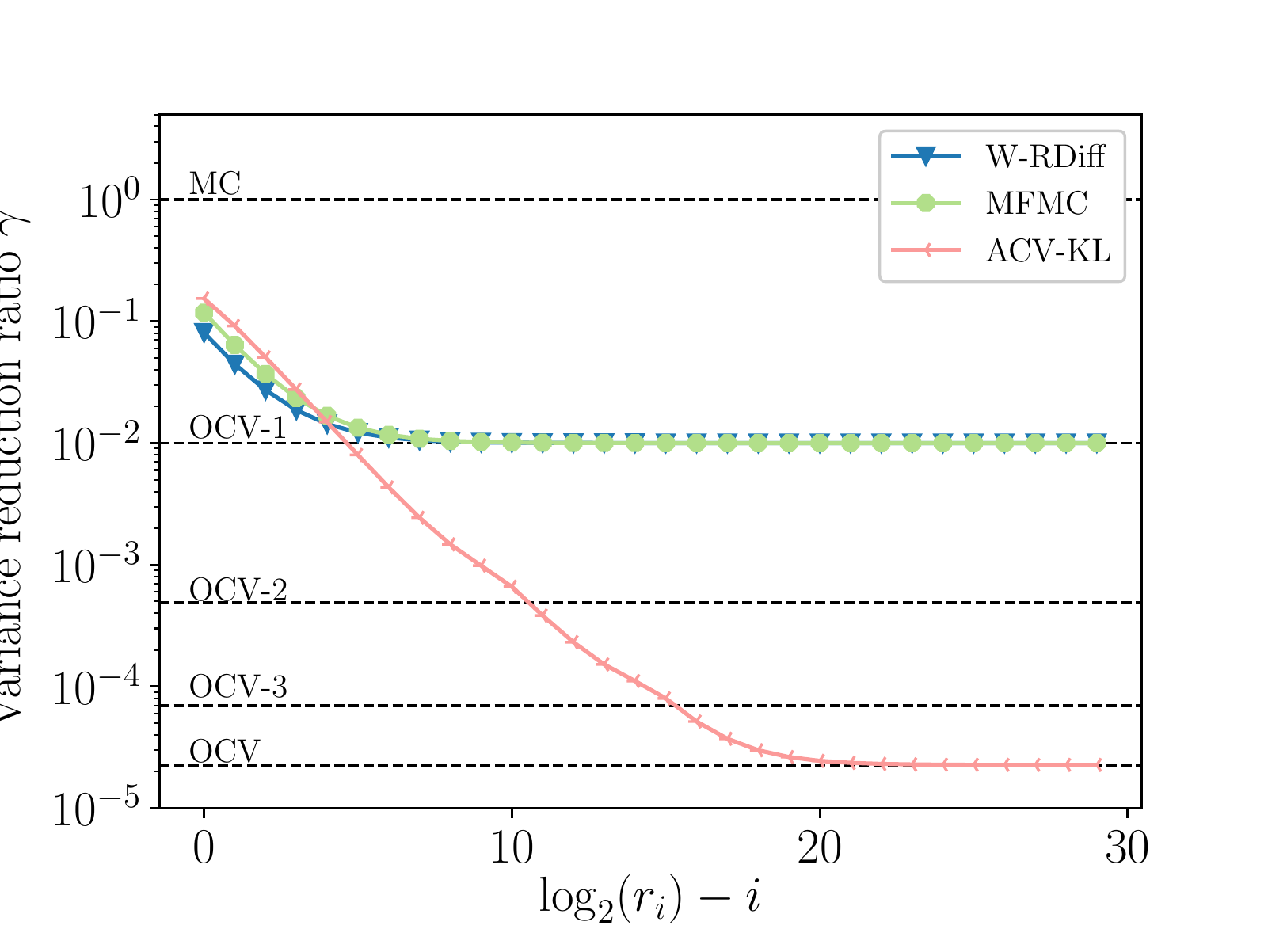}
    \caption{Best combination of $(K,L)$.}
    \label{fig:acv_kl_conv}
  \end{subfigure}
  \caption{Accelerated convergence to target levels by the ACV-KL class of estimators for various $(K,L)$ combinations on the monomial example. The kinks in optimal $(K,L)$ this plot are indicative of transitions between optimal $(K,L)$ combinations}
  \label{fig:var_reduct_inter2}
\end{figure}

The convergence of the ACV-KL estimator for various $(K,L)$ parameters is shown in Figure~\ref{fig:var_reduct_inter2}. The plot highlights that the ACV-MF estimator can be accelerated to any target baseline level, outperforming the baseline ACV-MF algorithm. Furthermore, the results demonstrate that it is possible for ACV-KL to achieve similar performance to the fully recursive algorithm in the low $\absr{i}$ region.

The best choice of $K$ and $L$ is problem dependent; however, they can be estimated at negligible cost. Specifically, we can embed Algorithm~\ref{alg:acvkl} inside an outer loop that, for a given evaluation strategy, searches over all combinations of parameters $K$ and $L$ to minimize the variance. This approach would essentially follow the lowest-variance line of the $(K,L)$ options shown in the combined left and middle panels of Figure~\ref{fig:var_reduct_inter2}. Figure~\ref{fig:acv_kl_conv} shows the performance of the ACV-KL estimator that chooses the best $(K,L)$ combination. As a result of these discrete choices, it has ``kinks'' as increasing sample sizes lead to different combinations.

{\footnotesize
\begin{algorithm}[h!]
  \caption{Approximate Control Variate (ACV-KL)} \label{alg:acvkl}
  \begin{algorithmic}[1]
    \REQUIRE {$\covm \in \reals^{\nmodels \times \nmodels}$: estimate of covariance among control variates; $\covv \in \reals^{\nmodels}$: estimate of covariance between $\qoi$ and each CV; $V$: estimate of the variance of the $\qoi$;  $(r_i)_{i=1}^{\nmodels}$: ratio of the number of evaluations of $\cv{i}$ to $\qoi$; $K$ and $L$: algorithm parameters where where $K \leq \nmodels$  and $L \leq K$.}
    \ENSURE $\vec{\cvw},R_{\Kacv}$: optimal weights and estimated variance reduction for the ACV-KL
    \STATE $\mat{A} = \textrm{zeros}(\nmodels, \nmodels)$
    \STATE $\mat{b} = \textrm{zeros}(\nmodels)$
    \FOR{$i = 1, \ldots, \nmodels$}
    \IF {$i \leq K$}
    \STATE $\mat{A}_{ii} = \frac{\absr{i}-1}{\absr{i}} \covm_{ii}$
    \ELSE
    \STATE $\mat{A}_{ii} = \frac{\absr{i} - \absr{L}}{\absr{i}\absr{L}} \covm_{ii}$ 
    \ENDIF
    \STATE $\mat{b}_i = \mat{A}_{ii} \covv_{i}$
    \FOR{$j = i+1, \ldots, \nmodels$}
    \IF { $i \leq K$ and $j \leq K$}
    \STATE $\mat{A}_{ij} = \left[\frac{\min(\absr{i},\absr{j}) - 1}{\min(\absr{i},\absr{j})} \right]\covm_{ij}$
    \ELSIF { $i > K$ and $j > K$}
    \STATE $\mat{A}_{ij} = \left[\frac{(\absr{i}-\absr{L})(\absr{j}-\absr{L}) + \absr{L}(\min(\absr{i},\absr{j}) - \absr{L})}{\absr{i}\absr{j}\absr{L}}\right]\covm_{ij}$
    \ELSIF {$L < i \leq K$ and $j > K$}
    \STATE $\mat{A}_{ij} = \left[\frac{\absr{i} - \absr{L}} {\absr{i} \absr{L}}\right]\covm_{ij}$
    \ENDIF
    \STATE $\mat{A}_{ji} = \mat{A}_{ij}$
    \ENDFOR
    \ENDFOR
    \STATE $\vec{\cvw} = -\mat{A}^{-1}\mat{b}$
    \STATE $R^2_{\Kacv} = \mat{b}^T\mat{A}^{-1}\mat{b} / V$
  \end{algorithmic}
\end{algorithm}
}

\subsection{Sample allocation}\label{sec:sample-allocation}
As previously mentioned, reducing the variance of a control variate estimator can be achieved by increasing the number of high fidelity evaluations and/or increasing the
number of low-fidelity simulations in order to exploit the correlation structure among models to the greatest extent possible. An efficient control variate estimator needs to reach the minimum overall variance by investing a fixed computational budget where it is more effective, \ie where the variance reduction per unit cost is greatest. 

We will seek to minimize the estimator variance subject to a constraint on the total cost. The form of the (un-weighted) RDiff and MFMC estimator enable analytic closed-form solutions to similar minimization problems, and their optimal allocation strategies can be found in~\cite{Giles2008} and~\cite{Peherstorfer2016b}, respectively. In our case the (weighted) RDiff, ACV-MF, and ACV-KL estimators do not yield analytic solutions, as far as we are aware, because of the complex inversion of $\covdiff$. In this paper, we rely on optimization approaches.

Let $J_{ACV}(\nhf,\vec{\absr{}}) = \left(1 - R^2_{ACV}\right) \frac{\VarF{\qoi}}{\nhf}$ denote the objective function for some ACV, where the expressions for $R^2_{ACV}$ are dependent on the ACV type. We minimize this objective subject to an inequality constraint on the cost and linear and bound constraints on the sampling design parameters:
\begin{equation}\label{eq:sample-allocation-optimization}
  \min_{\nhf, \vec{\absr{}},K,L}  \ \log(J_{ACV}(\nhf, \vec{\absr{}},K,L)) \quad  \textrm{ subject to  } 
   \nhf \left(\cost + \sum_{i=1}^{\nmodels} \cost_i \absr{i}\right) \leq C,  \quad
  \nhf      ~\geq ~1,     \quad 
  \absr{1}  ~\geq ~1,
\end{equation}
where  $\cost$ and $\cost_i$ denote the cost of obtaining a sample of $\qoi$ and $\cv{i}.$ In addition to the above constraints and unless specified otherwise in the numerical results, we constrain $\absr{i-1} > \absr{i}$. This constraint is not technically required, but rather empirically motivated as we have found it leads to more robust results. Again, we emphasize that our goal is not to provide the best optimization formulation, but rather a preliminary approach to demonstrate our theoretical findings.
For this problem to be well posed, we require $C \geq \cost + \sum_{i=1}^{\nmodels}\cost_i$, \ie, that the cost allowable is larger than that which corresponds to evaluating the quantity of interest and each control variates a single time.

In terms of implementation, we use a local gradient-based optimization procedure, in particular the interior point method SLSQP (see  \eg ~\cite{Nocedal2006}), to optimize over $\nhf$ and $\absr{i}$.  We use the automatic differentiation tool of the python PyTorch\footnote{https://pytorch.org} library to obtain gradients of the objective with respect to the optimization variables.  A couple of comments are in order:
\begin{enumerate}
  \item We minimize the log because it is better-conditioned 
  \item The covariances among the QoI and the CVs are estimated from pilot samples
\item  We relax the the integer optimization problem over $N$ by considering it as a continuous variable and then taking the ceiling of all non-integer sample allocations. 
\end{enumerate}
We leave developing additional optimization formulations for future work. Our intention here is to provide a baseline sample allocation procedure for studying the estimators.

\section{Numerical Experiments}\label{sec:numexamples}
Now we consider several numerical experiments to demonstrate the results of the theory.

\subsection{Model problem}\label{sec:num_model}
First, we again consider the monomial example from \S\ref{sec:motivate_example}. Recall, that the efficiency of the ACV estimators ultimately depends on the particular problem only through the correlation matrix and the costs. The correlation matrix in Table~\ref{tab:corr} is reasonable for what might be observed in realistic application
scenarios. However, the costs have thus far been unspecified. Here we investigate the effects of various cost prescriptions.

In Figure~\ref{fig:sample_opt}, we plot the actual variance reduction for two different cost prescriptions. Both have $\cost = 1$, but they differ in the gap between the QoI and the first control variate; in the left panel we have $\cost_i = 10^{-i}$ and in the right panel we have $\cost_i=10^{-i-1}$ for $i=1,2,3,4$. These results are based on solving \eqref{eq:sample-allocation-optimization} to select the number of samples assigned to each model. We see that all of the recursive estimators perform virtually identically, while ACV-KL provides greater variance reduction.
\begin{figure}
  \centering
  \begin{subfigure}[t]{0.45\textwidth}
    \centering
    \includegraphics[width=0.8\textwidth]{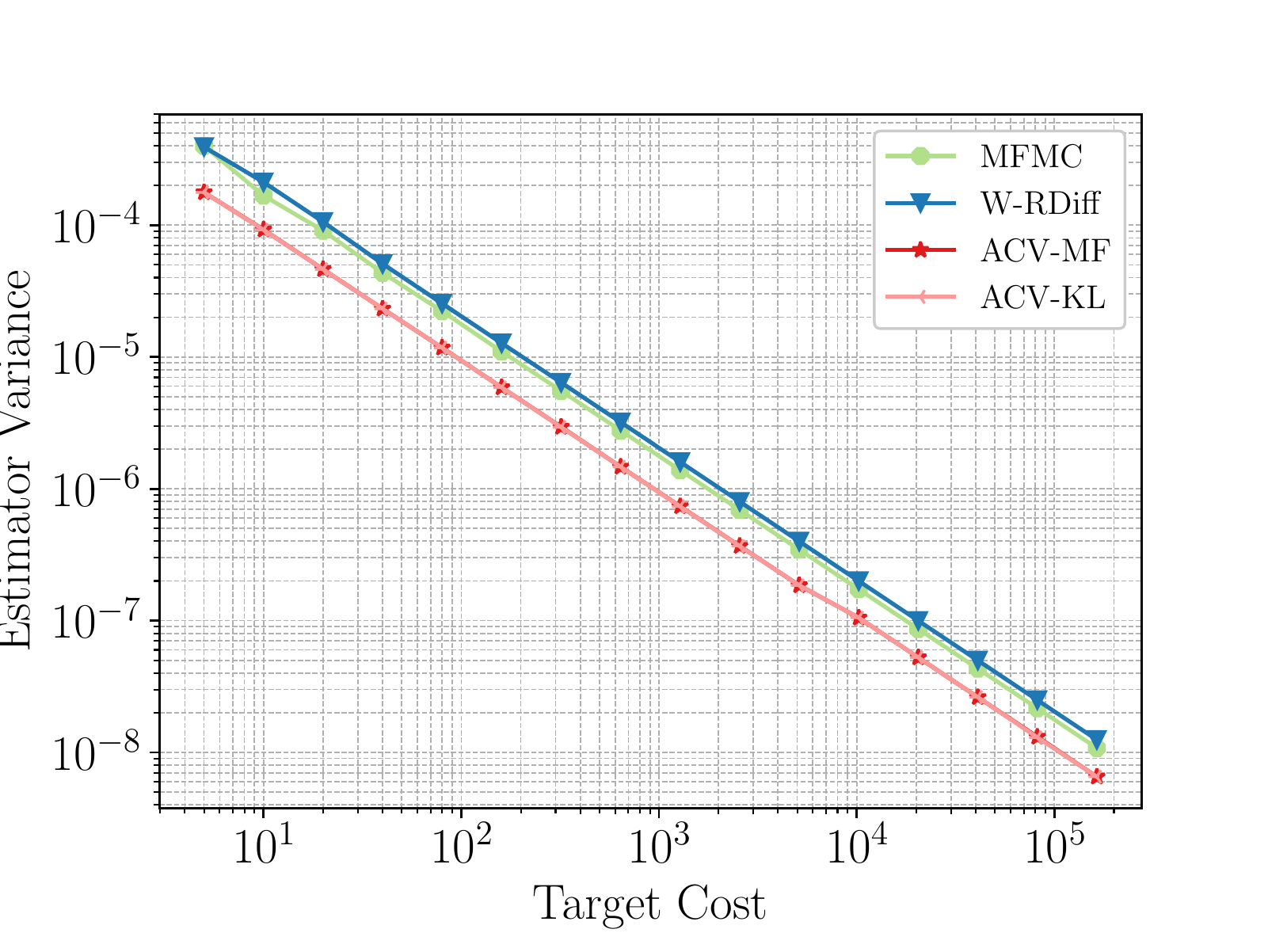}
    \caption{$\cost = 1$ and $\cost_i = 10^{-i}$ for $i = 1,2,3,4$.}
  \end{subfigure}
  \begin{subfigure}[t]{0.45\textwidth}
    \centering
    \includegraphics[width=0.8\textwidth]{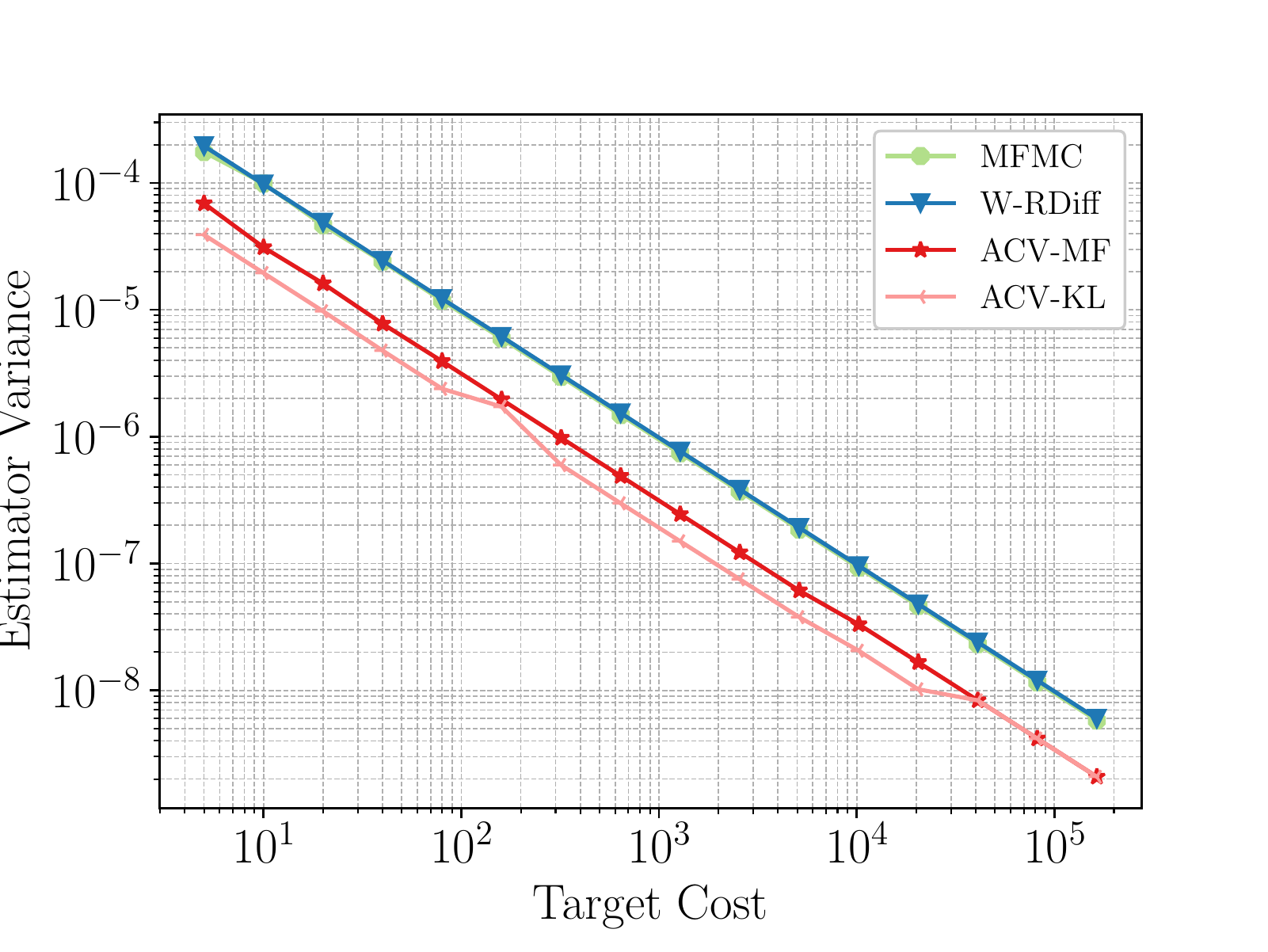}
    \caption{$\cost = 1$ and $\cost_i = 10^{-i-1}$ for $i = 1,2,3,4$.}
  \end{subfigure}  
  \caption{Variance reduction for the example from \S\ref{sec:motivate_example} under an optimal allocation strategy for each estimator. The W-RDiff and MFMC estimators perform virtually identically. The non-recursive ACV-KL estimator achieves significantly greater variance reduction. Right panel indicates greater reduction for greater cost difference.}
  \label{fig:sample_opt}
\end{figure}

The performance of all the sampling algorithms discussed in this paper are dependent on the relative cost of the models used. In Figure~\ref{fig:cost_vary}, we compare the variance reduction of W-RDiff and ACV-KL, when they are applied to the three monomial models, as we vary the costs of the two control variates. We constrain our search space to the case where the second control variate is less expensive to evaluate than the first, \ie, $\cost_2 < \cost_1$. Our goal is to demonstrate the scales at which the ACV-KL approach is able to provide significant variance reduction over recursive estimators. We only consider W-RDiff for this test because the results for other estimators are very similar (e.g., see Figure~\ref{fig:sample_opt}). Figure~\ref{fig:cost_vary} demonstrates that ACV-KL is able to achieve significant performance gains over the recursive approaches when the cost of the first control variate is at least 100 times lower than the truth model ($\cost_1 < 0.01$). 
In virtually all cases, ACV-KL performs better than regular Monte Carlo.

\begin{figure}
  \centering
  \includegraphics[scale=0.35]{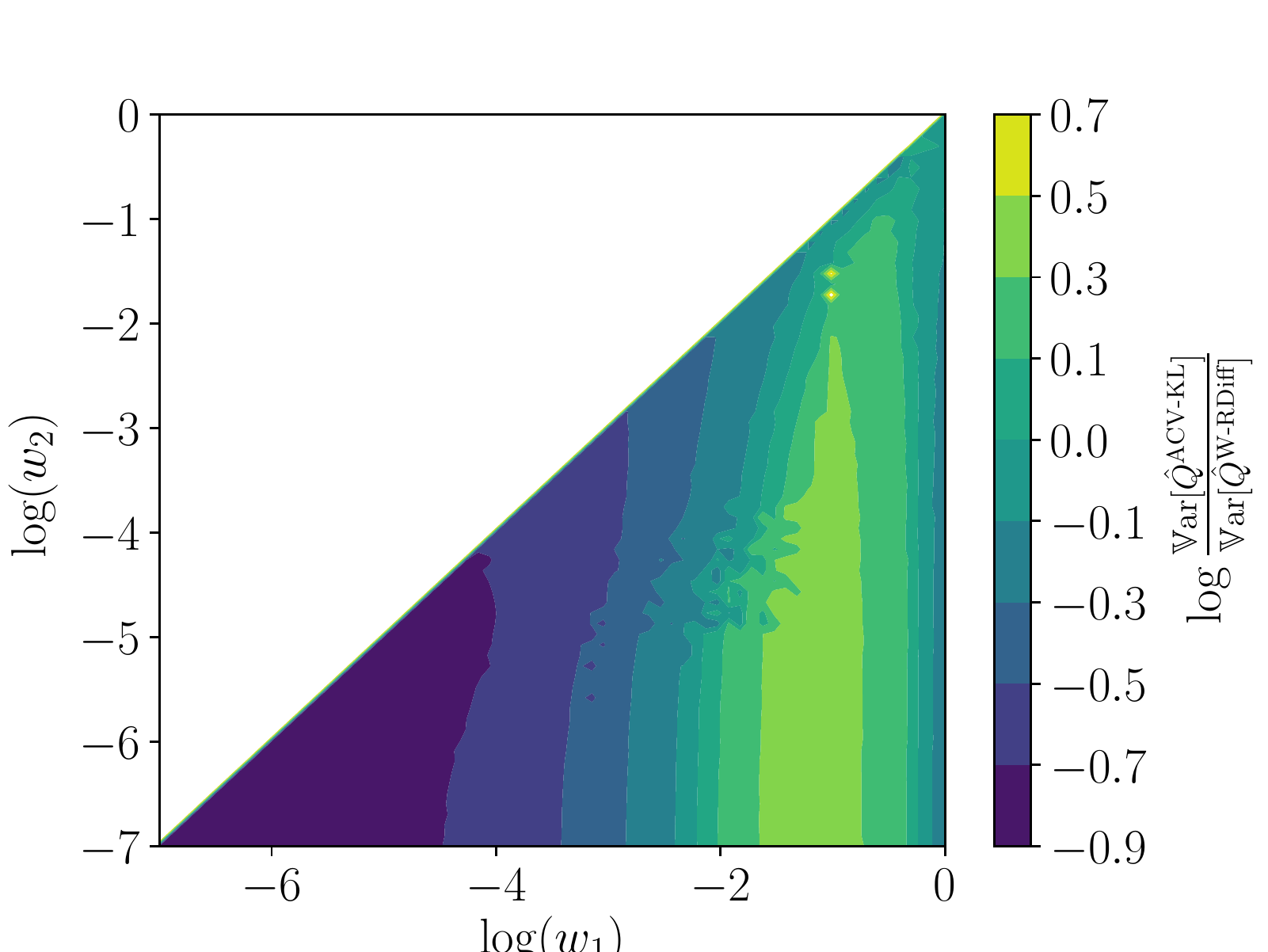}
  ~
  \includegraphics[scale=0.35]{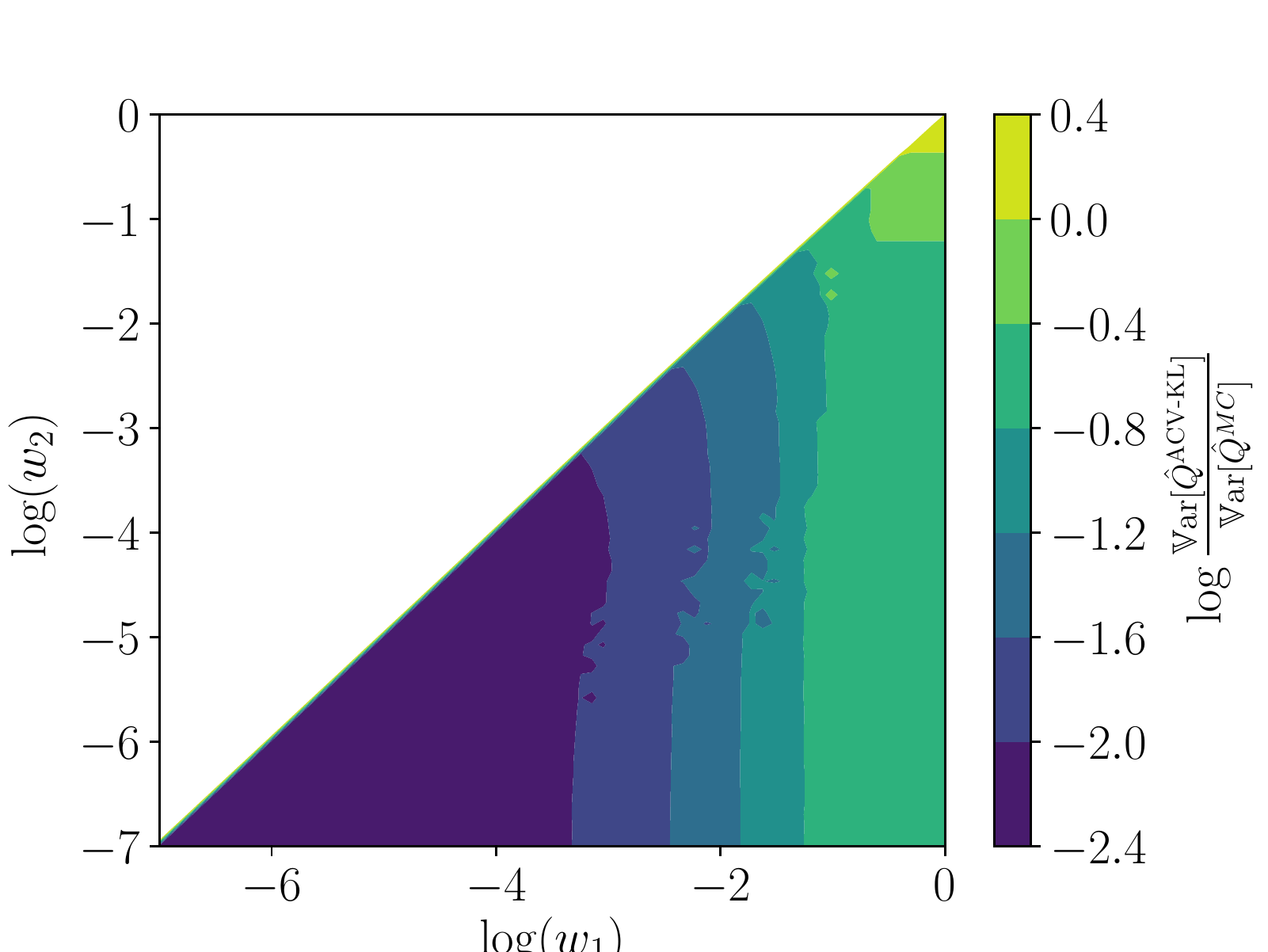}  
  \caption{Ratio of variance reduction achieved by ACV-KL compared to W-RDiff (left) and Monte Carlo (right) in the case of three total models ($\qoi(\om) = \om^3$, $\cv{1}(\om) = \om^2$, and $\cv{2}(\om) = \om$) as a function of the costs ($w_1, w_2$) of evaluating the control variates. The cost of the qoi is fixed to $\cost = 1$.}
  \label{fig:cost_vary}
\end{figure}

\input{tunable_problem}
\input{elastic_wave}

\section{Conclusion}

We have presented a framework that unifies many existing sampling methods used for multifidelity uncertainty quantification. In doing so, we have proven that the structure of existing recursive estimators fundamentally limits the amount of information that is extracted from multifidelity data sources. Regardless of the number of control variates these recursive schemes use, the estimators they produce only converge to the optimal control variate that uses a single model with known mean. This limitation can degrade the efficiency and accuracy of these estimators, especially in applications where the ability to obtain additional samples of the high-fidelity models is restricted.

To overcome this limitation, we propose new estimators that leverage all existing correlations among information sources. In other words, the formulations we propose break through the single-model barrier of existing recursive approaches and converge to the optimal multi-model control variate. In a number of numerical examples, we have shown that our proposed estimators are capable of achieving significant gains in variance reduction, especially in cases with more general model hierarchies that may lack high-levels of correlation and strong variance decay. %
In terms of numerical results, our monomial test problem shows strong benefit with significant variance reduction, both in the case of increasing numbers of high-fidelity evaluations and with fixed-numbers of high fidelity evaluations. Although the monomial test problem is quite simplified, the correlation matrix that it produces is representative of what may be encountered in practice. Since our algorithm only requires access to this correlation (or covariance) matrix, this indicates that there are potentially significant benefits to using our proposed approach. In order to gain additional insight regarding the effects of the correlations and cost ratios among models on the performance of control variate schemes, we introduce a parametric model problem that allows exploration of a large variety of scenarios. In particular, this example shows that the relative performance of the different approaches is related to the existence of a gap between OCV and OCV-1 and that this gap can be numerically exploited under certain cost conditions. As this gap shrinks, the advantage of the ACV strategies diminishes with respect to existing recursive strategies. And for the elastic equation example, we showed that the ACV approach was more robust to general multi-fidelity hierarchies that exhibit lower correlation levels. %

Future work will develop robust optimal sample allocation schemes for the ACV estimators. We conjecture that other optimization formulations can potentially lead to better exploitation of the variance gaps between OCV and OCV-1 estimators. Such optimal sample allocation schemes will most likely require sub-selecting groups of control variates that should be further evaluated and bypassing those whose uncertainty cannot be sufficiently reduced.
Another line of equally important work is assessing the variance introduced by estimating the covariance matrix. Such theory exists for the OCV, see e.g.~\cite{Nelson1990}, and we seek to extend it to the ACV.

\input{appendix}

\section*{Acknowledgements}
The authors thank Dr.~Laura Swiler and Dr.~Tim Wildey from Sandia National Laboratories for their insightful comments and suggestions regarding an draft of this manuscript. 

This work was fully supported by the DARPA EQUiPS project and partially  supported by the DOE SciDAC Advanced Scientific Computing Research (ASCR) program. Sandia National Laboratories is a multimission laboratory managed and operated by National Technology and Engineering Solutions of Sandia, LLC., a wholly owned subsidiary of Honeywell International, Inc., for the U.S. Department of Energy's  National Nuclear Security Administration under contract DE-NA-0003525. The views expressed in the article do not necessarily represent the views of the U.S. Department of Energy or the United States Government.

\vspace{-15pt}
\bibliographystyle{siamplain}
\bibliography{references}

\end{document}

%% file: macros_no_juq.tex
\newcommand{\argmin}{\operatornamewithlimits{arg\ min}}

\newcommand{\reals}{\mathbb{R}}
\newcommand{\posint}{\mathbb{Z}_{+}}
\renewcommand{\vec}[1]{\underline{#1}}
\newcommand{\bvec}[1]{\bm{#1}}
\newcommand{\mat}[1]{\bm{#1}}

\newcommand{\var}{{\normalfont \mathbb{V}\textrm{ar}}}
\newcommand{\varF}[1]{\var\left[ #1 \right] }
\newcommand{\VarF}[1]{\varF{#1}} 
\newcommand{\StDevF}[1]{{\normalfont \mathbb{V}\textrm{ar}^{1/2}\left[ #1 \right] }}

\newcommand{\cov}{{\normalfont \mathbb{C}\textrm{ov}}}
\newcommand{\covF}[1]{\cov\left[#1\right]}
\newcommand{\covm}{\bvec{C}}
\newcommand{\covv}{\bvec{c}}
\newcommand{\covvnorm}{\bvec{\bar{c}}}

\def\varspace{\Omega}

\newcommand{\nmodels}{M}
\newcommand{\nhf}{N}
\newcommand{\om}{z}

\newcommand{\sset}[1]{\bvec{\om}_{#1}}

\newcommand{\qoif}{Q} 
\newcommand{\qoi}{Q}   
\newcommand{\qspace}{\mathcal{Q}} 
\newcommand{\meanbase}{\mu}
\newcommand{\mean}[1]{\meanbase_{#1}}

\newcommand{\cv}[1]{\qoi_{#1}}
\newcommand{\est}[1]{\hat{\qoi}_{#1}}
\newcommand{\cvl}{{\normalfont \textrm{OCV}}}
\newcommand{\cvest}{\hat{\qoi}^{\cvl}}

\newcommand{\mll}{{\normalfont \textrm{W-RDiff}}}  
\newcommand{\mfmc}{{\normalfont \textrm{MFMC}}} 
\newcommand{\mlest}{\hat{\qoi}^{\mll}}

\newcommand{\mfmcest}{\hat{\qoi}^{\mfmc}}

\newcommand{\acvl}{{\normalfont \textrm{ACV}}}
\newcommand{\acvest}{\tilde{\qoi}}
\newcommand{\estm}[1]{\hat{\meanbase}_{#1}}
\newcommand{\rat}{r}
\newcommand{\cost}{w}

\newcommand{\cvw}{\alpha}
\newcommand{\optcvw}{\vec{\cvw}^{*}}
\newcommand{\mfmccvw}{\vec{\cvw}^{\mfmc}}

\newcommand{\optacvw}{\vec{\cvw}^{\acvl}}
\newcommand{\acvone}{{\normalfont \textrm{ACV-}IS}} 
\newcommand{\optacvone}{\vec{\cvw}^{\acvone}}
\newcommand{\acvoneest}{\est{}^{\acvone}}

\newcommand{\acvtwo}{{\normalfont \textrm{ACV-}MF}} 
\newcommand{\optacvtwo}{\vec{\cvw}^{\acvtwo}}
\newcommand{\acvtwoest}{\est{}^{\acvtwo}}
\newcommand{\acvtwoestK}{\est{K}^{\acvtwo}}

\newcommand{\Kacv}{{\normalfont \textrm{ACV-KL}}}
\newcommand{\optKacv}{\vec{\cvw}^{\Kacv}}
\newcommand{\Kacvest}{\est{}^{\Kacv}}

\newcommand{\cvdiff}[1]{\Delta_{#1}} 
\newcommand{\absr}[1]{r_{#1}}         
\newcommand{\rs}[1]{\tau_{#1}}        
\newcommand{\rcard}[1]{\eta_{#1}}        
\newcommand{\ccoeff}[1]{\rho_{#1}}  
\newcommand{\card}[1]{\normalfont \textrm{card}\!\left(#1\right)} 
\newcommand{\Diag}[1]{\normalfont \mathrm{diag}\left(#1\right)}

\newcommand{\covdiff}{\covF{\vec{\cvdiff{}},\vec{\cvdiff{}}}} 
\newcommand{\covdiffij}[2]{\covF{\cvdiff{#1},\cvdiff{#2}}} 
\newcommand{\covhc}{\covF{\vec{\cvdiff{}},\est{}}}              

\newcommand{\fmat}{{\bm{F}}}
\newcommand{\fmatone}{{\bm{F}^{(IS)}}}
\newcommand{\fmattwo}{{\bm{F}^{(MF)}}}


%% file: mlmc.tex
\subsection{Recursive difference estimators}\label{sec:mlmc}

While the MFMC estimator has a recursive nested structure, estimators that use recursive difference strategies as control variates are also found in the literature (see for instance \cite{Owen2013}). The classical difference estimator replaces the random variable $\cv{}$ with a ``difference'' random variable 
$\mu_1 + (\cv{}-\cv{1})$. Clearly both of these random variables have the same mean since $\mathbb{E}[\cv{1}] = \mu_1$, but this re-arrangement can result in a random variable with lower variance when $\VarF{\cv{}-\cv{1}}$ is smaller than $\VarF{\cv{}}$. 

In our case $\mu_1$ is unknown and must be estimated. The samples are partitioned $\sset{1} = (\sset{1}^1, \sset{1}^2)$ so that $\sset{1}^1=\sset{}$ are shared to compute the difference $\cv{}-\cv{1}$, and  $\sset{1}^2$ are used for computing $\estm{1}$. Then we have the estimator
\begin{equation}
 \acvest{}(\sset{},\sset{1}^2) = \estm{1}(\sset{1}^2) + \left( \est{}(\sset{}) - \est{1}(\sset{}) \right).  
\end{equation}
We can introduce additional low-fidelity models by proceeding in a recursive fashion
\begin{equation}
  \acvest{}(\sset{},\sset{1}^2,\sset{2}^2) = \est{2}(\sset{2}^2) + \left( \est{1}(\sset{1}^2)-\est{2}(\sset{1}^2) \right) + \left( \est{}(\sset{}) - \est{1}(\sset{}) \right).
  \label{eq:recur_diff_2model}
\end{equation}
We reinforce the point that the need to share samples arises because the difference between two random variables is required. If we re-arrange the terms, it is possible to write
\begin{equation}
 \acvest{}(\sset{},\sset{1}^2,\sset{2}^2) = \est{}(\sset{}) - \left( \est{1}(\sset{}) - \estm{1}(\sset{1}^2) \right) 
                                                 - \left( \est{2}(\sset{1}^2) - \estm{2}(\sset{2}^2) \right),
\end{equation}
which corresponds to a control variate estimator using two low-fidelity models and fixed weights $\cvw_i = -1$. Thus, one can think of a recursive difference estimator as a {\em control-variate approach with fixed control weights.} From an implementation perspective, this approach eliminates errors associated with imprecise estimation of the optimal weights. 
However, it introduces new sources of inefficiency by introducing weights that are, in the most general setting, sub-optimal. 

The estimator in~\eqref{eq:recur_diff_2model}, expanded from two to $M$ levels, defines the inner loop of a MLMC method~\cite{Giles2008,Cliffe2011}, where the inner loop corresponds to a fixed definition of the highest-fidelity model/most-resolved level\footnote{An outer loop may then adapt this most-resolved level to control bias error.}.
In the following definition, we consider this case of a fixed high-fidelity model and reintroduce the CV weights.
\begin{definition}[Weighted Recursive Difference (W-RDiff) Estimator]
  Let $\vec{\sset{}} = (\sset{},\sset{1},\ldots,\sset{\nmodels})$. Each ordered set of samples $\sset{i}$ is partitioned according to $\sset{i} = \sset{i}^1 \cup \sset{i}^2$, $\sset{i}^1 \cap \sset{i}^2 = \emptyset$, and $\sset{i}^1 = \sset{i-1}^2$ for $i = 1,\ldots,\nmodels$ where $\sset{0}^2 = \sset{1}^1 = \sset{}$ and $\sset{1}^2 \neq \emptyset$. Then the estimator RDiff is defined as
  \begin{equation}\label{eq:peer}
      \mlest(\vec{\cvw},\vec{\sset{}}) = \est{}(\sset{}) + \sum_{i=1}^{\nmodels} \cvw_i \left(\est{i}(\sset{i-1}^2) - \estm{i}(\sset{i}^2) \right).
  \end{equation}
\end{definition}

This W-RDiff estimator offers a simple modification to the traditional MLMC inner loop algorithm by introducing the optimal ACV weights~\eqref{eq:opt_acv_weights} without changing the MLMC sample strategy. A similar idea is also presented in \cite{Sukys2017}, where the authors recognize the limitation in using constant unitary weights withing the classical MLMC technique. In particular, the authors in \cite{Sukys2017} observed that the choice of unitary weights limits the variance reduction, and might lead to a variance increase if the correlation coefficients drop below a certain threshold\footnote{In~\cite{Sukys2017}, a sample allocation is also derived in closed form that begins by imposing a uniform cost redistribution across levels in an initial step. Our sample allocation scheme in \S\ref{sec:sample-allocation} will relax this assumption.}.
The variance reduction of the W-RDiff estimator is given by the following lemma.
\begin{lemma}[Variance reduction of W-RDiff]\label{lem:var_wRDiff}
  Let the W-RDiff estimator be defined as
  \begin{equation}
    \mlest(\vec{\cvw},\vec{\sset{}}) = \est{}(\sset{}) + \sum_{i=1}^{\nmodels} \cvw_i \left(\est{i}(\sset{i-1}^2) - \estm{i}(\sset{i}^2) \right), 
 \end{equation}
 with the sampling strategy: $\sset{i} = \sset{i}^1 \cup \sset{i}^2$, $\sset{i}^1 \cap \sset{i}^2 = \emptyset$, and $\sset{i}^1 = \sset{i-1}^2$ for $i = 1,\ldots,\nmodels$ where $\sset{0}^2 = \sset{1}^1 = \sset{}$, $\sset{1}^2 \neq \emptyset$. The variance of the estimator is
   \begin{equation}
    \varF{\mlest(\vec{\cvw})} = \varF{\est{}}\left(1 - R_{\mll}^2\right),
  \end{equation}
  where 
  \begin{equation}\label{eq:Rmll}
    R_{\mll}^2 = - \cvw_1^2 \rs{1}^2 - 2 \cvw_1 \ccoeff{1} \rs{1} - \cvw_\nmodels^2 \frac{\rs{\nmodels}}{\rcard{M}} - \sum_{1=2}^\nmodels \frac{1}{\rcard{i-1}}
                                                          \left( \cvw_i^2 \rs{i}^2 + \rs{i-1}^2 \rs{i-1}^2 - 2 \cvw_i \cvw_{i-1} \ccoeff{i,i-1} \rs{i} \rs{i-1} \right),
  \end{equation}
  where $\rs{i} = \frac{ \StDevF{\qoi_i} }{ \StDevF{\qoi}}$ is the ratio of the standard deviations, $\ccoeff{ij}$ is the Pearson correlation coefficient between  $(\cv{i},\cv{j})$, and $\rcard{i}= |\sset{i}^2| / \nhf$ is the ratio between the cardinality of the sets $\sset{i}^2$ and $\sset{}$.  
 \end{lemma}
 The proof is provided in \ref{sec:proof_wRDiff}, and leads directly to a bound on the maximum variance reduction.
\begin{theorem}[Maximum variance reduction of RDiff]\label{th:MLMC_varbound}
    The variance reduction of RDiff is bounded above by that of the optimal single CV, i.e.,
  \begin{equation}
    R_{\mll}^2 \leq \ccoeff{1}^2.
  \end{equation}
\end{theorem}
\begin{proof}
We want to understand the behavior of $R_{\mll}^2$, given in Equation~\eqref{eq:Rmll}, in the limit of infinite low-fidelity data $\rat_i \rightarrow \infty$ for $i=1,\dots, \nmodels$ (or equivalently $\rcard{i} \rightarrow \infty$). This limit implies
\begin{equation} \label{eq:proof_mlmc}
\lim_{\vec{\absr{}} \to \infty} R_{\mll}^2 = - \cvw_1^2 \rs{1}^2 - 2 \cvw_1 \rho_1 \rs{1},
\end{equation}
and the maximum of this function is obtained for $\cvw_1 = - \ccoeff{1}/\rs{1}$ for which
$\lim_{\vec{\absr{}} \to \infty} R_{\mll}^2 = \ccoeff{1}^2.$
\end{proof}

The convergence of the $R^2$ term of this estimator to $\ccoeff{1}^2$ requires the ability to set the weight $\cvw_1$ appropriately. Indeed we can see from~\eqref{eq:proof_mlmc} that if $\cvw_1$ is not set appropriately, we may not obtain the optimal variance reduction. Moreover, it is the recursive nature of this estimator that limits the maximum possible variance reduction. %

%% file: tunable_problem.tex
\subsection{A parametric model problem}

We now introduce a parametric model problem that enables us to quantify the performance of the different algorithms under several scenarios. We consider three two-dimensional functions; the first describes the high-fidelity quantity of interest and the next two serve as control variates:
\begin{equation}\label{eq:tunable_def}
\begin{split}
 Q   &= A \left( \cos\theta \, x^5 + \sin\theta \, y^5 \right), 
\quad  Q_1 = A_1 \left( \cos\theta_1 \, x^3 + \sin\theta_1 \, y^3 \right), 
\quad  Q_2 = A_2 \left( \cos\theta_2 \, x + \sin\theta_2 \, y \right), 
\end{split}
\end{equation}
where $x,y\sim\mathcal{U}(-1,1)$ and all $A$ and $\theta$ coefficients are real. We choose to set $A=\sqrt{11}$, $A_1=\sqrt{7}$ and $A_2=\sqrt{3}$ to obtain unitary variance for each model. This choice of unit variance for each model reduces the number of degrees of freedom in the problem parameterization since the correlation and covariance matrices are identical. Specifically, the analytic correlation/covariance matrix is given in Table~\ref{tab:corr_tunable}.
\begin{table}
  \centering
  {\footnotesize
  \begin{tabular}{|c|ccc|}
    \hline
             & $\qoi$   & $\cv{1}$ & $\cv{2}$   \\
    \hline
    $\qoi$   & 1        & $AA_1/9 \left(\sin\theta \sin\theta_1 + \cos\theta \cos\theta_1\right)$  & $AA_2/7\left(\sin\theta \sin\theta_2 + \cos\theta \cos\theta_2\right)$  \\
    $\cv{1}$ & \textit{sym} & 1        &  $A_1A_2/5\left(\sin\theta_1 \sin\theta_2 + \cos\theta_1 \cos\theta_2\right)$ \\
    $\cv{2}$ & \textit{sym} & \textit{sym} &  1        \\
   \hline
  \end{tabular}
  }
  \caption{Correlation/covariance matrix for Equation~\eqref{eq:tunable_def} with $A = \sqrt{11}$, $A_1 = \sqrt{7}$, and $A_2 = \sqrt{3}$.}
  \label{tab:corr_tunable}
\end{table}
To further reduce the number of degrees of freedom, 
we fix $\theta = \pi/2$ and $\theta_2 = \pi/6$ and let $\theta_1$ vary uniformly in the bounds $\theta_2 < \theta_1 < \theta$.  Each particular value of $\theta_1$ induces a different correlation $\ccoeff{1}$ between $\cv{1}$ and $\qoi$ and a different correlation $\ccoeff{12}$ between $\cv{1}$ and $\cv{2}$, whereas the correlation $\ccoeff{2}$ between $\cv{2}$ and $\qoi$ remains fixed. These correlations are reported for different settings of $\theta_1$ in Figure~\ref{fig:correlation_tunable_theta1}. We follow the optimization formulation specified in Section~\ref{sec:sample-allocation}; however, in this problem we have found we do not need to impose the $\absr{i-1} > \absr{i}$ constraint in the optimization.

\begin{figure}
  \centering
  \begin{subfigure}[b]{0.25\textwidth}
    \centering
    \includegraphics[width=\textwidth,clip,trim=0 0 25 0]{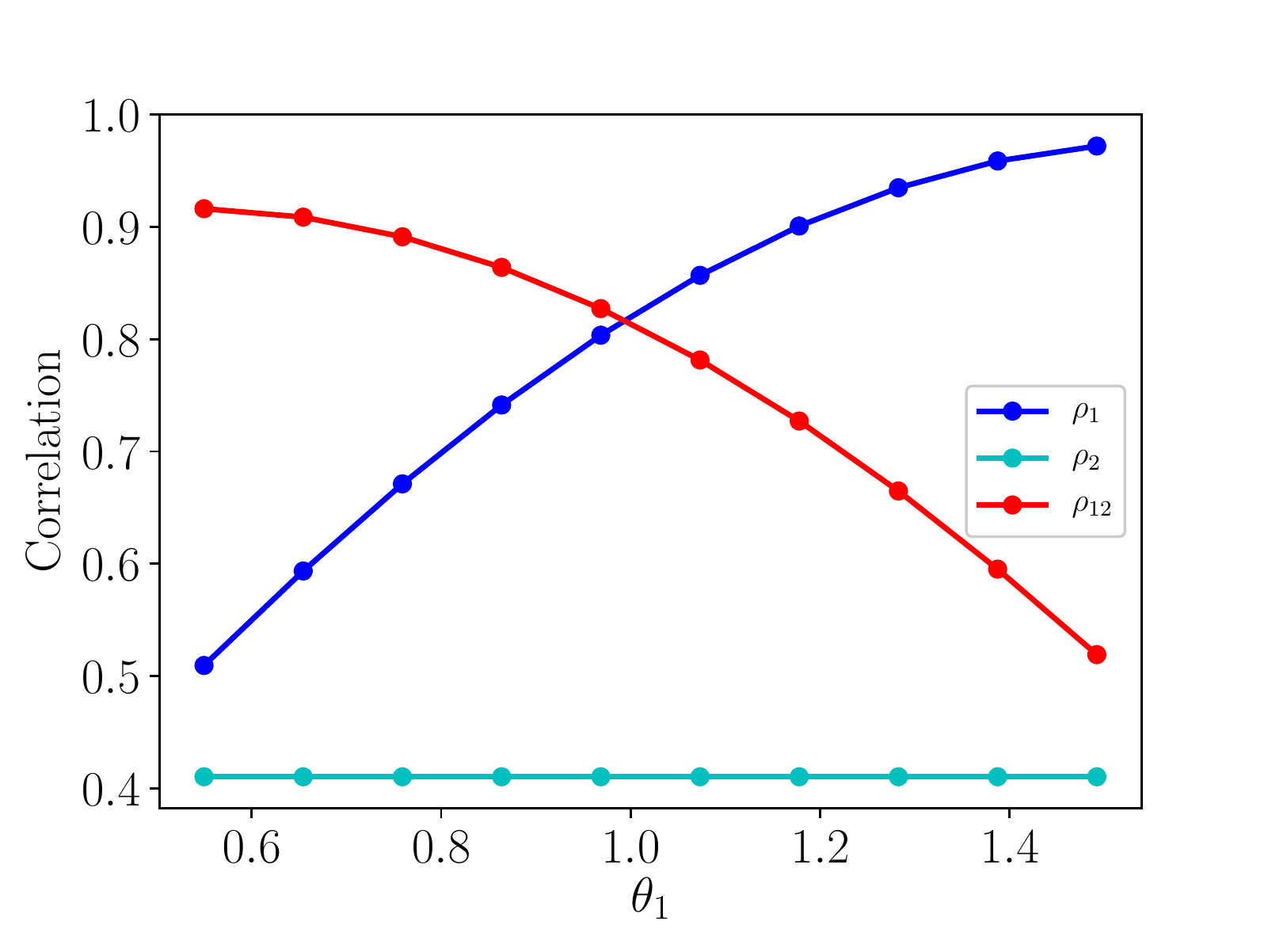}
    \caption{Correlations}
    \label{fig:correlation_tunable_theta1}
  \end{subfigure}
  ~
  \begin{subfigure}[b]{0.24\textwidth}
    \centering
    \includegraphics[width=\textwidth,clip,trim=0 0 10 0]{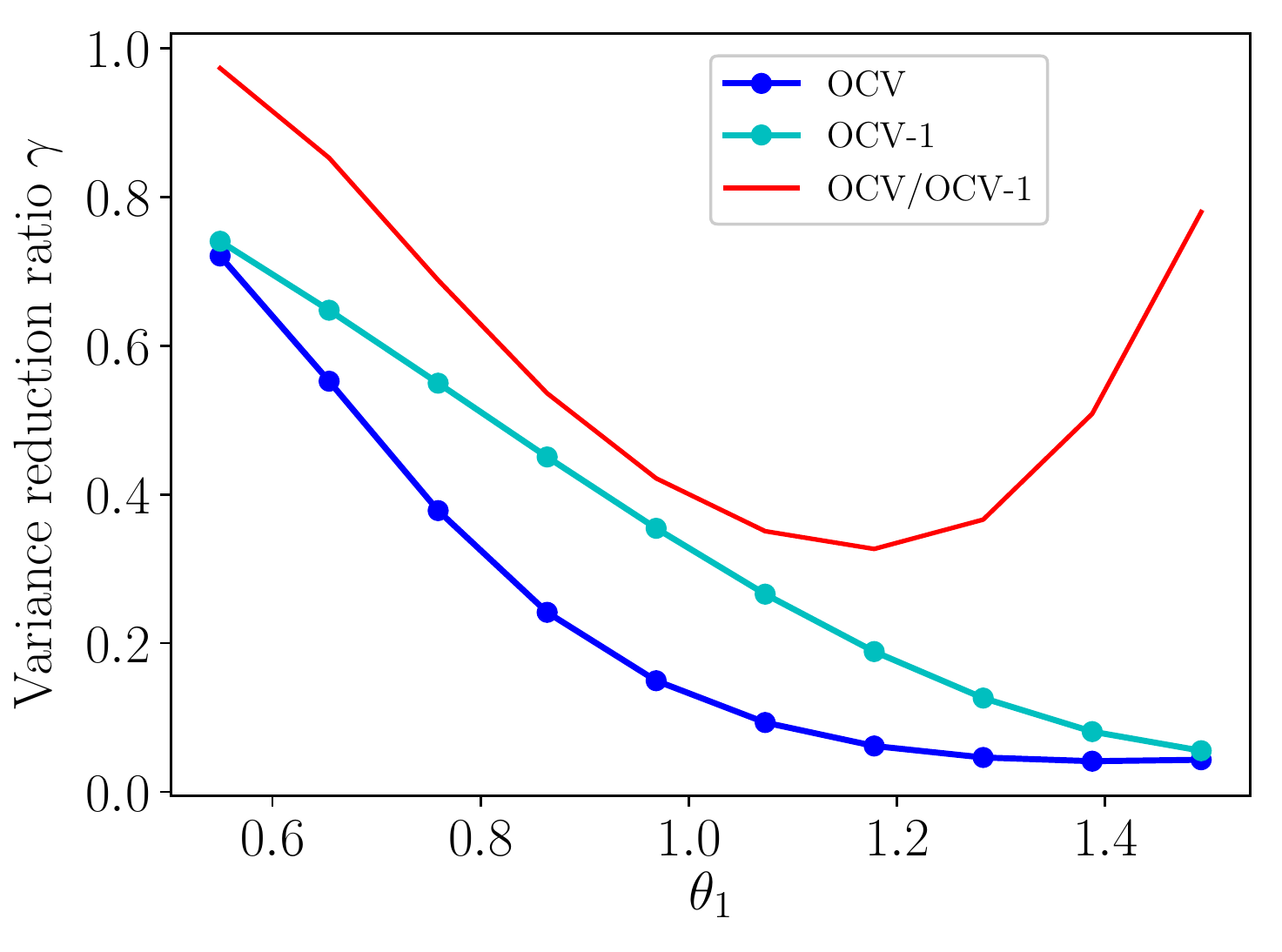}
    \caption{Var. reduction ratios}
    \label{fig:tunable_gap}
  \end{subfigure}  
  \caption{Correlation (a) and var. reduction gap (b) for Equation~\eqref{eq:tunable_def} with $\theta=\pi/2$, $\theta_2=\pi/6$ and $\theta_2<\theta_1<\theta$.}
\end{figure}

First we demonstrate that the variance reduction ratio of OCV is larger than the one obtained using only a single control variate OCV-1. The ratio of the variance of OCV-1 to OCV is reported in Figure~\ref{fig:tunable_gap}. In Figure~\ref{fig:tunable_gap} we also report the variance reduction obtained by OCV and OCV-1 with respect to MC. The greatest gap between OCV and OCV-1 occurs when $\theta_1$ is approximately $1.2$ (minimum of red curve).

Next we consider the effect of different cost relationships among the three models. For these purposes, we assign a relative cost of $1$ for $\qoi$, $1/w$ for $\cv{1}$ and $1 / w^2$ for $\cv{2}$. For an equivalent total cost of 100 runs of $\qoi$, we consider the effects of $w$ on the performance of each estimator. These results are given in Figure~\ref{fig:tunable_variable_reduction}.

Several interesting features are present in these results. First, the ACV-KL estimator generally outperforms all other estimators across the range of cost ratios enumerated in plots (a) through (e), and across the range of $\theta_1$ from $\pi/6$ to $\pi/2$.  The ACV-KL is nearly identical to ACV-MF since large deviations in $(K,L)$ cannot occur when there are only 2 control variates. In most of the scenarios, these two estimators also outperforms OCV-1. Second, we see a gradual convergence of the recursive estimators to OCV-1 and optimal estimators (ACV-IS, ACV-MF, ACV-KL) to OCV as $w$ increases. Qualitatively, the relationship of the performance of these two groups of estimators, with $\theta_1$, is very similar to their OCV-1 and OCV counterparts. In particular MFMC, W-RDiff, and RDiff, seem to decay almost linearly, just like OCV-1; and ACV-IS, ACV-MF, and ACV-KL seem to plateau with a similar shape to OCV. Finally, we see that the greatest advantage of our proposed estimators correspond to the cases of $\theta_1$ for which there is the greatest gap between OCV-1 and OCV. As this gap shrinks, our advantage decays.

\begin{figure}[h!]
  \centering
  \begin{subfigure}[b]{0.31\textwidth}
    \includegraphics[width=\textwidth,clip,trim=0 10 20 20]{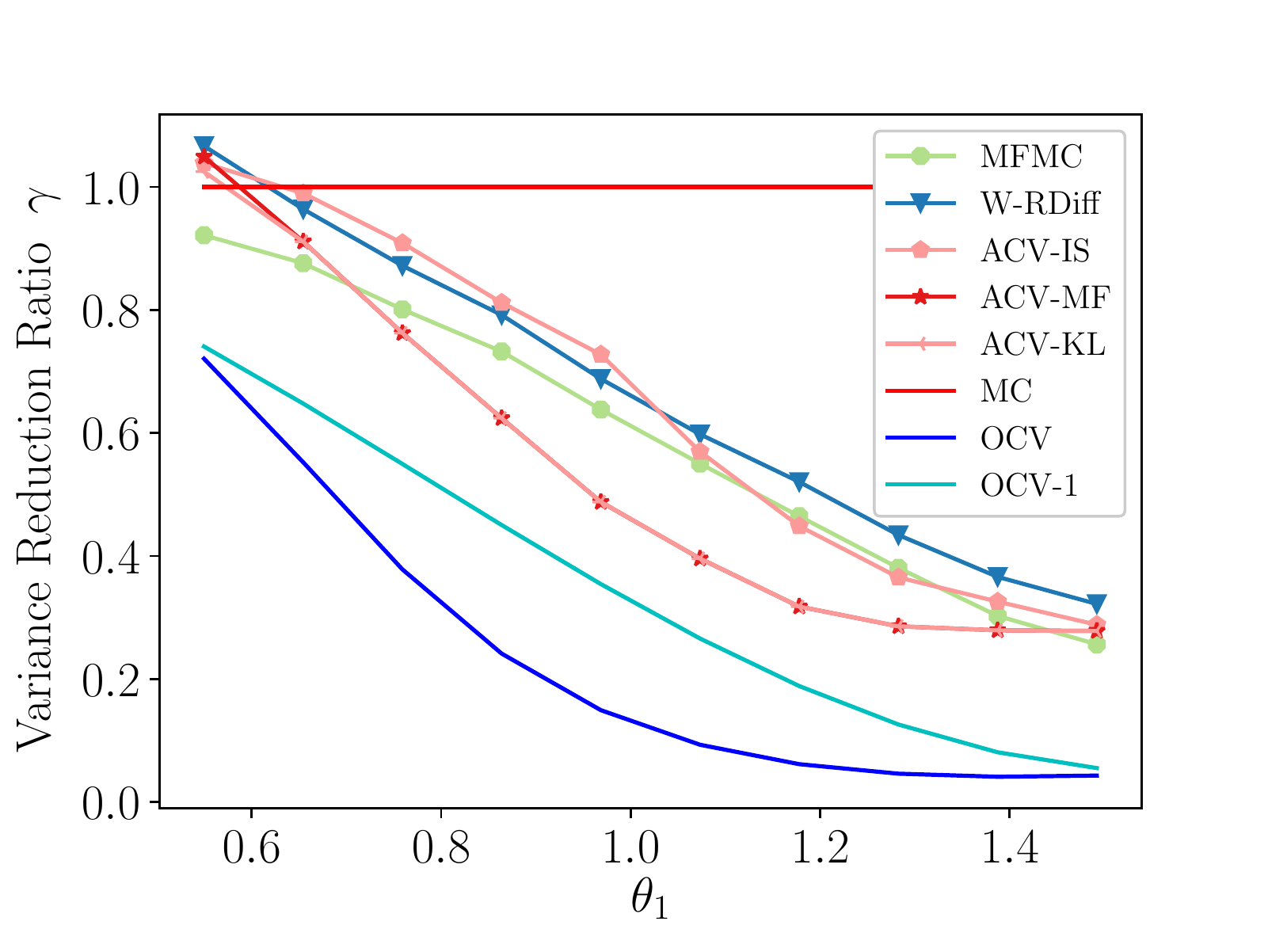}
    \caption{$w = 10$}
  \end{subfigure}
  ~
  \begin{subfigure}[b]{0.3\textwidth}
    \includegraphics[width=\textwidth,clip,trim=20 10 20 20]{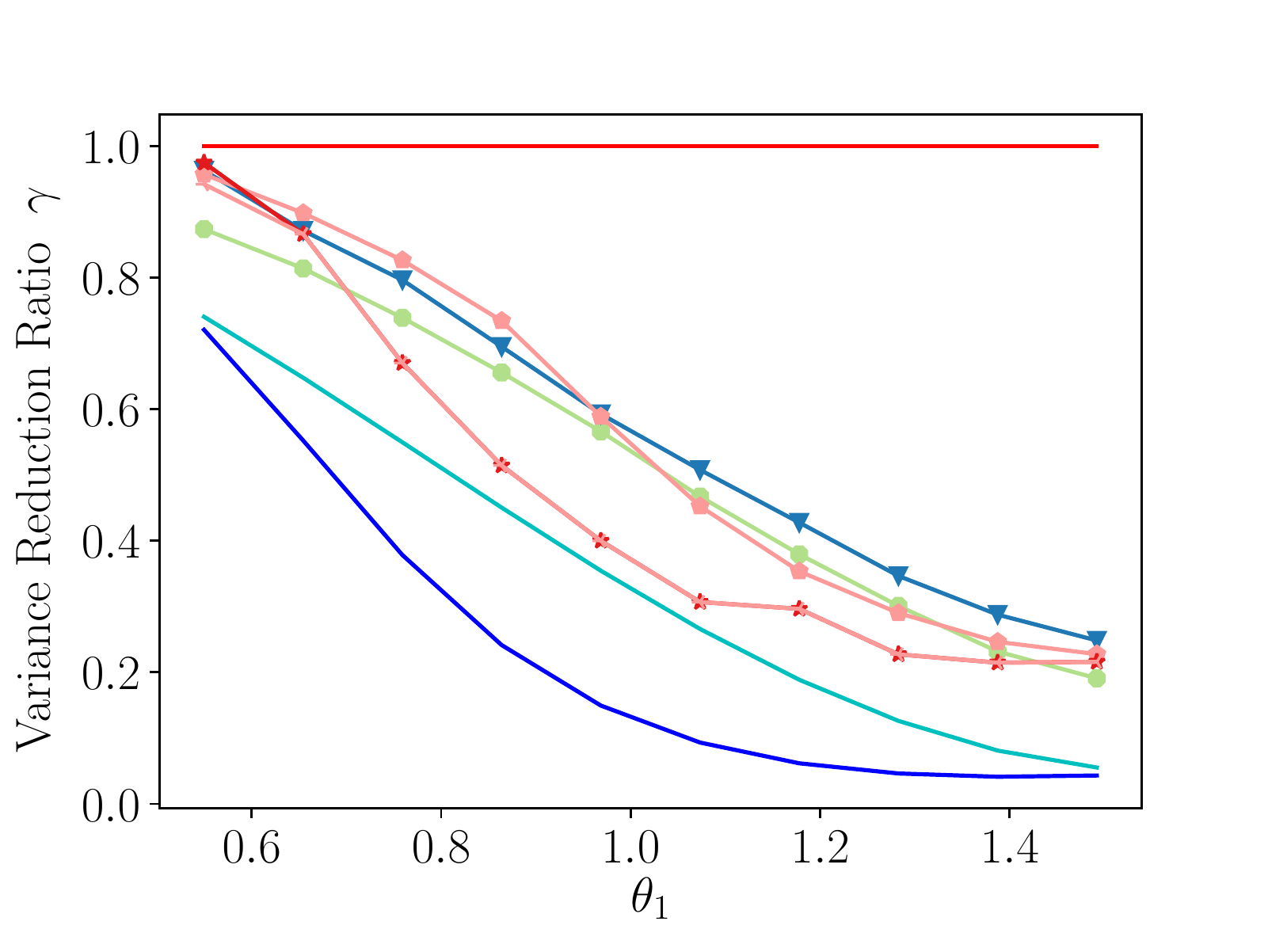}
    \caption{$w = 15$}    
  \end{subfigure}
  ~
  \begin{subfigure}[b]{0.3\textwidth}
    \includegraphics[width=\textwidth,clip,trim=20 10 20 20]{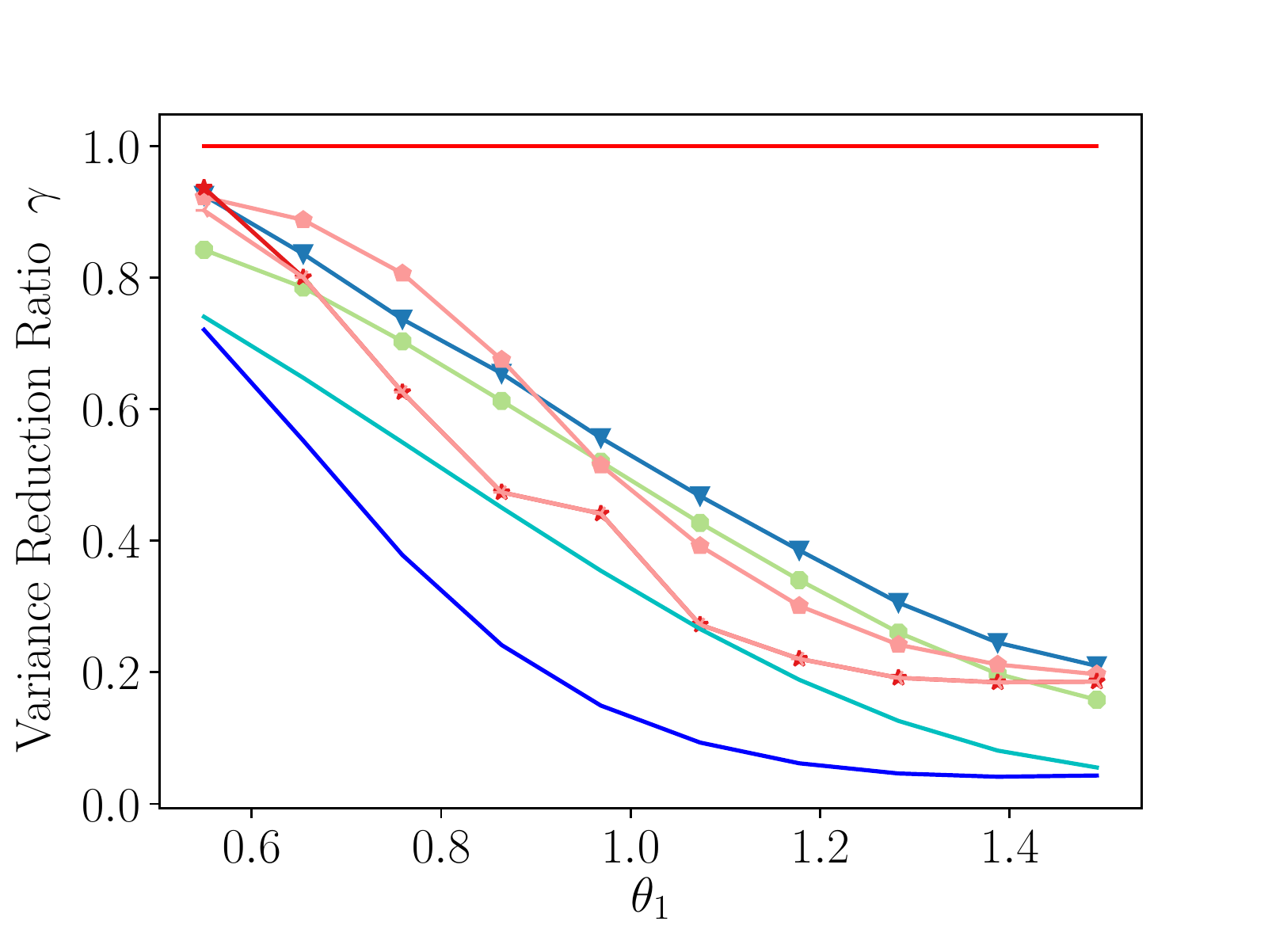}
    \caption{$w = 20$}    
  \end{subfigure}

  \begin{subfigure}[b]{0.31\textwidth}
    \includegraphics[width=\textwidth,clip,trim=0 10 20 20]{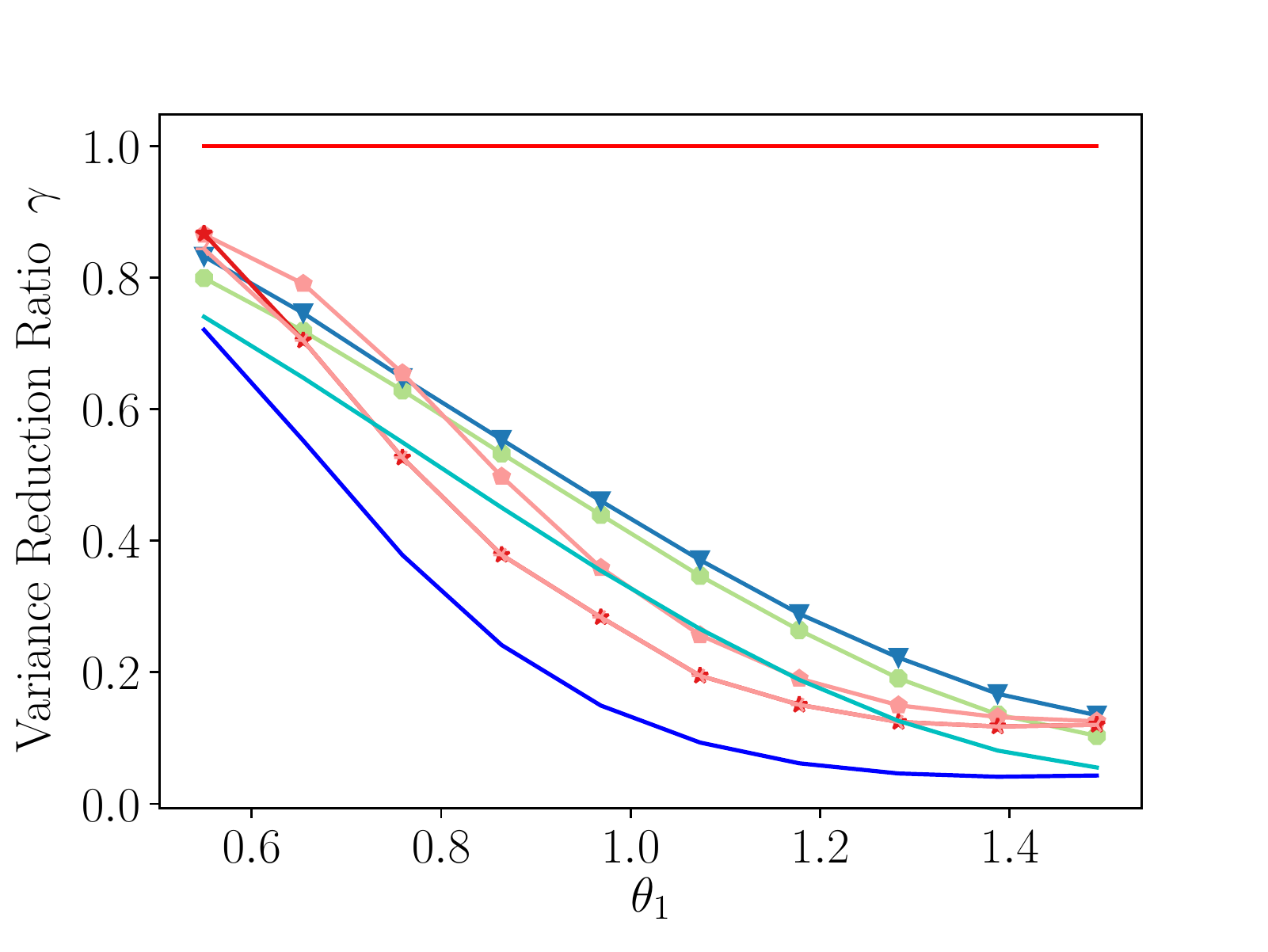}
    \caption{$w = 50$}    
  \end{subfigure}
  ~
  \begin{subfigure}[b]{0.3\textwidth}
    \includegraphics[width=\textwidth,clip,trim=20 10 20 20]{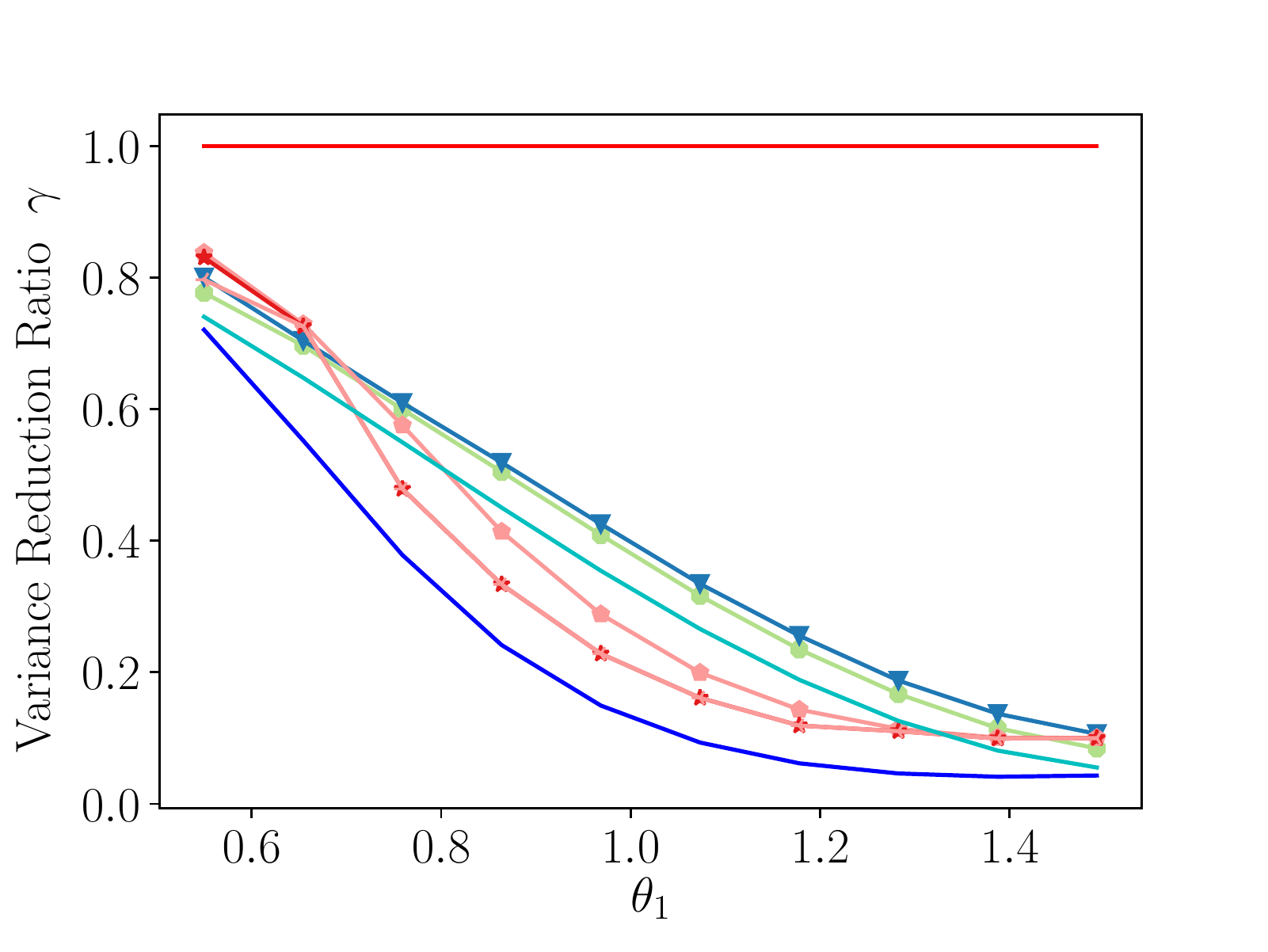}
    \caption{$w = 100$}
  \end{subfigure}
  ~
  \begin{subfigure}[b]{0.3\textwidth}
    \includegraphics[width=\textwidth,clip,trim=20 10 20 20]{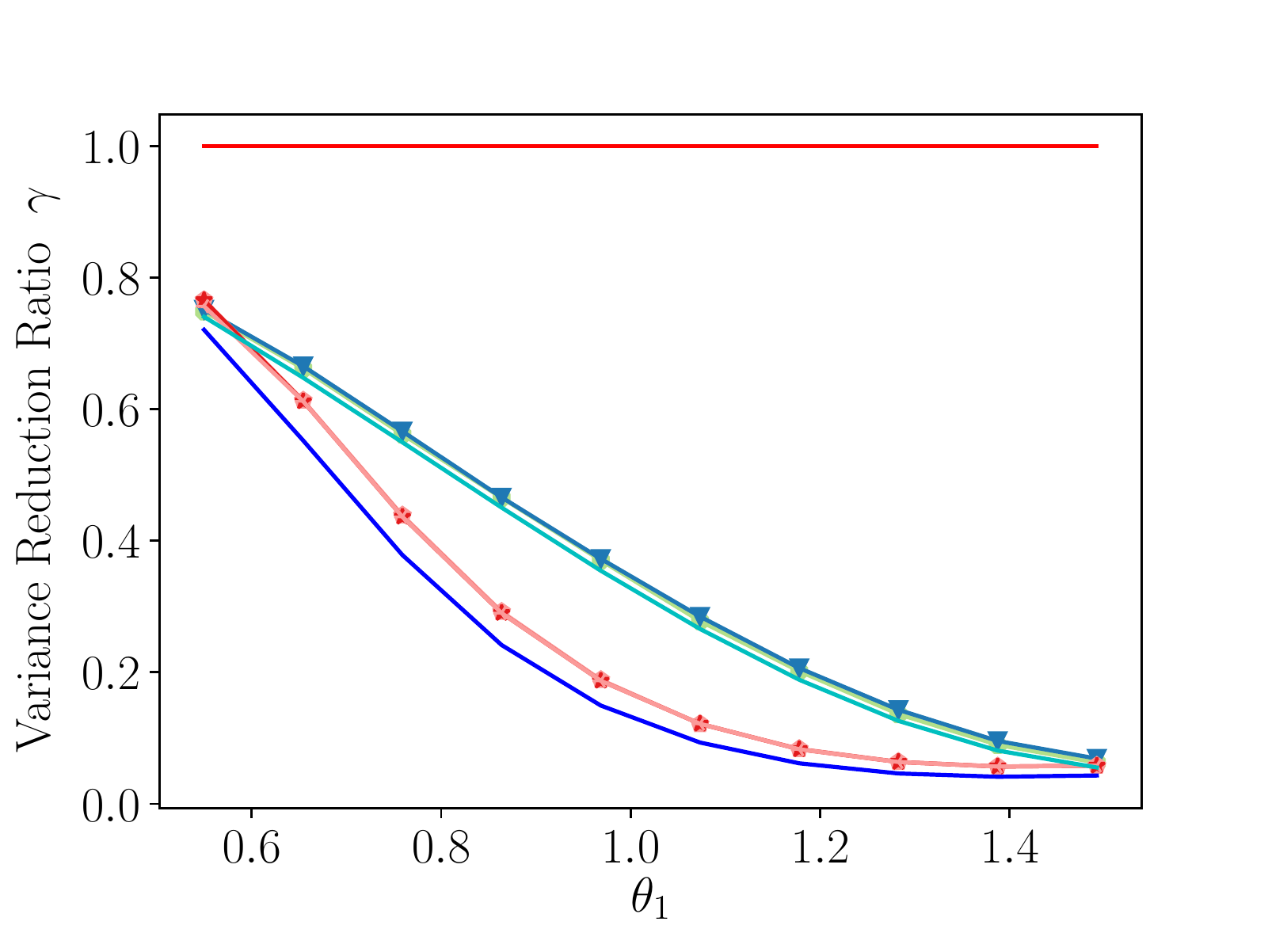}
    \caption{$w = 1000$}
  \end{subfigure}
  \caption{Variance reduction for cost ratios of $[1, 1/w, 1/w^2]$ for $\qoi$, $\cv{1}$, and $\cv{2}$ for the tunable problem Equation~\eqref{eq:tunable_def}. As the cost of the lower fidelity models becomes less expensive, the non-recursive ACV methods converge to the optimal variance reduction given by OCV and the recursive approaches converge to the variance reduction of OCV-1.}
  \label{fig:tunable_variable_reduction}
\end{figure}

In Figure~\ref{fig:mfmc_acvkl_contour}, we directly compare the ACV-KL estimator, which pursues OCV performance, with the MFMC estimator, which pursues of OCV-1 performance, over the range of $\theta_1$ for an expanded range of $w$ values. For the central portion of this plot, the ACV-KL estimator outperforms MFMC. MFMC gains an advantage for the more extreme values of $\theta_1$, and these are the same regions for which the OCV-1 and OCV gap is small according to Figure~\ref{fig:tunable_gap}.

\begin{figure}
  \centering
  \includegraphics[width=0.35\textwidth,clip,trim=5 5 5 35]{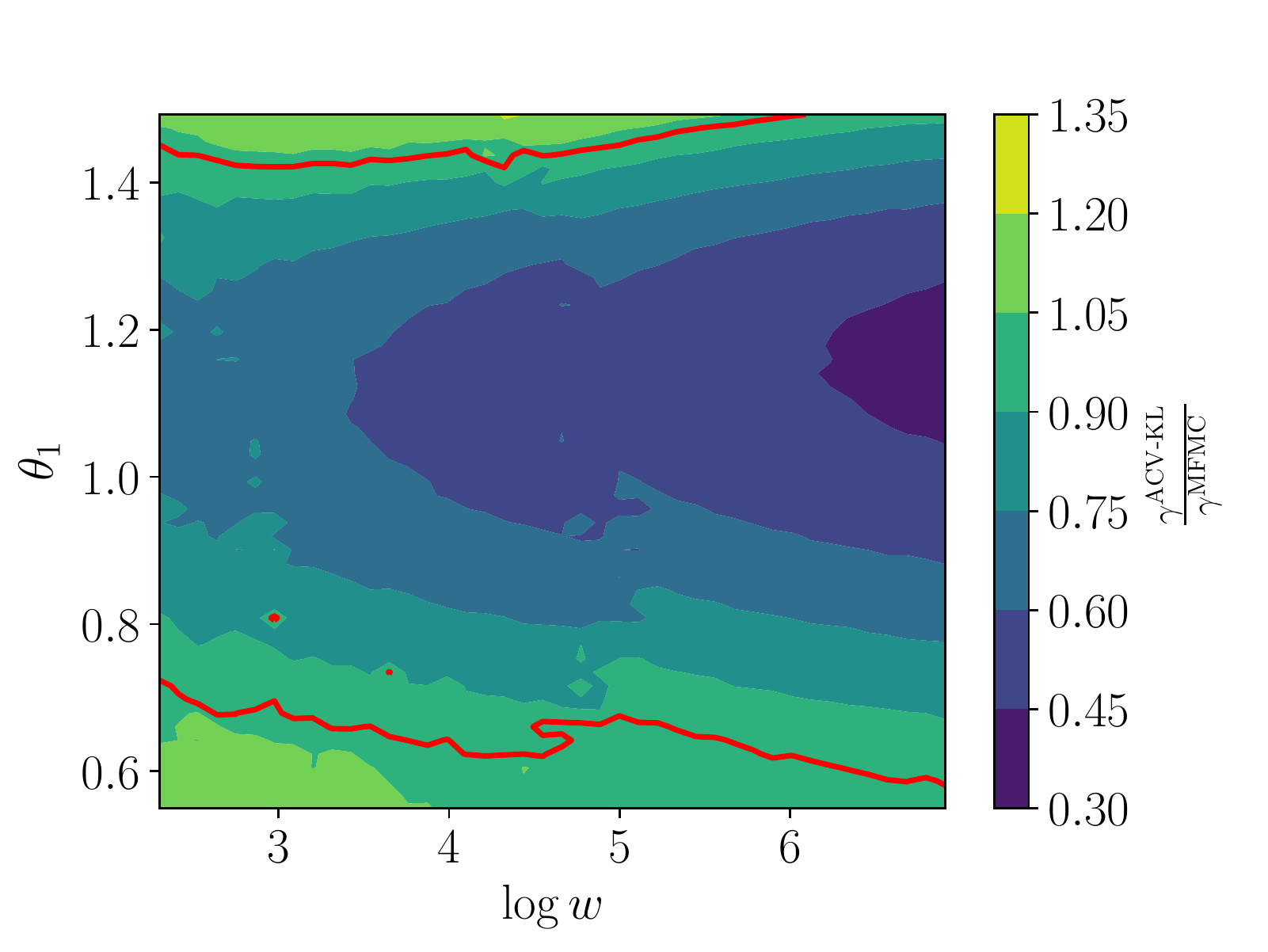}
  \caption{Comparison of the variance reduction between ACV-KL and MFMC over a range of $\theta_1$ and $w$. ACV-KL outperforms MFMC over a wide range of settings. The red line indicates the level-1 contour where the estimators have equal performance.}
  \label{fig:mfmc_acvkl_contour}
\end{figure}

\subsubsection{Weighted vs. unweighted estimators}
While the previous comparison allowed us to study the effect of how samples are distributed for the different estimators, we now demonstrate the improved performance of using non-fixed control-variate weights. In particular, we compare the unweighted recursive difference estimator $\cvw_i = -1$ with its weighted counterpart.  This comparison is applicable to many algorithms found in the literature because the unweighted recursive difference is used within the inner loop of MSE-adaptive MLMC algorithms, as discussed previosly, and might serve to provide a future direction for MLMC research to include weighted estimators. Our comparison is shown for the two extremes  $w=10$ and $w=1000$ in Figure~\ref{fig:weight_vs_unweight}. Here we show that the weighted estimator has generally lower variance than the unweighted estimator, and that it is able to outperform Monte Carlo even in regimes where the unweighted estimator cannot.

\begin{figure}
  \centering
  \begin{subfigure}[b]{0.3\textwidth}
    \centering
    \includegraphics[width=\textwidth,clip,trim=0 0 20 40]{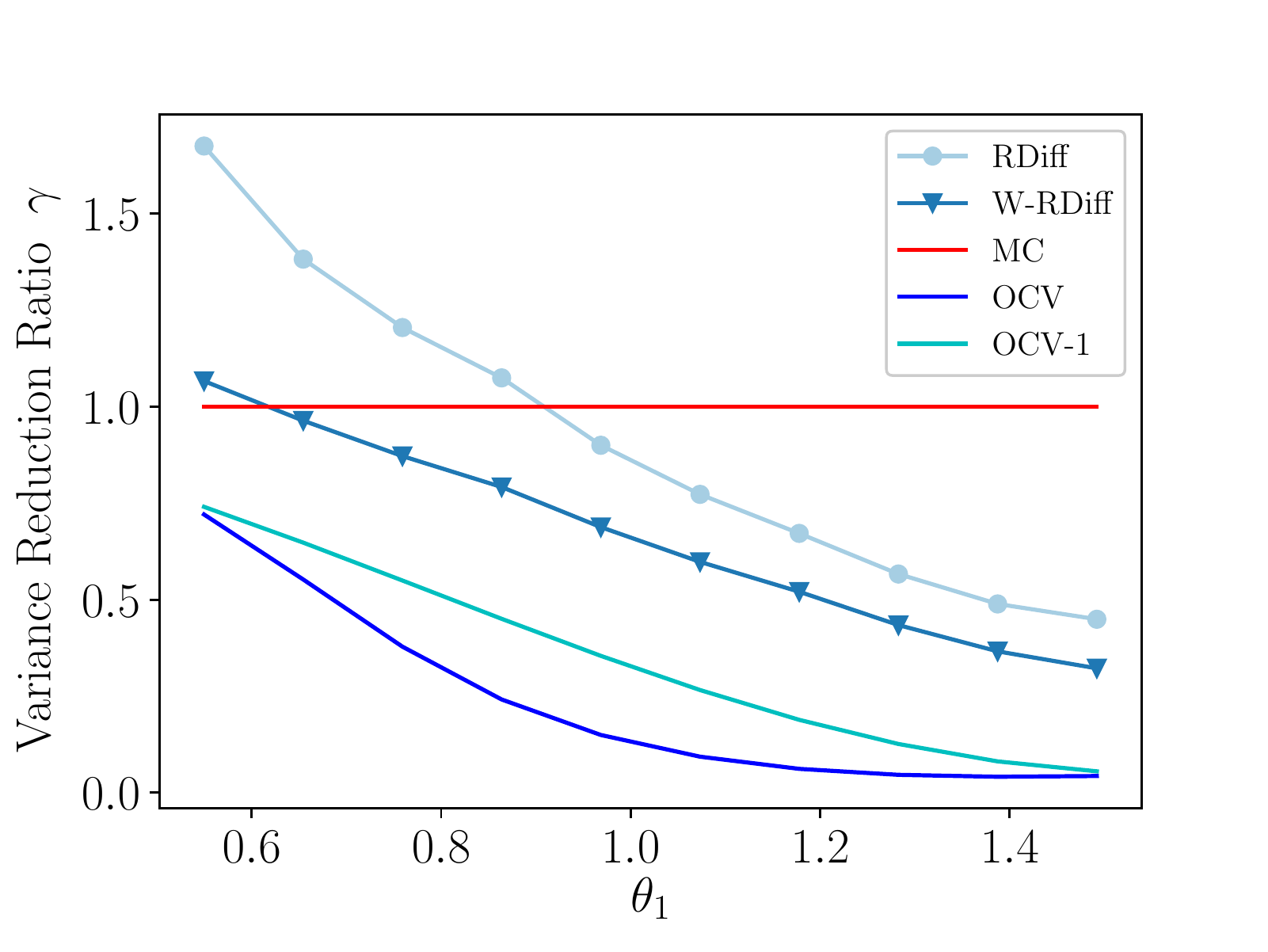}
    \caption{$w = 10$}
  \end{subfigure}
  ~
  \begin{subfigure}[b]{0.29\textwidth}
    \centering
    \includegraphics[width=\textwidth,clip,trim=21 0 20 40]{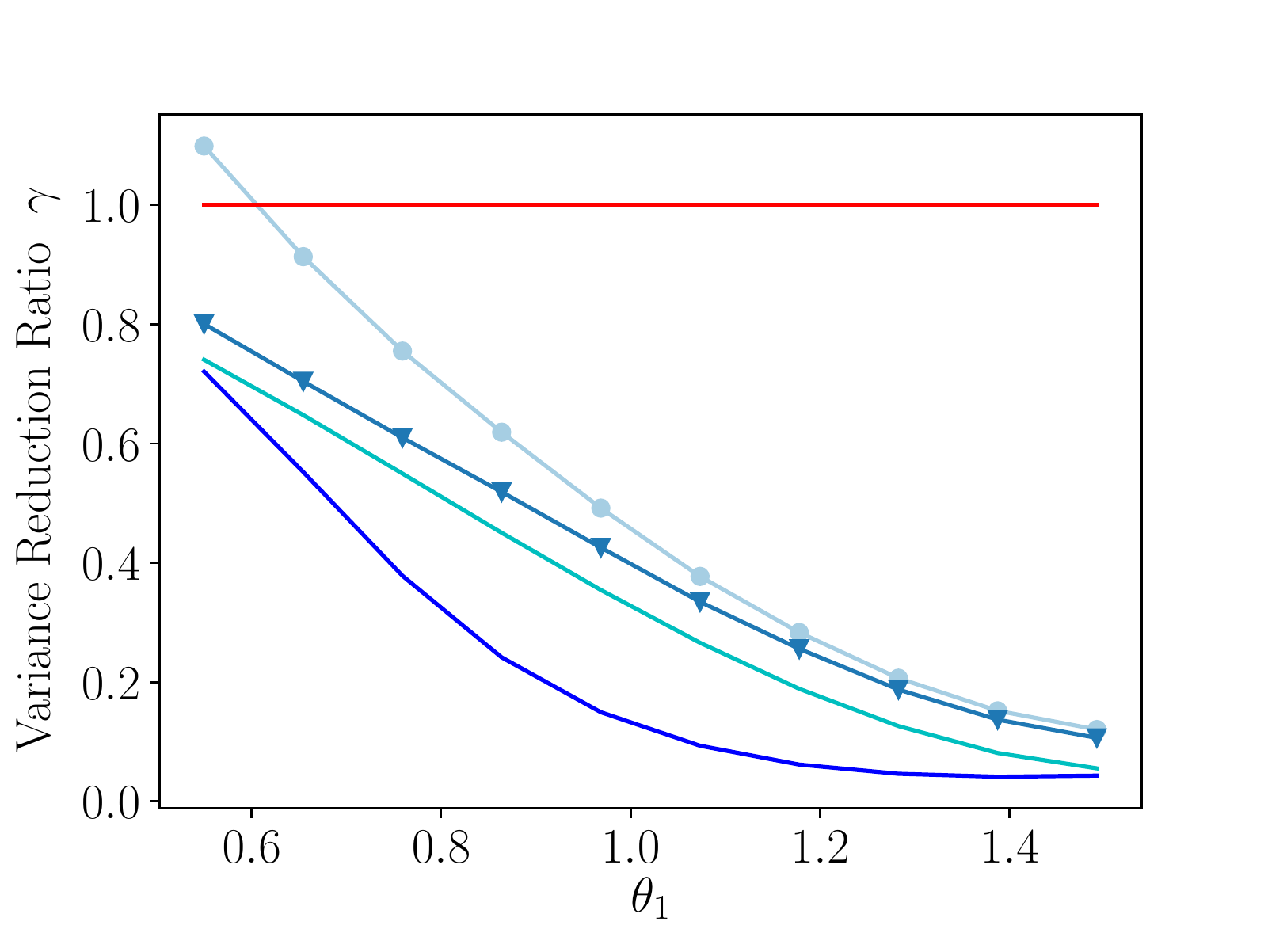}
    \caption{$w = 1000$}
  \end{subfigure}
  \caption{A comparison of the weighted and unweighted recursive difference estimators for the tunable Problem Equation~\eqref{eq:tunable_def}. The weighted estimator outperforms the unweighted alternative $\cvw_i = -1$ and is able to beat Monte Carlo when the unweighted cannot.}
  \label{fig:weight_vs_unweight}
\end{figure}

\subsubsection{Estimating control-variate weights}
Next we perform an experiment to assess a simple estimator for use when the covariance matrix is not known, but must be estimated from data. We provide results for estimator performance by tabulating the statistics of over 5000 repeated realizations of an estimator that: (1) uses 20 pilot samples to estimate the covariance matrix and (2) uses the remaining cost to perform optimal allocation. Note that the total cost of 100 is identical to that used for the known covariance estimator. In this way, we account for the cost of the pilot samples. Furthermore, we now include the unweighted recursive difference estimator ($\cvw_i = -1$) to assess whether performing this additional estimation could make it more efficient to prespecify the weight (as in the inner loop of MLMC schemes). These results are presented in Figure~\ref{fig:approx_cvw}.

First the estimators are indeed less efficient when the covariance must be estimated. Second, we notice that even when the weights must be estimated, the performance of weighted schemes is generally better than the unweighted recursive difference estimator. Finally, we note that the implementation of the two-step process for these results may in itself be sub-optimal. Future work will determine whether the estimation of the CV weights can be incorporated into the objective function and a more comprehensive approximate control variate estimator can be constructed. We note that some theory has been developed to address suboptimal estimation of the covariance matrix (and therefore $\cvw_i$) in the context of traditional control variates~\cite{Nelson1990}, but not for the approximate control variate context.

\begin{figure}
  \centering
  \begin{subfigure}[b]{0.31\textwidth}
    \centering
    \includegraphics[width=\textwidth,clip,trim=0 0 20 0]{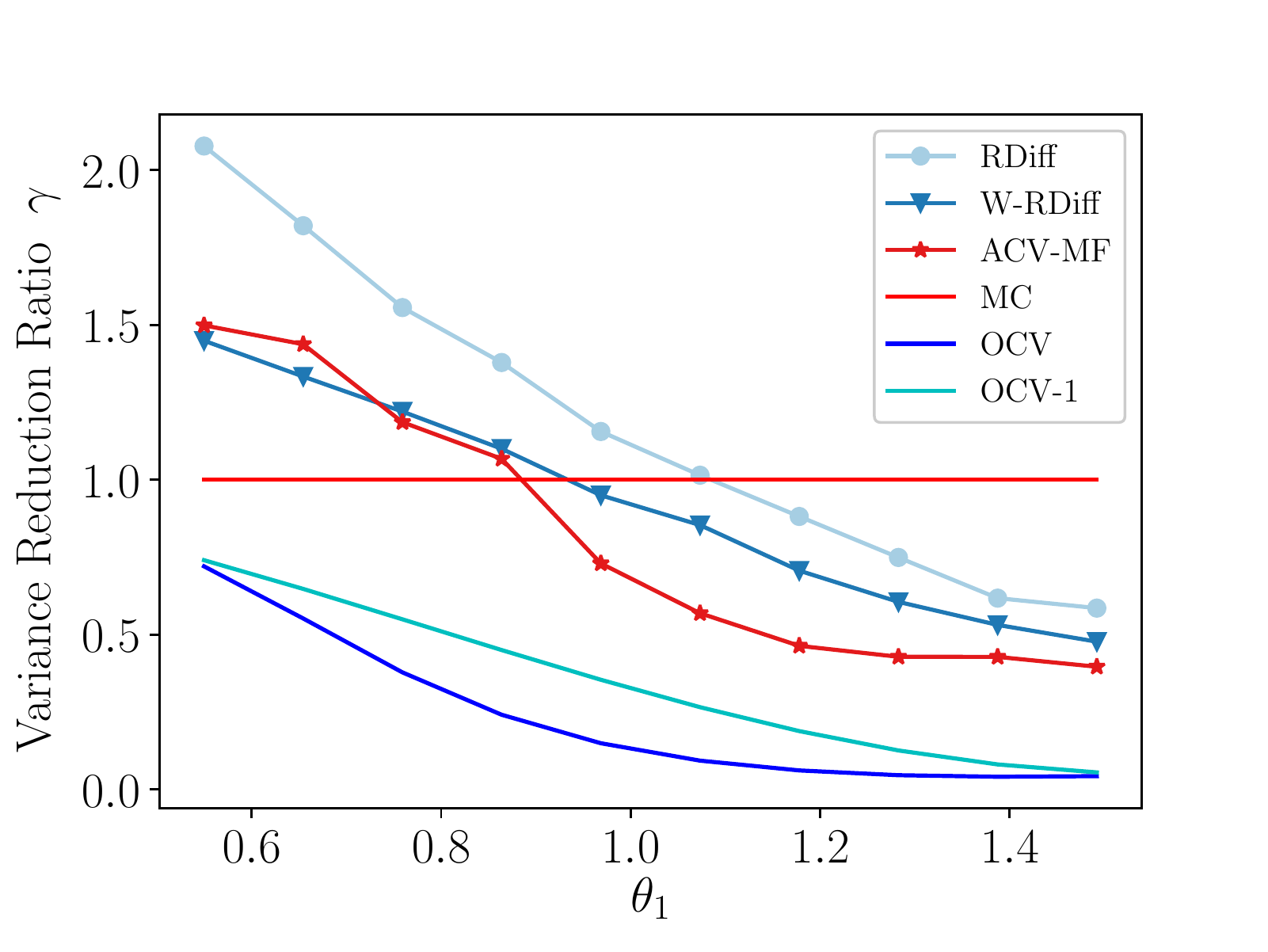}
    \caption{$w = 10$}
  \end{subfigure}
  ~
  \begin{subfigure}[b]{0.3\textwidth}
    \centering
    \includegraphics[width=\textwidth,clip,trim=20 0 20 0]{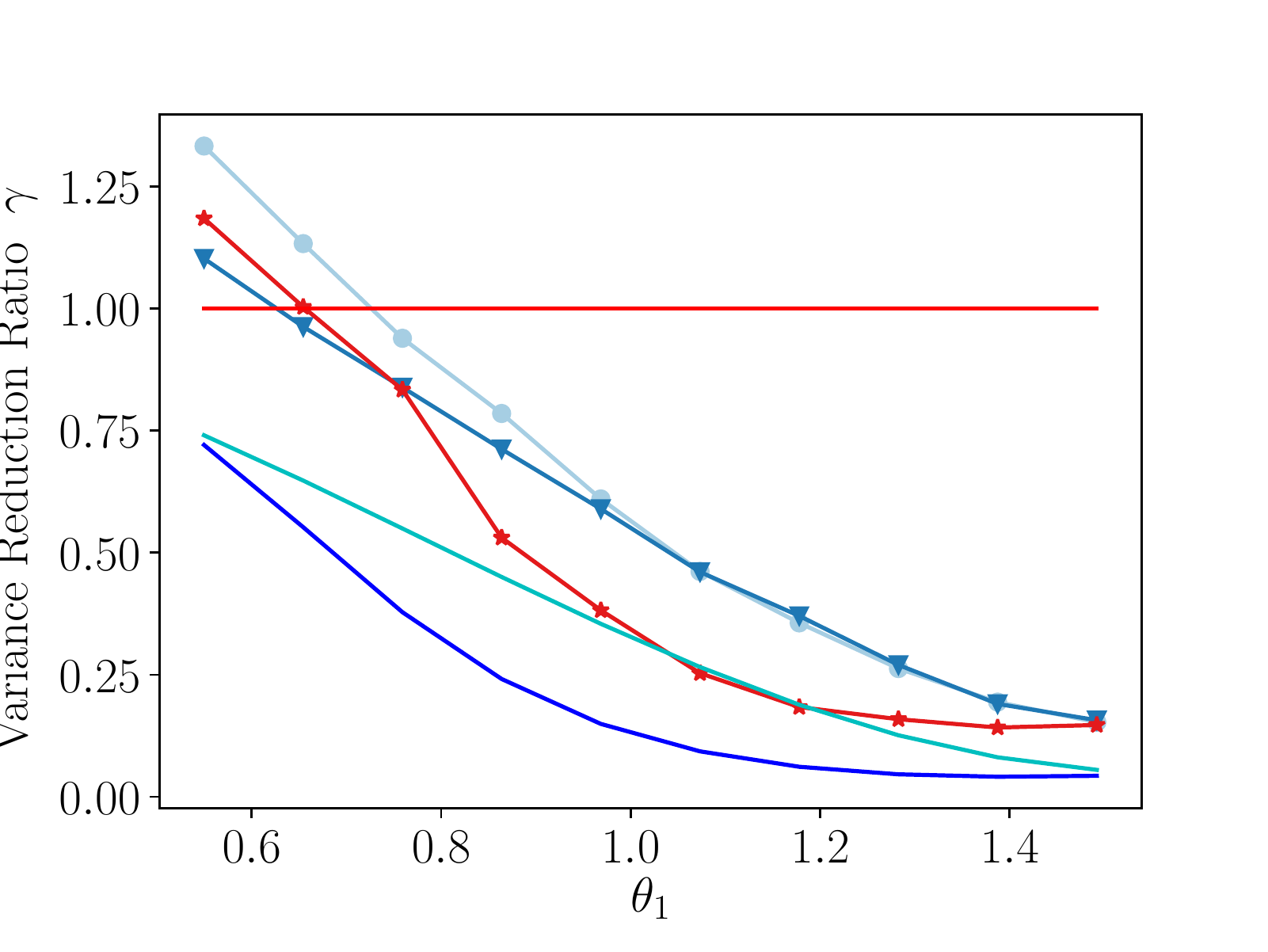}
    \caption{$w = 100$}
  \end{subfigure}
  \caption{Reduction in efficiency of the estimator when the covariance matrix must be estimated. ACV-MF still generally outperforms the alternative estimators; and the weighted recursive difference outperforms its unweighted counterpart. Note that any bias introduced by the estimation of $\cvw_i$ is neglible --- estimated to be $\mathcal{O}(10^{-3})$ from 5000 estimator realizations.}
  \label{fig:approx_cvw}
\end{figure}

%% file: elastic_wave.tex
\subsection{Two dimensional elasticity in heterogeneous media}

In this section, we consider the hyperbolic system of equations describing %
elastic wave propagation in two spatial dimensions.
The system of equations can be written in a %
quasi-linear form (following \cite{LeVeque}) as
\begin{equation}\label{eq:qlin}
 q_t + A q_x + B q_y = 0,
\end{equation}
where the vector $q = \left[ \sigma^{11}, \sigma^{22}, \sigma^{12}, u, v \right]^{\mathrm{T}} $ collects the normal stress components $\sigma^{11}$ and $\sigma^{22}$
in the $x$ and $y$ direction respectively, the shear stress component $\sigma^{12}$, and the time derivatives of the displacement $\delta(x,y,t)$, $u$ and $v$ in the $x-$ and $y-$direction respectively. The matrices $A$ and $B$ are given by
{\footnotesize
\begin{equation}
 A = -
\begin{bmatrix} 
0 & 0 & 0 & (\lambda+2\mu) & 0   \\
0 & 0 & 0 & \lambda        & 0   \\
0 & 0 & 0 & 0              & mu \\
\frac{1}{\rho} & 0 & 0 & 0 & 0 \\
0 & 0 & \frac{1}{\rho} & 0 & 0
\end{bmatrix},
\quad 
B = - 
\begin{bmatrix} 
0 & 0 & 0 & 0 & \lambda   \\
0 & 0 & 0 & 0 &  (\lambda+2\mu) \\
0 & 0 & 0 & \mu & 0 \\
0 & 0 & \frac{1}{\rho} & 0 & 0  \\
0 & \frac{1}{\rho} & 0 & 0 & 0
\end{bmatrix},
\quad
  \lambda = \frac{\nu E}{(1+\nu)(1-2\nu)} \quad \mathrm{and} \quad
  \mu = \frac{E}{2(1+\nu)},
\end{equation}
}
where the Lam\'e parameters $\lambda$ and $\mu$ are expressed as functions of the Young modulus $E$ and the Poisson ratio $\nu.$ 
The derivatives of the vector $q$ with respect to $x$ and $y$ are denoted as $q_x$ and $q_y$.

These constitutive relations stem from the linear elasticity assumption and are appropriate in the presence of small deformations. There are two types of waves propagating in an elastic solid: P-waves which are normal, \textit{i.e.} pressure waves, and S-waves which are shear waves. As a demonstration case, we use the code \texttt{CLAWPACK} \cite{clawpack} and in particular the example described in \S22.7 of \cite{LeVeque}.

The problem domain is a square on $[-1,1]\times[-1,1]$ divided in two regions, a left and a right region with different material properties. The domain is represented in Fig.~\ref{fig:elastic_initial}, where the separation between the two areas is shown as a continuous black line. The properties of the two materials, namely density $\rho$ and the two Lam\'e parameters $\lambda$ and $\mu$, are considered uncertain. The ranges and distribution for these parameters are reported in Table~\ref{tab:elastic_uncertain}. The value of $\nu$ is always lower than $0.5$ for all realizations, preserving the hyperbolic character of the system. Additional details on this system of equations can be found in \cite{LeVeque}.

\begin{table}
  \centering
  {\footnotesize
  \begin{tabular}{c|cccccc}
    Uncertain Parameter & $\rho$ & $\lambda_l$ & $\mu_l$ & $\rho_r$ & $\lambda_r$ &$ \mu_r$  \\
    \hline
    \hline
    Distribution & $\mathcal{U}(0.5,1.5)$  & $\mathcal{U}(3.0,5.0)$ & $\mathcal{U}(0.25,0.75)$ & $\mathcal{U}(0.5,1.5)$  & $\mathcal{U}(1.0,3.0)$ & $\mathcal{U}(0.5,1.5)$  
  \end{tabular}
  }
  \caption{Random parameters for the elastic wave equation problem.}
  \label{tab:elastic_uncertain}
\end{table}

For this problem, an initial perturbation that corresponds to a right-going normal wave is considered, namely $\sigma^{11}=\lambda_l + 2 \mu_l$, $\sigma^{22}=\lambda_l$, $\sigma^{12}=0$, $u = -\sqrt{ (\lambda_l+2\mu_l)/\rho_l }$, $v=0$ for the region $-0.35 < x < -0.2$. The initial solution is zero everywhere else. The initial condition is reported in Fig.~\ref{fig:elastic_initial} for reference. Since the initial condition is a function of the random parameters corresponding to the material properties of the material on the left, the initial condition is itself random.

\begin{figure}[h!]
  \centering
  \begin{subfigure}[b]{0.3\textwidth}
    \centering
    \includegraphics[width=\textwidth,clip,trim=50 20 50 40]{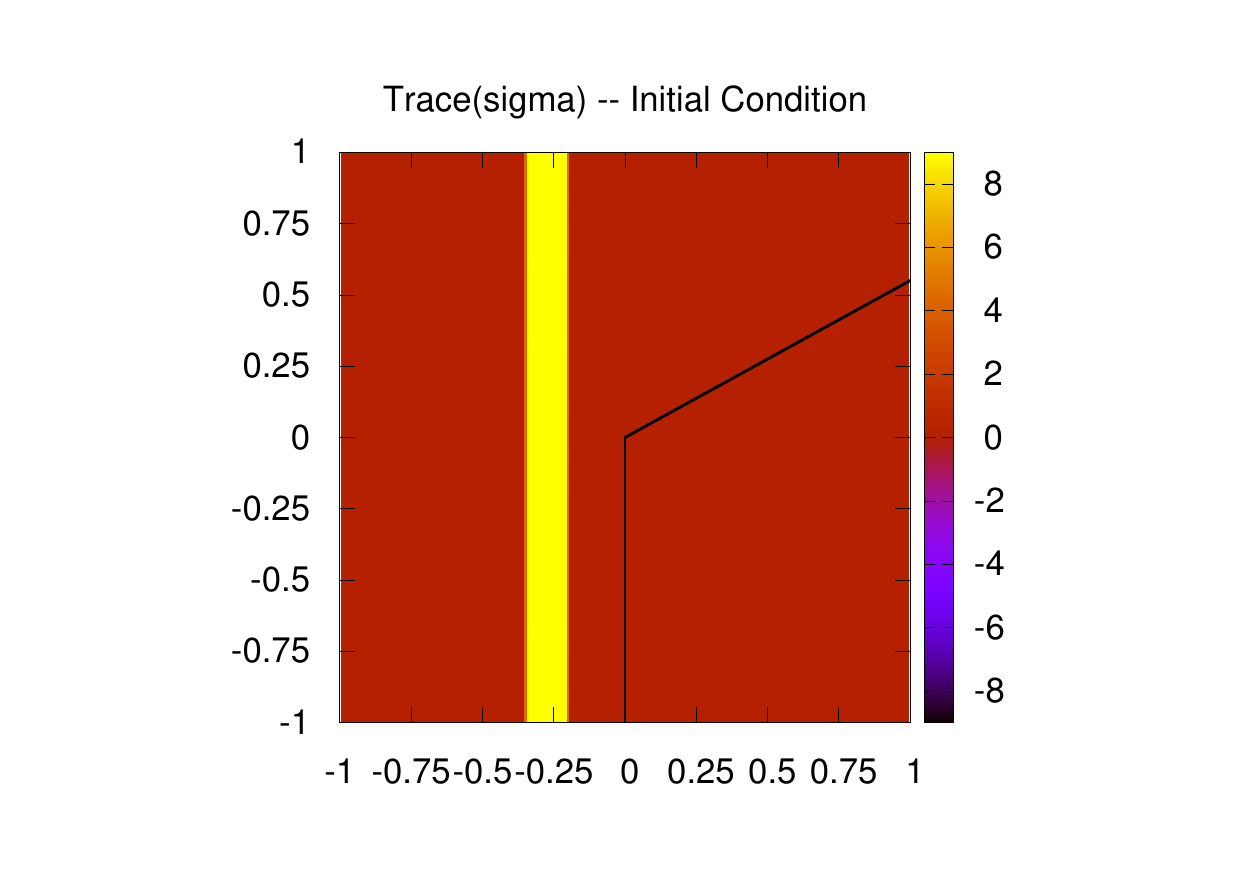}
    \caption{Trace of the stress tensor.}
    \label{fig:elastic_initial_sigma}
  \end{subfigure}
  ~
  \begin{subfigure}[b]{0.3\textwidth}
    \centering
    \includegraphics[width=\textwidth,clip,trim=50 20 50 40]{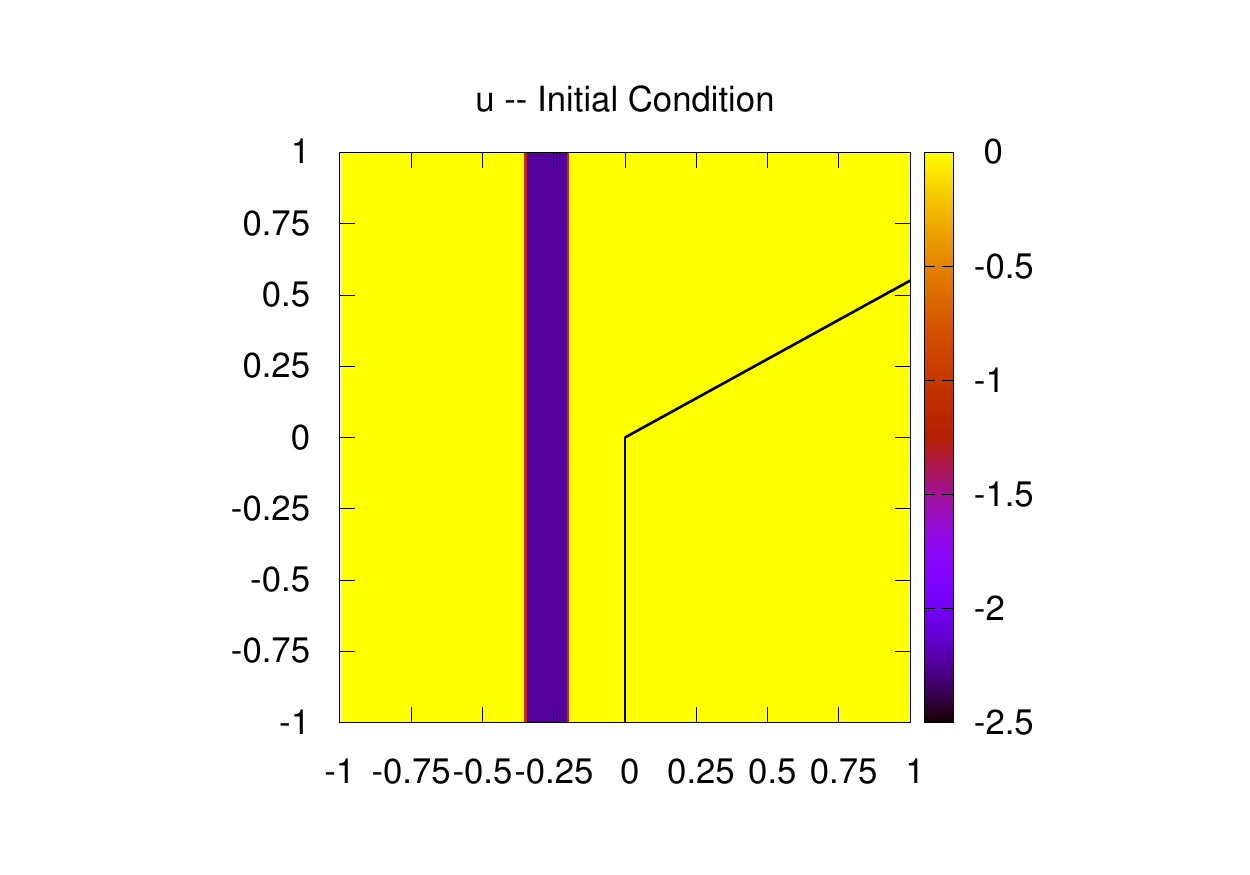}
    \caption{Velocity $u$ component.}
    \label{fig:elastic_initial_u}
  \end{subfigure}
  \caption{Initial conditions for the 2D elastic wave propagation problem in~\eqref{eq:qlin}.}
  \label{fig:elastic_initial}
\end{figure}

The initial perturbation propagates to the right side and hits the interface between the two materials. At the interface, the propagation velocity changes and both transmitted and reflected waves are generated for both normal and shear waves. The system is simulated until a total time equal to $t=0.5$ with a variable time step and a desired CFL of 0.9. For the solution of the equation, finite volumes are used based on the wave-propagation algorithm described in \cite{LeVeque}. Riemann solvers in the direction normal to the cell interface are employed. Two possible fidelities are available: a high-resolution second-order scheme, which employs a monotonized central limiter, and a first-order Godunov scheme. Non-reflecting outflow boundary conditions are imposed by using ghost cells. For both fidelities, we consider five discretization levels corresponding to $10\times10$, $25\times25$, $50\times50$, $100\times100$, $200\times200$ cells in the $x$ and $y$ directions. For this problem, the QoI is chosen to be the average value of the shear stress component $\sigma^{12}$ within a smaller domain inside the material on the right. The solution at time $t=0.55$ corresponding to the shear stress $\sigma^{12}$ is reported in Fig.~\ref{fig:elastic_array} for all the resolution levels and the two model fidelities.
\begin{figure}[h!]
  \centering
  \begin{subfigure}[b]{0.2\textwidth}
    \centering
    \includegraphics[width=\textwidth,clip,trim=50 40 95 40]{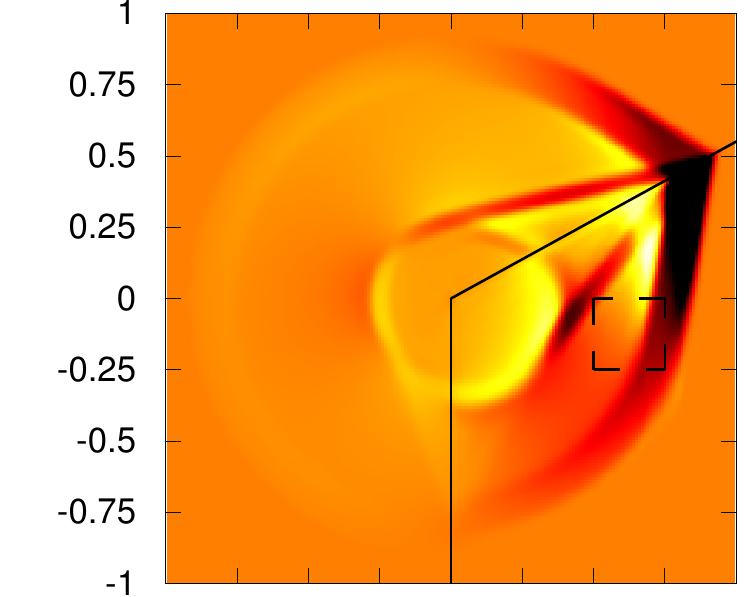}
  \end{subfigure}
  ~
  \begin{subfigure}[b]{0.155\textwidth}
    \centering
    \includegraphics[width=\textwidth,clip,trim=95 40 95 40]{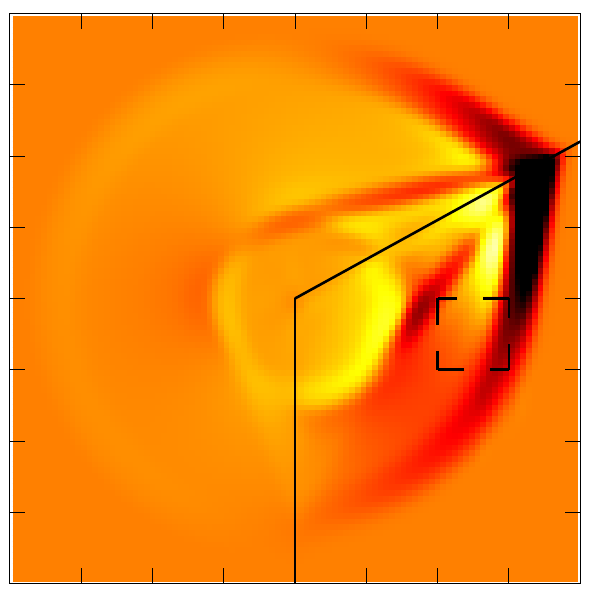}
  \end{subfigure}
  ~
  \begin{subfigure}[b]{0.155\textwidth}
    \centering
    \includegraphics[width=\textwidth,clip,trim=95 40 95 40]{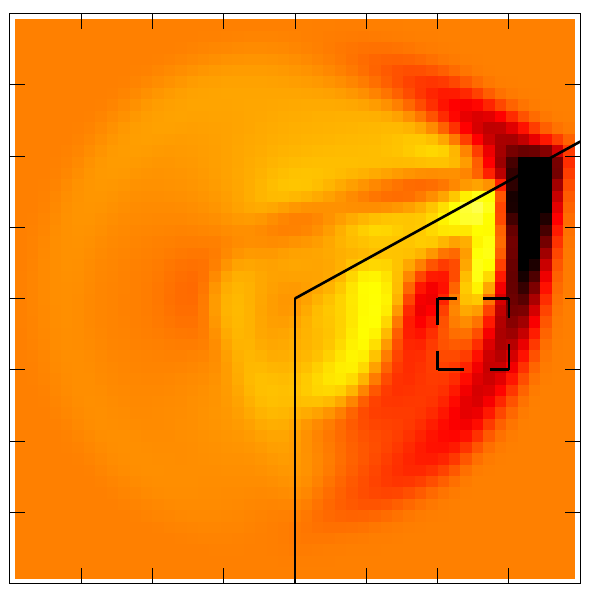}
  \end{subfigure}
  ~
  \begin{subfigure}[b]{0.155\textwidth}
    \centering
    \includegraphics[width=\textwidth,clip,trim=95 40 95 40]{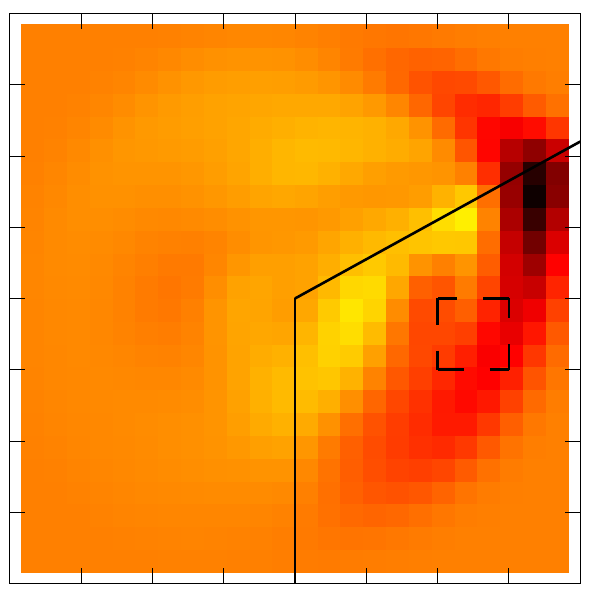}
  \end{subfigure}
  ~
  \begin{subfigure}[b]{0.2\textwidth}
    \centering
    \includegraphics[width=\textwidth,clip,trim=95 40 50 40]{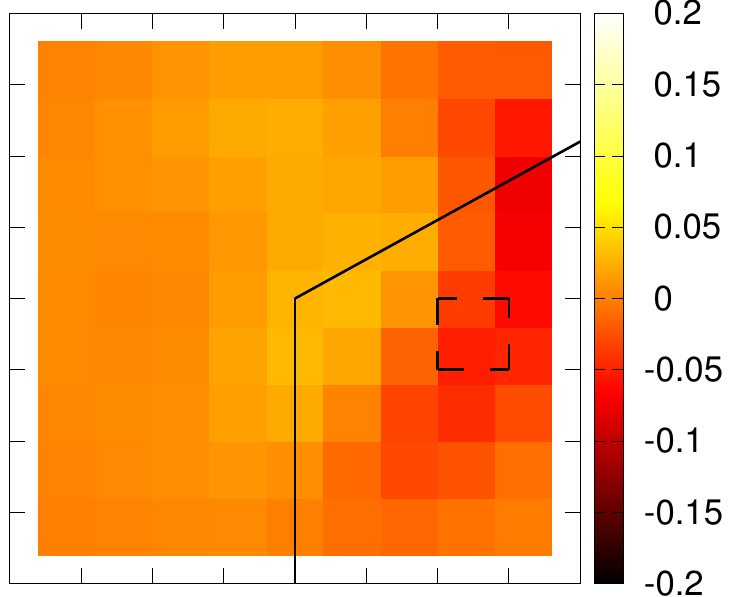}
  \end{subfigure}
  
  \begin{subfigure}[b]{0.2\textwidth}
    \centering
    \includegraphics[width=\textwidth,clip,trim=50 20 95 40]{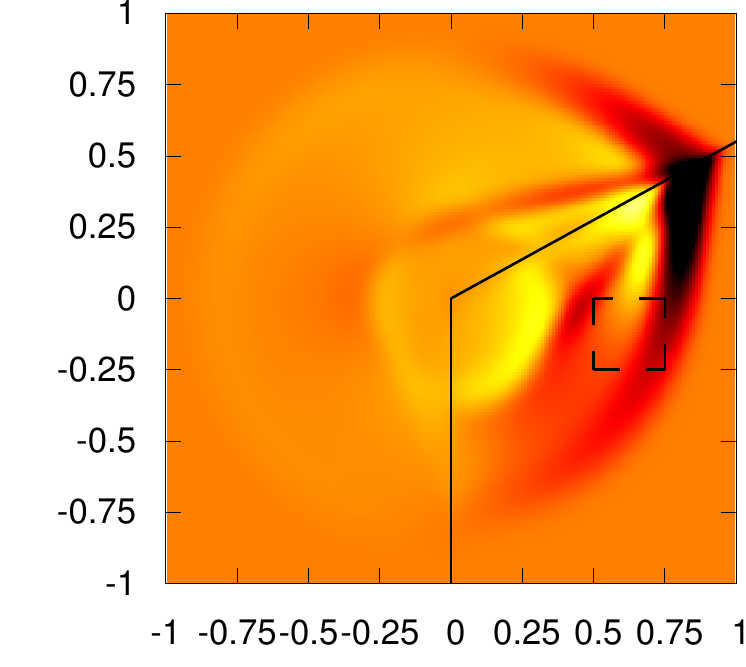}
  \end{subfigure}
  ~
  \begin{subfigure}[b]{0.155\textwidth}
    \centering
    \includegraphics[width=\textwidth,clip,trim=95 20 95 40]{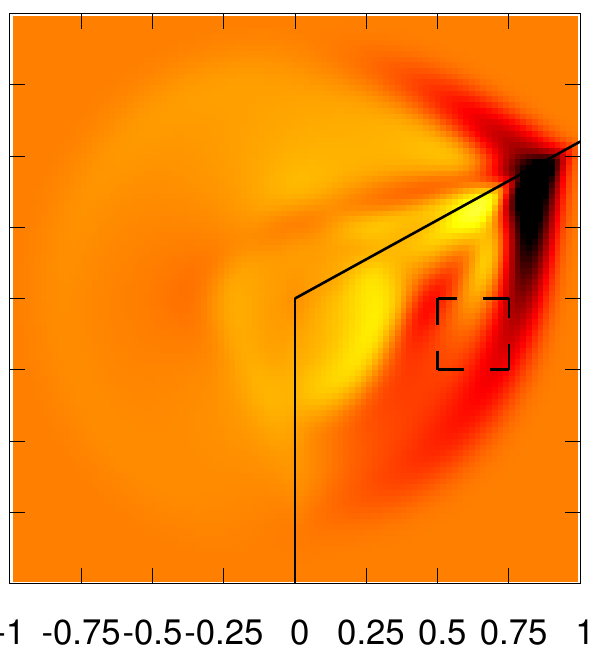}
  \end{subfigure}
  ~
  \begin{subfigure}[b]{0.155\textwidth}
    \centering
    \includegraphics[width=\textwidth,clip,trim=95 20 95 40]{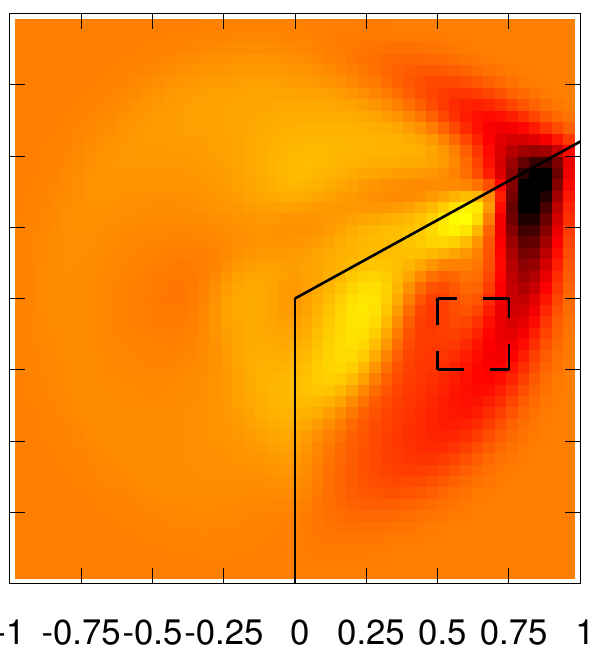}
  \end{subfigure}
  ~
  \begin{subfigure}[b]{0.155\textwidth}
    \centering
    \includegraphics[width=\textwidth,clip,trim=95 20 95 40]{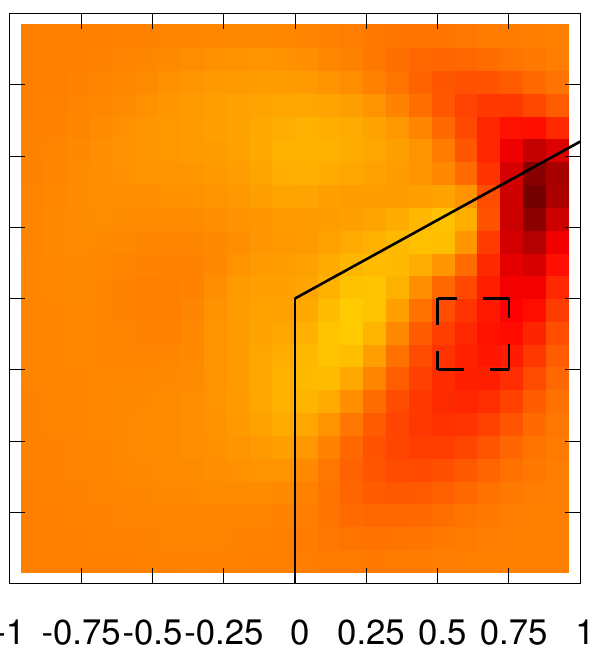}
  \end{subfigure}
  ~
  \begin{subfigure}[b]{0.2\textwidth}
    \centering
    \includegraphics[width=\textwidth,clip,trim=95 20 50 40]{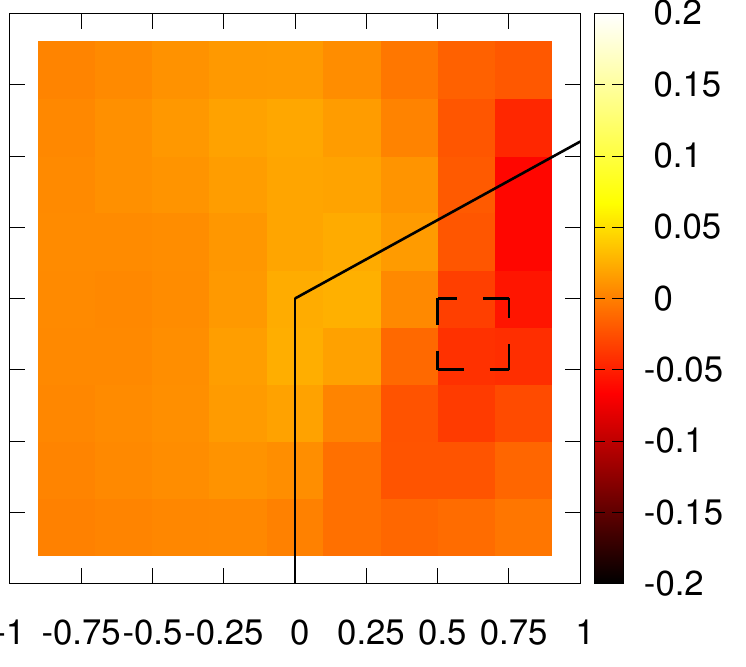}
  \end{subfigure}
  \caption{Shear stress at final time $0.5$ for the two model fidelities (top and bottom rows) and the five discretization levels (200 x 200, 100 x 100, 50 x 50, 25 x 25, 10 x 10 from left to right) corresponding to the mean values of the random parameters for the problem in~\eqref{eq:qlin}. The QoI is the average value of the shear in the dashed region within the right material.}
  \label{fig:elastic_array}  
\end{figure} 
For the ten models reported in Fig.~\ref{fig:elastic_array}, the correlation matrix and the normalized cost are reported in Table~\ref{tab:elastic_correlation} and  Table~\ref{tab:elastic_cost}, respectively.

\begin{table}
  \centering
  {\footnotesize 
        \begin{tabular}{cccccccccc} 
        200 (II) & 100 (II) & 50 (II) & 25 (II) & 10 (II) & 200 (I) & 100 (I) & 50 (I) & 25 (I) & 10 (I) \\ \hline
   1.00000   &0.99838   &0.99245   &0.96560   &0.70267   &0.99312   &0.98333   &0.93857   &0.85400   &0.56719 \\
   0.99838   &1.00000   &0.99092   &0.96461   &0.69060   &0.99160   &0.98380   &0.93360   &0.84743   &0.55127 \\
   0.99245   &0.99092   &1.00000   &0.98759   &0.76255   &0.99866   &0.99484   &0.96738   &0.89785   &0.63184 \\
   0.96560   &0.96461   &0.98759   &1.00000   &0.83904   &0.98697   &0.99400   &0.99102   &0.94874   &0.71607 \\
   0.70267   &0.69060   &0.76255   &0.83904   &1.00000   &0.76356   &0.79165   &0.89148   &0.96032   &0.96725 \\
   0.99312   &0.99160   &0.99866   &0.98697   &0.76356   &1.00000   &0.99700   &0.96965   &0.90058   &0.63184 \\
   0.98333   &0.98380   &0.99484   &0.99400   &0.79165   &0.99700   &1.00000   &0.98022   &0.92207   &0.66156 \\
   0.93857   &0.93360   &0.96738   &0.99102   &0.89148   &0.96965   &0.98022   &1.00000   &0.97785   &0.78607 \\
   0.85400   &0.84743   &0.89785   &0.94874   &0.96032   &0.90058   &0.92207   &0.97785   &1.00000   &0.89023 \\
   0.56719   &0.55127   &0.63184   &0.71607   &0.96725   &0.63184   &0.66156   &0.78607   &0.89023   &1.00000 
        \end{tabular}
        }
  \caption{Correlation matrix for the ten models used in the elastic equation problem Equation~\eqref{eq:qlin}. The second-order (II) and the first-order (I) schemes both employ five different resolution levels. }
  \label{tab:elastic_correlation}
  \end{table}

   For this problem, we consider three scenarios reflecting different analysis situations that might occur in practice. The first scenario is a simple discretization-based hierarchy for a single model fidelity, \textit{i.e.} five total models are adopted 
and the $\qoi_i$ for $i=1,\dots,4$ are obtained by coarsening the spatial resolution. For consistency of results, we choose the high-fidelity, {\em i.e.} the second-order discretization model, for all resolutions in this case. The performance of several estimators are reported in Figure~\ref{fig:elastic_ML_scenario}.

  The second scenario reflects a common need to resort to a lower-fidelity approximation, in this case the first-order numerical discretization. For this low-fidelity alternative, we use resolution levels ranging from $100\times100$ to $10\times10$ as control variates. In other words, our five models are now the high-fidelity fine-resolution reference combined with the four coarsest resolutions of the low-fidelity model. For this scenario the performance of several estimators are reported in Figure~\ref{fig:elastic_MF_scenario}.

  Finally, we consider a scenario in which we create more of a gap between the high-fidelity model and the first control variate by dropping the low-fidelity $100\times100$ resolution from the previous case. 
  In this case, we use a total of four models, comprised of the high-fidelity fine-resolution model and the three coarsest resolutions of the low-fidelity model (resolutions of $50 \times 50$, $25\times25$ and $10\times10$ cells). 
  This choice has the effect of decreasing the correlation between the high fidelity model and the most correlated low-fidelity model to approximately $0.93$. The performance for several estimators are reported in Figure~\ref{fig:elastic_MF_aggressive_scenario}.
\begin{table}[h!]
     \centering
     {\footnotesize
     \begin{tabular}{c|cccccccccc} 
   Model & 200 (II) & 100 (II) & 50 (II) & 25 (II) & 10 (II) & 200 (I) & 100 (I) & 50 (I) & 25 (I) & 10 (I) \\ \hline
   Norm. Cost & 1.000  & 0.147 &  0.026  & 0.009 &  0.002 &  0.498  & 0.080 &  0.013  & 0.004 &  0.002\\
     \end{tabular}
     }
   \caption{Normalized cost with respect to the cost of the second order $200\times200$ resolution for the problem given by Equation~\eqref{eq:qlin}. The highest fidelity requires approximately 6 seconds to run on serially in \texttt{CLAWPACK} on a laptop equipped with an \texttt{Intel\textregistered \; Xeon(R) CPU E3-1505M v5 @ 2.80GHz}.}
   \label{tab:elastic_cost}
\end{table}
  
\begin{figure}[h!]
  \centering
  \begin{subfigure}[t]{0.3\textwidth}
    \centering
    \includegraphics[width=\textwidth,clip,trim=0 5 20 20]{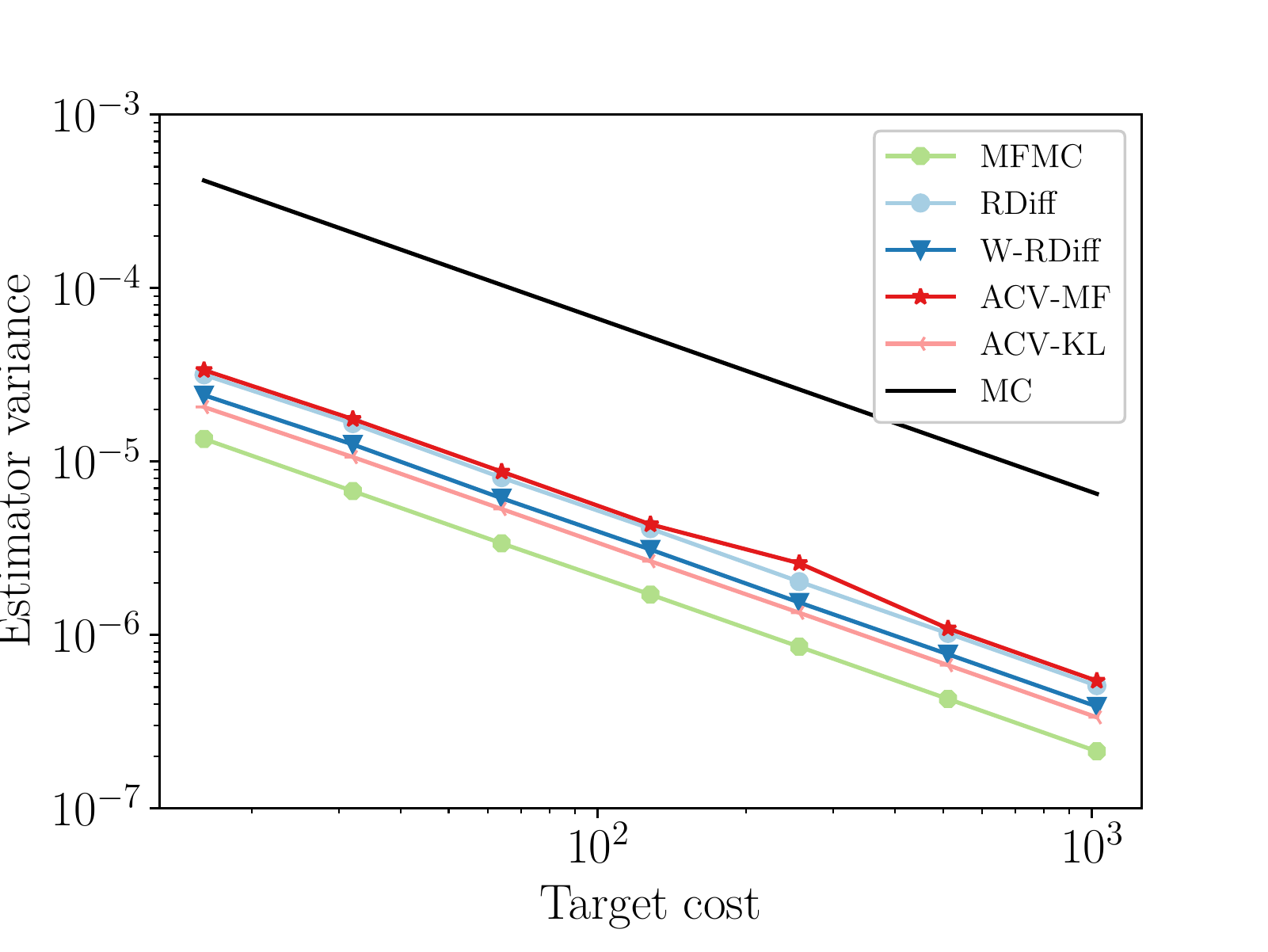}
    \caption{\footnotesize Single-fidelity hierarchy: high-fidelity fine-resolution with control variates from coarser discretizations.}
    \label{fig:elastic_ML_scenario}
  \end{subfigure}
  ~
  \begin{subfigure}[t]{0.3\textwidth}
    \centering
    \includegraphics[width=\textwidth,clip,trim=0 5 20 20]{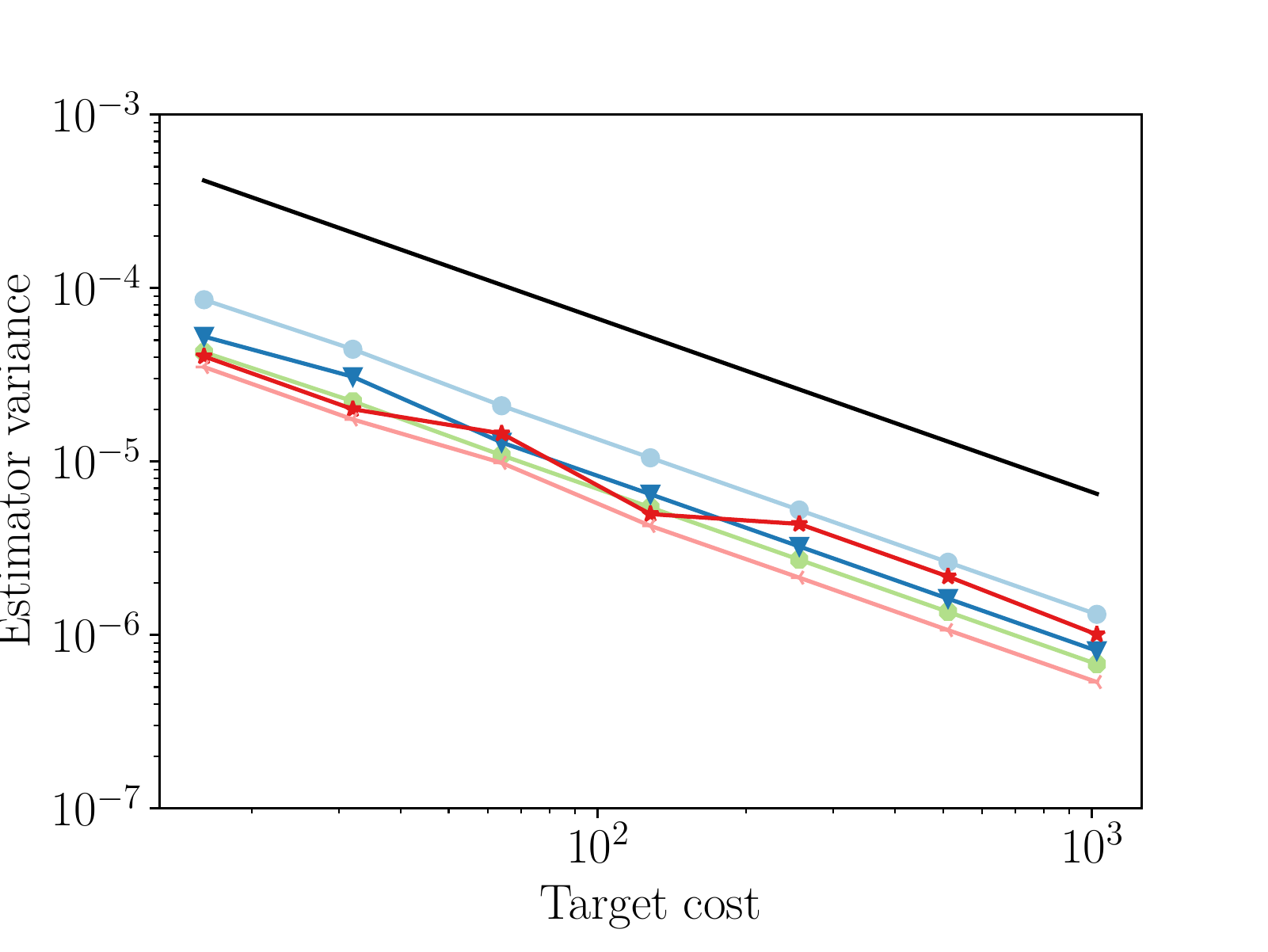}
    \caption{\footnotesize Multifidelity hierarchy: high fidelity fine-resolution with control variates from low-fidelity discretizations.}
    \label{fig:elastic_MF_scenario}
  \end{subfigure}
  ~
  \begin{subfigure}[t]{0.3\textwidth}
    \centering
    \includegraphics[width=\textwidth,clip,trim=0 5 20 20]{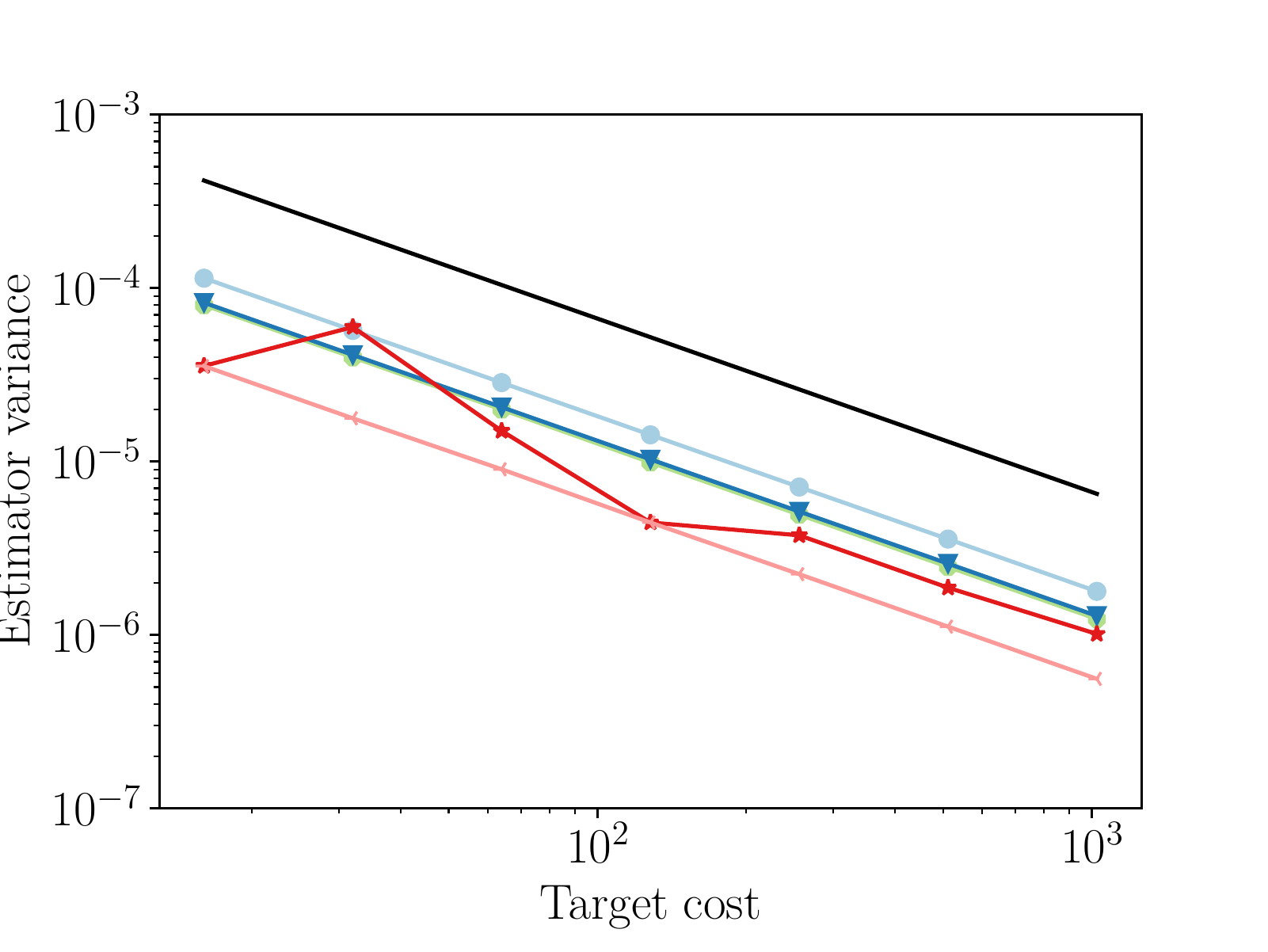}
    \caption{\footnotesize Multifidelity hierarchy with increased gap: high fidelity fine-resolution with control variates from three coarsest low-fidelity discretizations.}
    \label{fig:elastic_MF_aggressive_scenario}
  \end{subfigure}
  \caption{Estimator variance performance for several estimators applied to the hyperbolic system of equations~\eqref{eq:qlin} describing the elastic waves propagation in a set domain with two heterogeneous materials.} 
  \label{fig:elastic_estimators}
\end{figure}

The performance results indicate several interesting features. First, when the models are truly hierarchical as a result of numerical discretization (left panel), all the estimators perform well (more than an order of magnitude variance reduction) and we are in the regime where recursive assumptions (MFMC, RDiff, W-RDiff) are most appropriate. Second, the ACV approaches start to outperform the recursive approaches for the multifidelity cases with more general modeling relationships and lower correlations. One particularly interesting property is that the ACV-KL approach does not degrade with more aggressive mesh coarsening for the two multifidelity cases (mid panel to right panel). This demonstrates robustness in this ACV approach in that it can handle imperfect model relationships without significant degradation. Finally, we see that our simplistic optimization approach has trouble with the ACV-MF estimator causing inconsisitent variance decrease. Better optimization approaches will be the subject of future work.

%% file: appendix.tex
\appendix

\section{Proof of Lemma~\ref{lem:mfmc_weight}}\label{app:proof:mfmc_weight}
\begin{proof}
  Our result follows from computing $\covdiff$ and $\covhc$ and applying Proposition~\ref{prop:ocv}. Recall that
 \begin{equation}\label{eq:mfmc_s_est}
  \cvdiff{i}(\sset{i}) = \frac{1}{\absr{i-1} \nhf} \sum_{j=1}^{\absr{i-1} \nhf} \cv{i}\left(\sset{i}^{1(j)}\right)
                      - \frac{1}{\absr{i} \nhf} \sum_{j=1}^{\absr{i} \nhf} \cv{i}\left(\sset{i}^{2(j)}\right).
 \end{equation}
 Next we partition this quantity into three sums that contain non-overlapping (independent) sets of samples. The first sum is partitioned according to
 \begin{equation*}
 \frac{1}{\absr{i-1} \nhf} \sum_{j=1}^{\absr{i-1} \nhf} \cv{i}\left(\sset{i}^{1(j)}\right) = \frac{1}{\absr{i-1} \nhf}\left[ \sum_{j=1}^{\nhf} Q_{i}\left(\sset{j}^{1(j)}\right) +  \sum_{j=N+1}^{\absr{i-1}\nhf} Q_{i} \left(\sset{i}^{1(j)}\right)\right].
 \end{equation*}
 The second sum is partitioned according to
 \begin{equation*}
    \frac{1}{\absr{i} \nhf} \sum_{j=1}^{\absr{i} \nhf} \cv{i}\left(\sset{i}^{2(j)}\right) = \frac{1}{\absr{i} \nhf} \Bigg[ \sum_{j=1}^{\nhf} \cv{i} \left(\sset{i}^{1(j)}\right) + \sum_{j=N+1}^{\absr{i-1}\nhf} \cv{i} \left(\sset{i}^{1(j)}\right) + \sum_{j=\absr{i-1}\nhf+1}^{\absr{i}\nhf} \cv{i}\left(\sset{i}^{2(j)}\right) \Bigg]
 \end{equation*} 
 Now we can rewrite $\cvdiff{i}$ in terms of the three common summations
  \begin{equation}\label{eq:mfmc_form1}
  \begin{split}
  \cvdiff{i}(\sset{i}) 
  = \frac{\absr{i}-\absr{i-1}}{\absr{i}\absr{i-1}} \frac{1}{\nhf} \sum_{j=1}^{\nhf} \cv{i} \left(\sset{i}^{(j)}\right) +
  \frac{\absr{i}-\absr{i-1}}{\absr{i}\absr{i-1}} \frac{1}{\nhf} \sum_{j=\nhf+1}^{\absr{i-1}\nhf} \cv{i} \left(\sset{i}^{(j)}\right)
  -  \frac{1}{\absr{i}} \frac{1}{\nhf} \sum_{j=\absr{i-1}\nhf + 1}^{\absr{i} \nhf} \cv{i}\left(\sset{i}^{(j)}\right),
  \end{split}
  \end{equation}
  where the superscripts were dropped because it is no longer necessary to distinguish between $\sset{i}^1$ and $\sset{i}^2$.
 
Now we consider the diagonal of $\covdiff$. Since the sums are independent, the variance of the sums is the sum of the variances
 \begin{align*}
   \VarF{\cvdiff{i}} &=  \frac{(\absr{i}-\absr{i-1})^2}{\absr{i}^2\absr{i-1}^2} \frac{\VarF{\qoi}}{\nhf}  + \frac{(\absr{i}-\absr{i-1})^2}{\absr{i}^2\absr{i-1}^2} \left( \absr{i-1}-1 \right) \frac{\VarF{\qoi}}{\nhf} + \frac{\absr{i}-\absr{i-1}}{\absr{i}^2} \frac{\VarF{\qoi}}{\nhf} \nonumber \\
   &=  \absr{i-1}\frac{(\absr{i}-\absr{i-1})^2}{\absr{i}^2\absr{i-1}^2} \frac{\VarF{\qoi}}{\nhf}  + \frac{\absr{i}-\absr{i-1}}{\absr{i}^2} \frac{\VarF{\qoi}}{\nhf} 
   = \left[\frac{ \absr{i}^2-2\absr{i}\absr{i-1} + \absr{i-1}^2}{\absr{i}^2\absr{i-1}} + \frac{\absr{i-1}}{\absr{i-1}}\frac{\absr{i} - \absr{i-1}}{\absr{i}^2}\right]  \frac{\VarF{\qoi}}{\nhf} \\
   &= \frac{\absr{i} \left( \absr{i}-\absr{i-1} \right)}{\absr{i}^2\absr{i-1}} \frac{\VarF{\qoi}}{\nhf}
   = \frac{\absr{i}-\absr{i-1}}{\absr{i}\absr{i-1}} \frac{\VarF{\qoi}}{\nhf}. 
 \end{align*}
 
 Next we consider the off-diagonal terms of $\covdiff$. Without loss of generality, consider $\cvdiff{i}$ and $\cvdiff{j}$ for $j > i$. The nested structure of the sampling set means that $\sset{j}^{(k)} = \sset{i}^{(k)}$ for $k \leq \absr{i}\nhf$. Using this fact we can rewrite Equation~\eqref{eq:mfmc_form1} as
 \begin{equation*}
 \begin{split}
  \cvdiff{i}(\sset{i}) &= \frac{\absr{i}-\absr{i-1}}{\absr{i}\absr{i-1}N} 
  \sum_{k=1}^{\absr{i-1}\nhf} \cv{i}\left(\sset{i}^{(k)}\right) - 
             \frac{1}{\absr{i}\nhf}
             \sum_{k=\absr{i-1}\nhf+1}^{\absr{i}\nhf} \cv{i}\left(\sset{i}^{(k)}\right)
 \end{split}
 \end{equation*}
 for the $i$th  CV, and 
 \begin{align*}
  \cvdiff{j}(\sset{j}) &= \frac{\absr{j}-\absr{j-1}}{\absr{j}\absr{j-1}} \frac{1}{\nhf}
  A_j
  - \frac{1}{\absr{j}} \frac{1}{\nhf} \sum_{k=\absr{j-1}\nhf+1}^{\absr{j}\nhf} Q_j\left(\sset{j}^{(k)}\right),
 \end{align*}
 where \[A_j = \sum_{k=1}^{\absr{i-1}\nhf} Q_j\left(\sset{i}^{(i)}\right) \sum_{k=\absr{i-1}\nhf+1}^{\absr{i}\nhf} Q_j\left(\sset{i}^{k}\right) +  \sum_{k=\absr{i}\nhf+1}^{\absr{j-1}\nhf} Q_j\left(\sset{k}^{k}\right),\]
 for the $j$th. Only the first two sums in this last expression share samples with the $\cvdiff{i}$, and these sums correspond to the samples delegated to $\est{j}$. Therefore we have
 \begin{equation*}
 \begin{split}
     \covdiffij{i}{j} &= \frac{(\absr{i}-\absr{i-1})(\absr{j}-\absr{j-1})}{\absr{i}\absr{i-1}\absr{j}\absr{j-1}} \absr{i-1} \covF{\cv{i},Q_j} - 
     \frac{\absr{j}-\absr{j-1}}{\absr{i}\absr{j}\absr{j-1}} (\absr{i} - \absr{i-1})\covF{\cv{i},\cv{j}} 
     = 0.
     \end{split}                                              
 \end{equation*}
To summarize $\covdiff$ is diagonal
 \begin{equation}
   \covdiffij{i}{j} =
   \left\{
   \begin{array}{cl}
     \dfrac{\absr{i}-\absr{i-1}}{\absr{i}\absr{i-1}} \dfrac{\varF{\cv{i}}}{\nhf} & \textrm{ if } i = j \\
     0 & \textrm{ otherwise }
   \end{array}
   \right. , \textrm{ where } \absr{0} = 1.
   \label{eq:mfmc_covq}
 \end{equation}

Next we turn to $\covhc$. From Equation~\ref{eq:mfmc_form1} we see that samples that each CV shares with the high fidelity model are entirely contained in the first sum, and therefore
 \begin{equation}
  \covF{\est{},\cvdiff{i}} = \frac{\absr{i}-\absr{i-1}}{\absr{i}\absr{i-1}} \frac{\covF{Q,\cv{i}}}{\nhf}. \label{eq:mfmc_covhc}
 \end{equation}
 Using $\covhc$ and $\covdiff$ within Equation~\ref{eq:ocv_var} of Prop.~\ref{prop:ocv} yields our stated result.
\end{proof}

\section{Proof of Lemma~\ref{lem:MFMC_var_reduction}}\label{app:proof:mfmc_var_reduction}
\begin{proof}
  The result follows directly from Proposition~\ref{prop:ocv} as
  \begin{equation}
    \varF{\mfmcest(\vec{\cvw})} = \varF{\est{}}\left( 1 - R_{\mfmc}^2\right),
  \end{equation}
  where $R_{\mfmc}^2 = \covhc^T \frac{\covdiff^{-1}}{\varF{\est{}}}\covhc$. Equations~\eqref{eq:mfmc_covq} and~\eqref{eq:mfmc_covhc} in the proof of Lemma~\ref{lem:mfmc_weight} provide us with expressions of $\covdiff$ and $\covhc$, respectively. Since the covariance is diagonal and invertible we have
  \begin{align*}
    R_{\mfmc}^2 &= \sum_{i=1}^{\nmodels} \frac{\absr{i}-\absr{i-1}}{\absr{i}\absr{i-1}} \frac{\covF{\qoi,\cv{i}}^2}{\varF{\cv{i}}\varF{\qoi}} = \sum_{i=1}^{\nmodels} \frac{\absr{i}-\absr{i-1}}{\absr{i}\absr{i-1}}\ccoeff{i}^2.
  \end{align*}
  We can pull out $\ccoeff{1}$ from this expression 
  \begin{equation*}
  \begin{split}
  R_{\mfmc}^2 &= \frac{\absr{1}-1}{\absr{1}} \ccoeff{1}^2 + \ccoeff{1}^2 \sum_{i=2}^{\nmodels} \frac{\absr{i}-\absr{i-1}}{\absr{i}\absr{i-1}} \frac{\ccoeff{i}^2}{\ccoeff{1}^2} 
             = \ccoeff{1}^2 \left( \frac{\absr{1}-1}{\absr{1}} + \sum_{i=2}^{\nmodels} \frac{\absr{i}-\absr{i-1}}{\absr{i}\absr{i-1}} \frac{\ccoeff{i}^2}{\ccoeff{1}^2} \right),
  \end{split}
  \end{equation*}
  where we have used $\absr{0} = 1$ to enable a clear comparison with the single control variate.
\end{proof}


\input{appmlmc}

\section{Proof of Theorem~\ref{th:opt_acv1}}\label{app:proof:opt_acv1}

For reference we recall the definition of $\fmatone$
  \begin{equation*}
    \fmatone_{ij} = \left\{
      \begin{array}{cl}
        \frac{\absr{i}-1}{\absr{i}}\frac{\absr{j}-1}{\absr{j}}  & \textrm{if } i \neq j \\
        \frac{\absr{i}-1}{\absr{i}} & \textrm{otherwise} 
      \end{array}
      \right. . 
  \end{equation*}  

\begin{proof}
  This proof again makes use of Propositions~\ref{prop:acv_var} and~\ref{prop:ocv}, which require the computation of $\covdiff$ and $\covhc$. We begin with $\covdiff$ by noticing that $\cvdiff{i}(\sset{i}) = \est{i}\left(\sset{} \right) - \estm{i}\left(\sset{i}\right)$ and that for $i \neq j$, $\cvdiff{i}$ is correlated with $\cvdiff{j}$ only through the first $\nhf$ samples. This pattern will emerge for all the various quantities that we require and so we first split each $\cvdiff{i}$ into two sums that have independent samples. This splitting happens by first noticing that we can write
  \begin{equation*}
  \begin{split}
  \estm{i} &= \frac{1}{\absr{i}\nhf} \sum_{k=1}^{\absr{i}\nhf} \cv{i}(\sset{i}^{(k)}) 
     = \frac{1}{\absr{i}\nhf}\left[\sum_{k=1}^{\nhf}\cv{i}(\sset{}^{(k)}) + \sum_{k=\nhf+1}^{\absr{i}\nhf}\cv{i}(\sset{i}^{(k)})  \right],
  \end{split}
  \end{equation*}
  so that 
  \begin{equation}
    \cvdiff{i}(\sset{i}) =
    \frac{\absr{i}-1}{\absr{i}\nhf}\sum_{k=1}^{\nhf} \cv{i}(\sset{}^{(k)}) - \frac{1}{\absr{i}\nhf} \sum_{k=\nhf+1}^{\absr{i}\nhf} \cv{i}(\sset{i}^{(k)}). \label{eq:splitting}
  \end{equation}
  Since these two sums are independent, the covariance between $\cvdiff{i}$ and $\cvdiff{j}$ for $i \neq j$ is due only to the first summation:
  \begin{equation*}
    \covdiffij{i}{j} = \frac{\absr{i}-1}{\absr{i}}\frac{\absr{j}-1}{\absr{j}} \frac{1}{\nhf}\covF{\cv{i},\cv{j}}.
  \end{equation*}
  Using the same splitting we can derive the diagonal terms (variances) as 
  \begin{equation} \label{eq:acv_diag}
  \resizebox{0.9\hsize}{!}{$
  \begin{split}
    \VarF{\cvdiff{i}} &= \frac{(\absr{i}-1)^2}{\absr{i}^2\nhf}\varF{\cv{i}} + \frac{\absr{i}-1}{\absr{i}^2\nhf} \varF{\cv{i}} 
    = \left(\frac{(\absr{i}-1)^2}{\absr{i}^2} + \frac{\absr{i}-1}{\absr{i}^2}\right)\frac{\varF{\cv{i}}}{\nhf}  
    = \frac{\absr{i}-1}{\absr{i}} \frac{\varF{\cv{i}}}{\nhf}. 
    \end{split}
    $}
  \end{equation}
  Ths result can be succinctly represented as
   $ \covdiff = \frac{1}{\nhf}\left[ \covm \circ \fmatone\right]. $

  Now we consider the $\covhc$ term. The covariance between $\est{}$ and $\cvdiff{i}$ is again a result of the first $\nhf$ samples of $\sset{i}$, or the first term in~\eqref{eq:splitting}, yielding 
  \begin{equation*}
    \covhc = \frac{1}{\nhf}\left[\Diag{\fmatone} \circ \covv\right] \label{eq:covhc}
  \end{equation*}
  Using $\covdiff$ and $\covhc$ and Propositions~\ref{prop:acv_var} and~\ref{prop:ocv}, we obtain our stated result.
\end{proof}

\section{Proof of Theorem~\ref{th:opt_acv2}}\label{app:proof:opt_acv2}
For reference, recall the definition of $\fmattwo$
  \begin{equation*}
    \fmattwo_{ij} = \left\{
      \begin{array}{cl}
        \frac{\min(\absr{i},\absr{j})-1}{\min(\absr{i},\absr{j})}  & \textrm{if } i \neq j \\
        \frac{\absr{i}-1}{\absr{i}} & \textrm{otherwise} 
      \end{array}
      \right. .
  \end{equation*}

\begin{proof}
  The difference between ACV-IS and ACV-MF is that the samples used for $\estm{i}$ are also the first $\min(\absr{j}\nhf, \absr{i}\nhf)$ samples of $\estm{j}$ when $j \geq i$. As a result, we have different expressions for the off-diagonal terms of $\covdiff$, but we retain the same expressions for the diagonal terms $\varF{\cvdiff{i}}$ and for $\covhc$. Therefore, we need to derive a new expression for $\covdiffij{i}{j}$ for this estimator, and then reuse the previous results within Propositions~\ref{prop:acv_var} and~\ref{prop:ocv} to obtain our results. We have two cases to consider, when $\absr{j} \geq \absr{i}$ and when $\absr{j} < \absr{i}$.

\textbf{Case 1 $\mathbf{\absr{j} \geq \absr{i}}:$}
Our derivation begins by splitting $\cvdiff{i}$ and $\cvdiff{j}$ into sums with independent samples. The split for $\cvdiff{i}$ is given by~\eqref{eq:splitting}.
  Next we consider $\cvdiff{j}$, and begin by splitting the summation of the samples used to obtain $\estm{j}$ into its three constituent sample sets:
  \begin{equation*}
    \begin{split}
    \sum_{k=1}^{\absr{j}\nhf} \cv{j}(\sset{j}^{(k)}) &= \sum_{k=1}^{\nhf}\cv{j}(\sset{}^{(k)}) + \sum_{k=\nhf+1}^{\absr{i}\nhf}\cv{j}(\sset{i}^{(k)}) + 
    \sum_{k=\absr{i}\nhf+1}^{\absr{j}\nhf}\cv{j}(\sset{j}^{(k)}).
    \end{split}
  \end{equation*}
  Using this splitting we obtain the following result
    \begin{equation} \label{eq:split-acv2}
    \begin{split}
    \cvdiff{j}(\sset{j}) &= \frac{1}{\absr{j}\nhf}\Big((\absr{j}-1)\sum_{k=1}^{\nhf}\cv{j}(\sset{}^{(k)}) - \sum_{k=\nhf+1}^{\absr{i}\nhf}\cv{j}(\sset{i}^{(k)}) - 
    \sum_{k=\absr{i}\nhf+1}^{\absr{j}\nhf}\cv{j}(\sset{j}^{(k)})\Big)
    \end{split}
     \end{equation}
  where we have used the fact that $\sset{j}^{(k)} = \sset{i}^{(k)}$ for $k \leq \absr{i}\nhf$. Now it is clear that the covariance between $\cvdiff{i}$ and $\cvdiff{j}$ is due to the samples from the first two summations. A straightforward computation leads to
  \begin{align*}
    \covF{\cvdiff{i}, \cvdiff{j}} &=  \frac{\absr{i} - 1}{\absr{i}} \frac{\absr{j} - 1}{\absr{j}}\frac{\covF{\cv{i},\cv{j}}}{\nhf} + \frac{(\absr{i}-1)}{\absr{i}\absr{j}} \frac{\covF{\cv{i},\cv{j}}}{\nhf} 
    = \frac{\absr{i} - 1}{\absr{i}} \frac{\covF{\cv{i},\cv{j}}}{\nhf}.
  \end{align*}

\textbf{Case 2 $\mathbf{\absr{i} > \absr{j}}$:}
  This requires a different splitting of $\cvdiff{j}$, one which uses
  \begin{equation*}
  \estm{j} = \frac{1}{\absr{j}\nhf} \sum_{k=1}^{\absr{j}\nhf} \cv{j}(\sset{j}^{(k)}) = \frac{1}{\absr{j}\nhf}\left[\sum_{k=1}^{\nhf}\cv{j}(\sset{}^{(k)}) + \sum_{k=\nhf+1}^{\absr{j}\nhf}\cv{j}(\sset{i}^{(k)})\right]
  ,
  \end{equation*}
  to obtain
  \begin{equation*}
    \cvdiff{j}(\sset{j}) = \frac{\absr{j}-1}{\absr{j}\nhf}\sum_{k=1}^{\nhf}\cv{j}(\sset{}^{(k)}) - \frac{1}{\absr{j}\nhf} \sum_{k=\nhf+1}^{\absr{j}\nhf}\cv{j}(\sset{i}^{(k)}).
  \end{equation*}
  The difference from Case 1 is that the first equation does not have an extra term to account for samples in $\sset{j}$ that are not in $\sset{i}$. Using the same reasoning as for Case 1 we obtain
    \begin{align*}
      \covF{\cvdiff{i}, \cvdiff{j}} &=  \frac{\absr{i} - 1}{\absr{i}} \frac{\absr{j} - 1}{\absr{j}}\frac{\covF{\cv{i},\cv{j}}}{\nhf} + \frac{(\absr{j}-1)}{\absr{i}\absr{j}} \frac{\covF{\cv{i},\cv{j}}}{\nhf} 
      = \frac{\absr{j} - 1}{\absr{j}} \frac{\covF{\cv{i},\cv{j}}}{\nhf}.
    \end{align*}
We combine these two cases into a single formula
  \begin{equation}
    \covF{\cvdiff{i}, \cvdiff{j}} = \frac{\min(\absr{i},\absr{j}) - 1}{\min(\absr{i},\absr{j})} \frac{\covF{\cv{i},\cv{j}}}{\nhf}.
  \end{equation}

  The expression for $\varF{\cvdiff{i}}$ is identical to that of estimator ACV-IS~\eqref{eq:acv_diag} because the same splitting strategy can be used. The expression for $\covhc$ is also identical to that of estimator ACV-IS~\eqref{eq:covhc} for the same reason. The stated result follows by using these expressions within Propositions~\ref{prop:acv_var} and~\ref{prop:ocv} as before.
\end{proof}

\section{Proof of Theorem~\ref{th:ACV}}\label{app:proof:ACV}
\begin{proof}
  Using the previous proposition we need to show that $\fmatone \to \mat{1}_{\nmodels \times \nmodels}$ and that $\fmattwo \to \mat{1}_{\nmodels \times \nmodels}$. Since these two matrices share diagonal entries, we begin with
  \begin{equation*}
    \lim_{\absr{i} \to \infty}\fmatone_{ii} = \lim_{\absr{i} \to \infty}\fmattwo_{ii} = \lim_{\absr{i} \to \infty} \frac{\absr{i} - 1}{\absr{i}} = 1.
  \end{equation*}
  Next we consider the off-diagonals of $\fmatone$:
  \begin{align*}
    \lim_{\absr{i},\absr{j} \to \infty}\fmatone_{ij} = \lim_{\absr{i},\absr{j} \to \infty} \frac{\absr{i} - 1}{\absr{i}} \frac{\absr{j} - 1}{\absr{j}} = 1
  \end{align*}
  because each component converges to 1. Finally we consider the off-diagonals of $\fmattwo$
  \begin{align*}
    \lim_{\absr{i},\absr{j} \to \infty}\fmattwo_{ij} = \lim_{\absr{i},\absr{j} \to \infty} \frac{\min(\absr{i},\absr{j}) - 1}{\min(\absr{i},\absr{j})} = 1
  \end{align*}
  because the minimum value goes to infinity since both $\absr{i}$ and $\absr{j}$ go to infinity.
\end{proof}

\section{Proof of Theorem~\ref{th:ACV-KL}} \label{app:proof:ACV-KL}
\begin{proof}
  Again we rely on Propositions~\ref{prop:acv_var} and~\ref{prop:ocv}. The separated nature of this estimator implies that the computation of $\covhc$ and $\covdiff$ can be separated into three regimes: (1) $i,j \leq K$; (2) when both $i,j > K$; and (3) either $i \leq K$ and $j>K$, which is symmetric to $j \leq K$ and $i>K$.
  The first two cases are the most straight forward because they essentially follow from ACV-MF.
  
  \textbf{Case 1: $i,j \leq K$}. This case is identical to ACV-MF for a CV estimator with $K$ CVs:
  \begin{equation*}
    \covF{\est{},\cvdiff{i}} = \frac{\absr{i}-1}{\absr{i}} \frac{\covF{\qoi, \cv{i}}}{\nhf}, \quad 
    \covdiffij{i}{j} 
    = \frac{\min(\absr{i},\absr{j})-1}{\min(\absr{i},\absr{j})}\frac{\covF{\cv{i},\cv{j}}}{\nhf}, \textrm{ and }
    \varF{\cvdiff{i}} 
    = \frac{\absr{i}-1}{\absr{i}} \frac{\varF{\cv{i}}}{\nhf}.
  \end{equation*}
  
  \textbf{Case 2: $i,j > K$}. For this set of indices we generate a new set of indepedent splittings amongst $\sset{}$, $\sset{L}$, $\sset{i}$ and $\sset{j}$. If we assume $\absr{j} > \absr{i}$ for $j > i$ then
    \begin{align}
    \cvdiff{j}(\sset{j}) 
    &=\frac{1}{\absr{L}\nhf}\sum_{k=1}^{\absr{L}\nhf}\cv{j}(\sset{L}^{(k)}) -  
    \frac{1}{\absr{j}\nhf}\Bigg[\sum_{k=1}^{\absr{L}\nhf}\cv{j}(\sset{L}^{(k)}) +  \sum_{k=\absr{L}\nhf+1}^{\absr{i}\nhf} \cv{j}(\sset{i}^{(k)}) + \sum_{k=\absr{i}\nhf+1}^{\absr{j}\nhf} \cv{j}(\sset{j}^{(k)}) \Bigg] \nonumber 
    \end{align}
    Combining terms we obtain
    \begin{equation}
    \begin{split}
    \cvdiff{j}(\sset{j}) &= \frac{\absr{j}-\absr{L}}{\absr{L}\absr{j}\nhf} \left[ \sum_{k=1}^{\nhf}\cv{j}(\sset{}^{(k)}) + \sum_{k=\nhf+1}^{\absr{L}\nhf}\cv{j}(\sset{L}^{(k)})\right] - 
    \frac{1}{\absr{j}\nhf}\left[ \sum_{k=\absr{L}\nhf+1}^{\absr{i}\nhf} \cv{j}(\sset{i}^{(k)}) + \sum_{k=\absr{i}\nhf+1}^{\absr{j}\nhf} \cv{j}(\sset{j}^{(k)}) \right].
    \end{split}
    \label{eq:kacv-split2}
    \end{equation}
    If on the other hand,  $\absr{j} < \absr{i}$ for $j > i$ then instead of two sums in the second component of Equation~\eqref{eq:kacv-split2} we only have one
    \begin{equation}
    \begin{split}
    \cvdiff{j}(\sset{j}) &= \frac{\absr{j}-\absr{L}}{\absr{L}\absr{j}\nhf} \left[ \sum_{k=1}^{\nhf}\cv{j}(\sset{}^{(k)}) + \sum_{k=\nhf+1}^{\absr{L}\nhf}\cv{j}(\sset{L}^{(k)})\right] - 
                         \frac{1}{\absr{j}\nhf} \sum_{k=\absr{L}\nhf+1}^{\absr{j}\nhf} \cv{j}(\sset{i}^{(k)}),
    \end{split}
    \label{eq:kacv-split3}
    \end{equation}
    because in this case $\sset{j}$ would consist of the first $\absr{j}\nhf$ samples of $\sset{i}$.

    Now we use Equations~\eqref{eq:kacv-split2} and~\eqref{eq:kacv-split3} to compute the required covariances. First we have
    \begin{align*}
      \covF{\est{},\cvdiff{i}} &=  \frac{\absr{i}-\absr{L}}{\absr{i}\absr{L}}\frac{\covF{\qoi,\cv{i}}}{\nhf} ,
    \end{align*}
    where we have used the fact that the $\nhf$ shared samples used in $\cvdiff{i}$ are contained in the first sum in Equation~\eqref{eq:kacv-split2}, and that this sum has the coefficient $\frac{\absr{i}-\absr{L}}{\absr{i}\absr{L}} \frac{1}{\nhf}$. 

    Next we have the covariance between the control variates. These estimators share several groups of independent samples that we will first consider separately and then sum together. The first group of shared samples is $\sset{L}$ which consists of $\absr{L}\nhf$ samples. These samples correspond to the first two summations of Equation~\eqref{eq:kacv-split2}, and therefore their covariance becomes
    \begin{align*}
      \frac{\absr{i}-\absr{L}}{\absr{i}}\frac{\absr{j}-\absr{L}}{\absr{j}} \frac{1}{\absr{L}}\frac{\covF{\cv{i},\cv{j}}}{\nhf}.
    \end{align*}
    Next, these estimators share an additional $(\min(\absr{i},\absr{j}) - \absr{L})\nhf$. If $\absr{j} > \absr{i}$ for $j > i$ then these shared samples arise in the third sum of the splitting~\eqref{eq:kacv-split2} and we obtain
    $  \frac{1}{\absr{i}\absr{j}} (\absr{i} - \absr{L}) \frac{\covF{\cv{i},\cv{j}}}{\nhf},$
    and if $\absr{j} < \absr{i}$ for $j > i$ then we have
    $ \frac{1}{\absr{i}\absr{j}} (\absr{j} - \absr{L}) \frac{\covF{\cv{i},\cv{j}}}{\nhf}.$
    For $j > i$ these imply that
    \begin{equation*}
    \covdiffij{i}{j} = \left[\frac{(\absr{i}-\absr{L})(\absr{j}-\absr{L}) + \absr{L}(\min(\absr{i},\absr{j}) - \absr{L})}{\absr{i}\absr{j}\absr{L}} \right] \frac{\covF{\cv{i},\cv{j}}}{\nhf}.
    \end{equation*}
    Finally we have the diagonal component that considers both parts of Equation~\eqref{eq:kacv-split2}
    \begin{equation*}
      \varF{\cv{i}} = \left[ \frac{(\absr{i}-\absr{L})^2}{\absr{i}^2\absr{L}} + \frac{1}{\absr{i}^2}(\absr{i}-\absr{L})\right] \frac{\varF{\cv{i}}}{\nhf} = \frac{\absr{i} - \absr{L}}{\absr{i}\absr{L}} \frac{\varF{\cv{i}}}{\nhf}.
    \end{equation*}

    \textbf{Case 3:} Now we consider the final case $i\leq K$, and $j>K$. The reverse follows from symmetry. This case itself can be broken into three subcases: (a) $i = L$, (b) $i < L$, and (c) $i > L$. First we consider case (3a). In this case, we have the splitting from Equation~\eqref{eq:splitting} and the one just derived above~\eqref{eq:kacv-split2}. Computing the covariances between these terms inolves the covariances of each summation and leads to 
    \begin{equation*}
    \covdiffij{L}{j} = \left[\frac{\absr{L}-1}{\absr{L}} \frac{\absr{j}-\absr{L}}{\absr{L}\absr{j}} - \frac{1}{\absr{L}}\frac{\absr{j} - \absr{L}}{\absr{L}\absr{j}}(\absr{L}-1)\right] \frac{\covF{\cv{L},\cv{j}}}{\nhf} = 0.
    \end{equation*}
    For case (3b) all of the samples in $\sset{i}$ are shared and we need to use a different splitting $\cvdiff{j}$ to elucidate this fact. The relevant splitting for $\cvdiff{i}$ is given by Equation~\eqref{eq:splitting}, and the necessary splitting for $\cvdiff{j}$ is
    \begin{equation}
      \cvdiff{j}(\sset{j}) = \frac{1}{\absr{L}\nhf}\sum_{k=1}^{\absr{L}\nhf}\cv{j}(\sset{L}^{(k)}) - \frac{1}{\absr{j}\nhf} \sum_{k=1}^{\absr{j}\nhf} \cv{j}(\sset{j}^{(k)}) 
     = \frac{1}{\absr{L}\nhf} A_j -  \frac{1}{\absr{j}\nhf} \left(A_j +  \sum_{k=\absr{L}\nhf+1}^{\absr{j}\nhf} \cv{j}(\sset{j}^{(k)})\right)
   \end{equation}
   where
   \begin{align*}
   A_j &= \sum_{k=1}^{\nhf}\cv{j}(\sset{}^{(k)}) + \sum_{k=\nhf+1}^{\absr{i}\nhf} \cv{j}(\sset{i}^{(k)}) + \sum_{k=\absr{i}\nhf+1}^{\absr{L}\nhf} \cv{j}(\sset{L}^{(k)}).
   \end{align*}
   This leads to
   \begin{align}
   \cvdiff{j}(\sset{j}) &= \frac{\absr{j}-\absr{L}}{\absr{j}\absr{L}\nhf}A_j - \frac{1}{\absr{j}\nhf} \sum_{k=\absr{L}\nhf+1}^{\absr{j}\nhf} \cv{j}(\sset{j}^{(k)}).
    \end{align}
    The terms that have $\sset{}$ and $\sset{i}$ are shared with~\eqref{eq:splitting} leading to
    \begin{equation*}
      \covF{\cvdiff{i},\cvdiff{j}} = \left[\frac{\absr{i}-1}{\absr{i}}\frac{\absr{j}-\absr{L}}{\absr{j}\absr{L}} - \frac{1}{\absr{i}}\frac{\absr{j}-\absr{L}}{\absr{j}\absr{L}}(\absr{i}-1)\right] \frac{\covF{\cv{i}, \cv{j}}}{\nhf} = 0
    \end{equation*}
    Finally, we consider case (3c) where $L < i \leq K$ and $j > K$. The splitting for $\cvdiff{i}$ becomes
    \begin{equation*}
    \begin{split}
    \frac{\absr{i}}{\nhf} \estm{i} &= \sum_{k=1}^{\absr{i}\nhf} \cv{i}(\sset{i}^{(k)}) 
    = \Bigg[\sum_{k=1}^{\nhf}\cv{i}(\sset{}^{(k)}) + \sum_{k=\nhf+1}^{\absr{L}\nhf}\cv{i}(\sset{L}^{(k)}) + \sum_{k=\absr{L}\nhf+1}^{\absr{i}\nhf}\cv{i}(\sset{i}^{(k)})\Bigg]
    \end{split}
    \end{equation*}
    This equivalence leads to 
    \begin{equation*}
    \cvdiff{i}(\sset{i}) = \frac{\absr{i}-1}{\absr{i}\nhf}\sum_{k=1}^{\nhf}\cv{i}(\sset{}^{(k)}) - \frac{1}{\absr{i}\nhf} \left[\sum_{k=\nhf+1}^{\absr{L}\nhf}\cv{i}(\sset{L}^{(k)}) + \sum_{k=\absr{L}\nhf+1}^{\absr{i}\nhf}\cv{i}(\sset{i}^{(k)})\right],
    \end{equation*}
    Now, using the splitting Equation~\eqref{eq:kacv-split2} for $\cvdiff{j}$ and the following variable
    \begin{align*}
    A \equiv \left[\frac{\absr{i}-1}{\absr{i}}\frac{\absr{j}-\absr{L}}{\absr{j}\absr{L}} - \frac{1}{\absr{i}}\frac{\absr{j}-\absr{L}}{\absr{j}\absr{L}}(\absr{L}-1) + \frac{1}{\absr{i}}\frac{1}{\absr{j}}(\absr{i}-\absr{L})\right],
    \end{align*}
    we obtain
    \begin{equation*}
      \covdiffij{i}{j} = A \frac{\covF{\cv{i},\cv{j}}}{\nhf} = \left[\frac{\absr{i} - \absr{L}} {\absr{i} \absr{L}} \right] \frac{\covF{\cv{i},\cv{j}}}{\nhf}.
    \end{equation*}
  To summarize, we have
  \begin{equation}
  \begin{split}
    \covdiffij{i}{j} = \frac{\covF{\cv{i},\cv{j}}}{\nhf}  \times
    \left\{
    \begin{array}{cl}
      \frac{\min(\absr{i},\absr{j}) - 1}{\min(\absr{i},\absr{j})}  & \textrm{ if } i,j \leq K \\
      \left[\frac{(\absr{i}-\absr{L})(\absr{j}-\absr{L}) + \absr{L}(\min(\absr{i},\absr{j}) - \absr{L})}{\absr{i}\absr{j}\absr{L}} \right]  & \textrm{ if } i,j > K \\
      \left[\frac{\absr{i} - \absr{L}} {\absr{i} \absr{L}} \right] & \textrm{ if } L < i \leq K, \ j > K \\
      \left[\frac{\absr{j} - \absr{L}} {\absr{j} \absr{L}} \right] & \textrm{ if }  L < j \leq K, \ i > K \\
      0 & \textrm{ otherwise }
    \end{array}
    \right. , 
    \end{split}
  \end{equation}
  for $i \neq j$. The diagonal elements are $\varF{\cvdiff{i}} = \frac{\absr{i}-1}{\absr{i}} \frac{\varF{\cv{i}}}{\nhf}$ if $i \leq K$ and $\varF{\cvdiff{i}} = \frac{\absr{i} - \absr{L}}{\absr{i}\absr{L}} \frac{\varF{\cv{i}}}{\nhf}$ otherwise. Furthermore, $\covhc$ is equal to the diagonal of this matrix. The stated result follows by using these expressions within Propositions~\ref{prop:acv_var} and~\ref{prop:ocv} as before.
\end{proof}

%% file: appmlmc.tex
 \section{Proof of Lemma~\ref{lem:var_wRDiff}}\label{sec:proof_wRDiff}
 \begin{proof}
  The proof uses the definition of the w-RDiff estimator and is based on the identification of the statistical independent difference contributions.
  First, the w-RDiff estimator is re-arranged as
  \begin{equation}
  \begin{split}
    \mlest(\vec{\cvw},\vec{\sset{}}) &= \est{}(\sset{}) + \sum_{i=1}^{\nmodels} \cvw_i \left(\est{i}(\sset{i-1}^2) - \estm{i}(\sset{i}^2) \right) \\
                                     &= - \cvw_\nmodels \estm{\nmodels}(\sset{\nmodels}^2) 
                                        + \left( \cvw_1 \est{1}(\sset{}) + \est{}(\sset{}) \right)
                                        + \sum_{i=2}^\nmodels \cvw_i \est{i}(\sset{i-1}^2) - \cvw_{i-1} \estm{i-1}(\sset{i-1}^2). 
  \end{split}
  \end{equation}
  The variance of the estimator is obtained as sum of the independent contributions (\textit{i.e.} all the covariances terms are zero). After few manipulations
  we obtain the stated result
  \begin{equation}
  \begin{split}
   \VarF{ \mlest(\vec{\cvw},\vec{\sset{}})} 
        &= \cvw_\nmodels^2 \frac{\VarF{\qoi_{\nmodels}}}{\card{\sset{\nmodels}^2}} 
        + \frac{1}{N} \left( \cvw_1^2 \VarF{\qoi_1} + \VarF{\qoi} + 2 \cvw_1 \ccoeff{1} \sqrt{ \VarF{\qoi_1} \VarF{\qoi}}\right) \\
        &+ \sum_{i=2}^\nmodels \frac{1}{\card{\sset{i-1}^2}} \left( \cvw_i^2 \VarF{\qoi_i} + \cvw_{i-1}^2 \VarF{\qoi_{i-1}} - 2 \cvw_i \cvw_{i-1} \sqrt{\VarF{\qoi_i}\VarF{\qoi_{i-1}}} \right) \\
        &= \VarF{\est{}} \left( 1 - \left( - \cvw_1^1 \rs{1}^2 - 2 \cvw_1 \ccoeff{1} \rs{1} - \cvw_\nmodels^2 \frac{\rs{\nmodels}}{\rcard{\nmodels}} 
                                           - \sum_{i=2}^{\nmodels} \frac{1}{\rcard{i-1}} \left( \cvw_i^2 \rs{i}^2 + \cvw_{i-1}^2 \rs{i-1}^2 - 2 \cvw_i \cvw_{i-1} \rs{i} \rs{i-1}
                                           \right) \right) \right),
  \end{split}
  \end{equation}

  where $\rcard{i}$ indicates the ratio between the cardinality of the set $\sset{i}^2$ and $\sset{}$.

 \end{proof}